\documentclass[reqno, 10pt]{amsart}
\usepackage{color}
\usepackage{amsmath}
\usepackage{amsfonts}
\usepackage{amssymb}
\usepackage{amsthm}
\usepackage{newlfont}
\usepackage{graphicx}

\newtheorem{thm}{Theorem}[section]
\newtheorem{cor}[thm]{Corollary}
\newtheorem{lem}[thm]{Lemma}
\newtheorem{prop}[thm]{Proposition}
\theoremstyle{definition}

\theoremstyle{remark}
\newtheorem{rem}[thm]{Remark}

\numberwithin{equation}{section}
\numberwithin{thm}{section}

\newcommand{\R}{{\mathbb{R}}}
\renewcommand{\v}[1]{\ensuremath{\mathbf{#1}}}

\newcommand{\mveca}[2]{ \langle {#1} \rangle
\frac{{#1}\cdot {#2}}{|{#1}||{#2}|} \frac{|{#1}-{#2}|}{\langle {#1}-{#2} \rangle}}

\newcommand{\mvecc}[2] { |{#1}| \frac{ ({#1}-{#2})\cdot {#2} } {|{#1}-{#2}||{#2}|}}

\newcommand{\jnab}{\langle \nabla \rangle}
\newcommand{\psa}{ \frac i 4\langle \xi \rangle \langle \eta \rangle^{-1} |\eta|
\frac{\xi\cdot(\xi-\eta)}{|\xi||\xi-\eta|}}
\newcommand{\psb}{\frac i 8 |\xi|\frac{|\eta||\xi-\eta|}{\langle \eta \rangle \langle \xi-\eta\rangle}}
\newcommand{\psc}{\frac i 8 |\xi|\frac{\eta\cdot(\xi-\eta)}{|\eta||\xi-\eta|}}
\newcommand{\ed}{\end {document}}

\newcounter{smalllist}

\newcommand{\lnr}{\langle \nabla \rangle}
\newcommand{\lsr}{\langle s \rangle^{\delta_0}}

\title[global solution for Euler-Poisson]{The Cauchy problem for the two dimensional
Euler-Poisson system}

\author[D. Li]{Dong Li}
\address{Department of Mathematics, University of British Columbia, Vancouver, BC, V6T 1Z2}%
\email{dli@math.ubc.ca}
\author[Y. Wu]{Yifei Wu}
\address{School of Mathematical Science, Beijing Normal University, Beijing, China, 100875.}
\email{yerfmath@yahoo.cn}

\begin{document}
\maketitle
\begin{abstract}
The Euler-Poisson system is a fundamental two-fluid model to
describe the dynamics of the plasma consisting of compressible
electrons and a uniform ion background. In the 3D case Guo \cite{Guo98} first constructed a global
smooth irrotational solution by using the dispersive Klein-Gordon effect. It has
been conjectured that same results should hold in the
two-dimensional case. In our recent work \cite{jlz}, we proved the existence of a family of smooth
solutions by constructing the wave operators for the 2D system.
In this work we completely settle the 2D Cauchy problem.
\end{abstract}

\section{Introduction}

The Euler-Poisson system is one of the simplest  two-fluid models used to describe the
dynamics of a plasma consisting of moving electrons and ions. In this model the heavy ions
are assumed to be immobile and uniformly distributed in space, providing only a background of
positive charge. The light electrons are modeled as a charged compressible fluid
 moving against the ionic forces. Neglecting magnetic effects,
the governing dynamics of the electron fluid is given
by the following Euler-Poisson system in $(t,x) \in [0,\infty) \times \mathbb R^d$,
\begin{align} \label{eq_EP_1}
\begin{cases}
\partial_t n + \nabla \cdot (n \v u) =0, \\
m_e n (\partial_t \v u + (\v u\cdot \nabla) \v u ) + \nabla p(n) = e n \nabla \phi, \\
\Delta \phi = 4\pi e (n-n_0).
\end{cases}
\end{align}

Here $n=n(t,x)$ and $\v u=\v u(t,x)$ denote the density and average velocities of the electrons respectively.
The symbol $e$ and $m_e$ denote the unit charge and
mass of electrons. The pressure term $p(n)$ is assumed to obey the polytropic $\gamma$-law, i.e.
\begin{align} \label{eq_EP_2}
p(n) = A n^\gamma,
\end{align}
where $A$ is the entropy constant and $\gamma \ge 1$ is called the adiabatic index.
The term $en\nabla \phi=(-ne)\cdot(-\nabla \phi)$  quantifies the
electric force acting on the electron fluid
by the positive ion background. Note that the electrons carry negative charge $-ne$.
We assume at the equilibrium the density of ions and electrons are both a constant denoted by
$n_0$. To ensure charge neutrality it is natural to impose the condition
\begin{align*}
\int_{\mathbb R^d} (n-n_0) dx =0.
\end{align*}
The boundary condition for the electric potential $\phi$ is a decaying condition at infinity,
i.e.
\begin{align} \label{eq_EP_3}
\lim_{|x| \to \infty} \phi(t,x) = 0.
\end{align}

The first and second equations in \eqref{eq_EP_1} represent mass conservation and momentum balance of
the electron fluid respectively. The third equation in \eqref{eq_EP_1} is the usual Gauss law in electrostatics.
It computes the electric potential self-consistently through the charge distribution $(n_0e -ne)$.
The Euler-Poisson system is one of the simplest
two-fluid model in the sense that the ions are treated as uniformly distributed sources in space and they
appear only as a constant
$n_0$ in the Poisson equation. This is a very physical approximation since $m_{ion}\gg m_{e}$
and the heavy ions move much more slowly than the light electrons.

Throughout the rest of this paper, we shall consider an irrotational flow
\begin{align}
\nabla \times {\mathbf u}=0 \label{e_irrot}
\end{align}
 which
is preserved in time. For flows with nonzero curl the magnetic field is no longer negligible and it is more physical
to consider the full Euler-Maxwell system.

We are interested in constructing smooth global solution around the equilibrium $(n,\v u)\equiv (n_0,0)$. To do this
we first transform the system \eqref{eq_EP_1} in terms of certain perturbed
variables. For simplicity set all physical constants $e$, $m_e$, $4\pi$ and $A$ to be
one. To simplify the presentation, we also set $\gamma=3$ although other cases of $\gamma$ can be easily
treated as well. Define the rescaled functions
\begin{align*}
u(t,x) & = \frac{n(t/c_0, x) -n_0} {n_0}, \\
\v v(t,x) & = \frac 1 {c_0} \v u (t/c_0, x), \\
\psi(t,x) & = 3 \phi (t/c_0, x),
\end{align*}
where the sound speed is $c_0 = \sqrt 3 n_0$. For convenience we set $n_0 = 1/3$ so that the characteristic
wave speed is unity. The Euler-Poisson system \eqref{eq_EP_1} in new variables takes the form
\begin{align} \label{e927_450}
\begin{cases}
\partial_t u + \nabla \cdot \v v + \nabla \cdot( u \v v) =0, \\
\partial_t \v v + \nabla u + \nabla \left( \frac 12 u^2 + \frac 12 |\v v|^2 \right) =  \nabla \psi, \\
\Delta \psi = u.
\end{cases}
\end{align}
Taking one more time derivative and using \eqref{e_irrot} then transforms \eqref{e927_450}
into the following quasi-linear Klein-Gordon system:
\begin{align} \label{e927_450abc}
\begin{cases}
(\square +1) u= \Delta\left( \frac 12 u^2+ \frac 12 |\v v|^2 \right) -\partial_t \nabla \cdot (u\v v), \\
(\square+1) \v v = -\partial_t \nabla \left( \frac 12 u^2+\frac 12 |\v v|^2 \right)
+ (1-\Delta^{-1}) \nabla \nabla \cdot (u \v v).
\end{cases}
\end{align}

For the above system, in the 3D case Guo \cite{Guo98} first constructed a global
smooth irrotational solution by using dispersive Klein-Gordon effect and adapting Shatah's normal
form method. It has
been conjectured that same results should hold in the
two-dimensional case. In our recent work \cite{jlz}, we proved the existence of a family of smooth
solutions by constructing the wave operators for the 2D system. The 2D problem with radial data
was studied in \cite{Jang_radial}.  Note that for radial data\footnote{The vector function $\mathbf v$ is radial if it
is the gradient of a scalar radial function}, one has
$$ \Delta^{-1} \nabla \nabla \cdot ( u {\v v}) = u \v v$$
and the  result follows easily from \cite{Oz96}.
In this work we completely settle the 2D Cauchy problem for general non-radial data.
The approach we take in this paper is inspired
from a new set-up of normal form transformation developed by Gustafson, Nakanishi, Tsai \cite{GNT06} and also
Germain, Masmoudi and Shatah \cite{GMS1,GMS2,GMS3}.
Roughly speaking (and over-simplifying quite a bit), the philosophy of the normal form method is that \textit{one should integrate parts
whenever you can in either (frequency) space
or time}. The part where one cannot integrate by parts is called the set of space-time resonances which
can often be controlled by some finer analysis
provided the set is not so large or satisfies some frequency separation properties. The implementation of such ideas
is often challenging and depends heavily on the problem under study. In fact the heart of the whole analysis is to
choose appropriate functional spaces utilizing the fine structure of the equations.
The main obstructions in the 2D Euler-Poisson system are slow(non-integrable) $\langle t \rangle^{-1}$ dispersion,
quasilinearity and nonlocality caused by
the Riesz transform. Nevertheless we overcome all such difficulties in this paper.
After our work is completed, a similar result requiring at least $30+$ derivatives is obtained in \cite{IP11}.
To put things into
perspective, we review below some related literature as well as some technical developments
on this problem.

The main difficulty in constructing time-global smooth solutions for the Euler-Poisson system
comes from the fact that the Euler-Poisson system is a hyperbolic conservation law with zero dissipation
for which no general theory is available.
The "Euler"-part of the Euler-Poisson system is the well-known compressible Euler equations.
 Indeed in \eqref{eq_EP_1}
if the electric field term $\nabla \phi$ is dropped, one recovers the usual Euler equations
for compressible fluids. In \cite{Si85}, Sideris considered the 3D compressible Euler equation
for a classical polytropic ideal gas with adiabatic index $\gamma>1$. For a class of
initial data which coincide with a constant state outside a ball, he proved that the lifespan
of the corresponding $C^1$ solution must be finite. In \cite{Ra89} Rammaha extended this
result to the 2D case. For the Euler-Poisson system, Guo and Tahvildar-Zadeh \cite{GuoTZ99}
established a "Siderian" blowup result for spherically symmetric initial data.
Recently Chae and Tadmor \cite{CT08} proved finite-time blow-up for $C^1$ solutions of
a class of pressureless attractive Euler-Poisson equations in $\mathbb R^n$, $n\ge 1$.
These negative results showed the abundance of shock waves for large solutions.

The "Poisson"-part of the Euler-Poisson system has a stabilizing effect which makes the whole
analysis of \eqref{eq_EP_1} quite different from the pure compressible Euler equations.
This is best understood in analyzing small irrotational perturbations of the equilibrium state
$n\equiv n_0$, $\v u\equiv 0$. For the 3D compressible Euler equation with irrotational initial data
$(n_{\epsilon}(0),\v u_{\epsilon}(0)) = (\epsilon \rho_0 + n_0, \epsilon \v v_0)$, where
$\rho_0 \in \mathcal S(\mathbb R^3)$,
$\v v_0\in \mathcal S(\mathbb R^3)^3$ are fixed functions  ($\epsilon$ sufficiently small),  Sideris
\cite{Si91} proved that the lifespan of the classical solution $T_\epsilon > \exp( C/\epsilon)$.
For the upper bound it follows from his previous paper \cite{Si85} that $T_\epsilon < \exp(C/\epsilon^2)$.
Sharper results were obtained by Godin \cite{Godin05} in which he showed for radial initial
data as a smooth compact $\epsilon$-perturbation of the constant state, the precise asymptotic of
the lifespan $T_\epsilon$ is exponential in the sense
\begin{align*}
\lim_{\epsilon \to 0+} \epsilon \log T_\epsilon = T^*,
\end{align*}
where $T^*$ is a constant. All these results rely crucially on the observation that after some simple reductions,
the compressible Euler equation in rescaled variables is given by a vectorial nonlinear wave equation with pure
quadratic nonlinearities. The linear part of the wave equation decays at most at the speed $t^{-(d-1)/2}$ which
in 3D is not integrable. Unless the nonlinearity has some additional nice structure such as
the null condition \cite{C86, K86}, one cannot in general expect global existence of small solutions.
On the other hand, the situation for the Euler-Poisson system \eqref{eq_EP_1}
is quite different due to the additional Poisson coupling term. As was already explained before, the
Euler-Poisson system \eqref{eq_EP_1} expressed in rescaled variables is given by
the  quasi-linear Klein-Gordon system \eqref{e927_450abc}
 for which the linear solutions have an enhanced decay of $(1+t)^{-d/2}$.
 This is in sharp contrast with
the pure Euler case for which the decay is only $t^{-(d-1)/2}$. Note that in $d=3$, $(1+t)^{-d/2}=(1+t)^{-3/2}$
which is integrable in $t$. In a seminal paper \cite{Guo98}, by exploiting the crucial decay property of the
Klein-Gordon flow in 3D, Guo \cite{Guo98} modified Shatah's normal form method \cite{Sh85} and
constructed a smooth irrotational global solution to \eqref{eq_EP_1} around the
equilibrium state $(n_0,0)$ for which the perturbations decay at a rate $C_p \cdot (1+t)^{-p}$ for any
$1<p<3/2$ (here $C_p$ denotes a constant depending on the parameter $p$). Note in particular that
the sharp decay $t^{-3/2}$ is marginally missed here due to a technical complication caused by
the nonlocal Riesz operator in the nonlinearity.

Construction of smooth global solutions to \eqref{eq_EP_1} in the two-dimensional case was
open since Guo's work.
The first obstacle comes from slow dispersion
 since the linear solution to the Klein-Gordon system in
$d=2$ decays only at $(1+t)^{-1}$ which is not integrable, in particular making the
strategy of \cite{Guo98} difficult to apply. The other main technical difficulty comes
from the nonlocal nonlinearity in \eqref{e927_450abc} which involves a Riesz-type singular operator.
For general scalar quasi-linear Klein-Gordon equations
in 3D with quadratic type nonlinearities, global small smooth solutions were first constructed independently by
Klainerman \cite{K85} using the invariant vector field method and Shatah \cite{Sh85}
using a normal form method. Even in 3D there are
essential technical difficulties in employing Klainerman's invariant vector field method due to the Riesz
type nonlocal term in \eqref{e927_450abc}.
The Klainerman invariant vector fields consist of infinitesimal generators which commute well with
the linear operator $\partial_{tt} -\Delta +1$. The most problematic part comes from the Lorentz boost
$\Omega_{0j} = t \partial_{x_j} + x_j \partial_t$. While the first part $t \partial_{x_j}$ commutes
naturally with the Riesz operator $R_{ij}=(-\Delta)^{-1} \partial_{x_i} \partial_{x_j}$,
the second part $x_j \partial_t $ interacts rather badly with $R_{ij}$, producing a commutator
which scales as
\begin{align*}
[x_j \partial_t, R_{ij}] \sim \partial_t |\nabla|^{-1}.
\end{align*}
After repeated commutation of these operators one obtains in general terms of the form
$|\nabla|^{-N}$ which makes the low frequency part of the solution out of control.
It is for this reason that in 3D case Guo \cite{Guo98} adopted Shatah's method of normal
form in $L^p$ ($p>1$) setting for which the Riesz term $R_{ij}$ causes no trouble.
We turn now to the 2D Klein-Gordon equations with pure quadratic nonlinearities. In this case, direct
applications of either Klainerman's invariant vector field method or Shatah's normal form method
are not possible since the linear solutions only decay at a speed of $(1+t)^{-1}$ which is not integrable
and makes the quadratic nonlinearity quite resonant. In \cite{ST93}, Simon and Taflin constructed wave
operators for the 2D semilinear Klein-Gordon system with quadratic nonlinearities.
In \cite{Oz96}, Ozawa, Tsutaya and Tsutsumi considered the Cauchy problem and
constructed smooth global solutions by first transforming the quadratic nonlinearity into a cubic one
using Shatah's normal form method and then applying Klainerman's invariant vector field method to obtain
decay of intermediate norms.
Due to the nonlocal complication with the Lorentz boost
which we explained earlier, this approach
seems difficult to apply in the 2D Euler-Poisson system.

As was already mentioned, the purpose of this work is to settle the Cauchy problem for
\eqref{eq_EP_1} in the two-dimensional case. Before we state our main results,
we need to make some further simplifications.
Since $\v v$ is irrotational, we can write $\v v = \nabla \phi_1$ and obtain from \eqref{e927_450}
(here $\langle \nabla \rangle =\sqrt{1-\Delta}$, see \eqref{not_jnab}):
\begin{align} \label{e9927_732a}
\begin{cases}
\partial_t u + \Delta \phi_1 + \nabla \cdot ( u \nabla \phi_1) =0, \\
\partial_t \phi_1 + |\nabla|^{-2} \langle \nabla \rangle^2 u + \frac12 (u^2 + |\nabla \phi_1|^2) =0.
\end{cases}
\end{align}

We can diagonalize the system \eqref{e9927_732a} by introducing the complex scalar function
\begin{align}
h(t)&= \frac{\langle \nabla \rangle} {|\nabla|} u - i |\nabla| \phi_1 \notag \\
& =   \frac{\langle \nabla \rangle} {|\nabla|} u + i \frac{\nabla}{|\nabla|} \cdot \v v .
\label{e92_537a}
\end{align}
Note that since $\v v$ is irrotational, we have
\begin{align}
\v v =- \frac{\nabla}{|\nabla|} \text{Im}(h).  \label{eq_vh}
\end{align}
By \eqref{e927_450}, we have
\begin{align}
h(t)& = e^{it\langle \nabla \rangle} h_0 +\int_0^t e^{i(t-s)\langle \nabla \rangle} \Bigl(
-  \frac{\langle \nabla \rangle \nabla}{|\nabla|} \cdot (u \v v) \notag \\
& \qquad + \frac i2 |\nabla| ( u^2 +|\v v|^2) \Bigr) ds, \label{e92_537b}
\end{align}
where $h_0$ is the initial data given by
\begin{align*}
h_0 =  \frac{\langle \nabla \rangle} {|\nabla|} u_0 + i \frac{\nabla}{|\nabla|} \cdot \v v_0.
\end{align*}
Here $u_0$ is the initial density (perturbation) and $\v v_0$ is the initial velocity.

For $T\ge 0$, $\delta>0$, $N\ge 8$, $N^{\prime}=N-\frac 32$, we introduce
the norms
\begin{align}
\|h\|_{\tilde X_T}:= & \| \langle t \rangle |\nabla|^{\delta}
\langle \nabla \rangle h(t) \|_{L_{t,x}^{\infty}([0,T])}
+\|\langle t \rangle^{1-2\delta} \langle \nabla \rangle  h(t) \|_{L_t^{\infty} L_x^{\frac 1 {\delta}}
([0,T])}  \notag \\
& \qquad\; + \| x (1-\Delta)e^{-it \langle \nabla \rangle} h(t) \|_{L_t^{\infty} L_x^{2+\delta}([0,T])}, \notag
\end{align}
and
\begin{align}
\| h \|_{ X_T} := \|h\|_{\tilde X_T} + \| h(t) \|_{C_t^{0} H^{N^\prime}([0,T])} +
\| \langle t \rangle^{-\delta} h(t) \|_{C_t^{0} H^N([0,T])}.
\notag
\end{align}
Here for simplicity we have suppressed the notational dependence of the $X_T$ norm on $\delta$. We will use the notation
$X_{\infty}$ (resp. $\tilde X_{\infty}$) when the norms are evaluated on the time interval $[0,\infty)$.


Our result is expressed in the following

\begin{thm}[Smooth global solutions for the Cauchy problem] \label{thm_main}
There exists an absolute constant $\delta_*>0$ sufficiently small such that the following hold:

For any $0<\delta<\delta_*$, there exists $\epsilon>0$ sufficiently small such that if the initial
data $h_0$ satisfies $\| e^{it\langle \nabla \rangle} h_0 \|_{X_{\infty}}\le \epsilon$, then
there exists a unique smooth global solution
to the 2D Euler-Poisson system \eqref{e92_537a}--\eqref{e92_537b} satisfying
$\| h\|_{ X_\infty} \le const \cdot \epsilon$. Moreover the solution scatters
in the energy space $H^{N^\prime}$.
\end{thm}

\begin{rem}
A simple inspection of our proof shows that it suffices to take
$\delta_*=\frac 1{500}$. We do not make much effort to lower down the regularity
assumption ($N\ge 8$) on the initial data although the result here is already better
than many existing methods. The main point here is to construct a smooth and global in time
classical solution.
\end{rem}

To prove Theorem \ref{thm_main}, we shall establish an a priori estimate of the form
\begin{align}
\| h\|_{ X_t} \lesssim \|e^{i\tau\langle \nabla \rangle } h_0\|_{ X_{\infty}} +
\| h\|_{ X_t}^2 + \| h\|_{ X_t}^3+ \| h\|_{ X_t}^4,
\label{J22_70a}
\end{align}
where the implied constant depends only on the parameter $\delta$ and $N$. The function can be shown
to be continuous in $t$ (see Step 2 below). By a standard continuity argument, if
$ \|e^{i\tau\langle \nabla \rangle } h_0\|_{ X_{\infty}}$ is sufficiently small, then $\|h\|_{X_t}$ remains bounded
for all $t\ge 0$ which yields global wellposedness easily. Therefore our main work is to show \eqref{J22_70a}.
We sketch its proof in the following steps.

\texttt{Step 1:} Preliminary transformations and normal form.

In this step, we introduce $f(t)=e^{-it\langle \nabla \rangle} h(t)$ and rewrite \eqref{e92_537b}
as
\begin{align}
\hat f(t,\xi) = \widehat{h_0}(\xi) + \int_0^t \int e^{-is \phi_0(\xi,\eta)}
\langle \xi \rangle \frac{\xi}{|\xi|} \widehat{\mathcal R f}(s,\xi-\eta) \widehat{\mathcal Rf}(s,\eta) d\eta ds, \label{J22_49}
\end{align}
where $\mathcal R$ is some Riesz-type operator and
\begin{align*}
\phi_0(\xi,\eta) =\langle \xi \rangle \pm \langle \xi-\eta \rangle \pm \langle \eta \rangle.
\end{align*}

By using the fact that the Klein-Gordon phase $\phi_0(\xi,\eta)$ never vanishes, we perform a normal
form transformation and integrate by parts in the time variable $s$. After some simplifications, we arrived
at an equation of the form
\begin{align*}
\hat f(t,\xi)="\text{initial data}"+"\text{quadratic boundary terms}" + \widehat{f_{\text{cubic}}}(t,\xi),
\end{align*}
where $f_{\text{cubic}}$ is cubic in $h$ and has the form $f_{\text{cubic}}= \mathcal R f_3$ with
\begin{align}
\widehat{f_3}(t,\xi) &= \int_0^t \int e^{-is \phi(\xi,\eta,\sigma)} \frac{\langle \xi \rangle \cdot
\langle \eta \rangle}{\phi_0(\xi,\eta)}
\cdot \frac{\eta}{|\eta|} \widehat{\mathcal Rf}(s,\xi-\eta) & \notag\\
&\qquad \cdot \widehat{\mathcal Rf}(s,\eta-\sigma) \cdot \widehat{\mathcal Rf}(s,\sigma)
d\sigma d\eta ds. \label{J22_50}
\end{align}
Here
\begin{align*}
\phi(\xi,\eta,\sigma)=\langle \xi \rangle \pm \langle \xi-\eta \rangle
\pm \langle \eta-\sigma \rangle \pm \langle \sigma\rangle.
\end{align*}

The estimates of the initial data part and the boundary terms are given in Section 5.

\texttt{Step 2}: Local theory, continuity of the $X$-norm along the flow and $H^{N^{\prime}}$-estimate.

At first we carry out the (standard) $H^N$-energy estimate and obtain an estimate of the form
\begin{align*}
\frac d {dt} \Bigl( \| h(t) \|_{H^N}^2 \Bigr)
& \lesssim ( \| u(t)\|_{\infty} + \| \nabla u(t) \|_{\infty} + \| \nabla \v v(t) \|_{\infty} )
\cdot \| h(t) \|_{H^N}^2.
\end{align*}
The subtle point here is that $\| \v v(t)\|_{\infty}$ does not appear in the energy estimate.

Due to the slow ($1/t$) decay in 2D, we need to have a slight $\langle t \rangle^{\delta}$ growth
of the norm $\|h(t)\|_{H^N}$ in order  to close the estimates. Note that
$u= \frac{|\nabla|}{\langle \nabla \rangle}
\text{Re}(h)$ and $\v v = -\frac {\nabla}{|\nabla|} \text{Im}(h)$, hence
\begin{align*}
\| u(t) \|_{\infty} +\| \nabla u(t) \|_{\infty} + \| \nabla \v v(t) \|_{\infty}
\lesssim \| |\nabla|^{\delta} \langle \nabla \rangle h(t) \|_{\infty}.
\end{align*}

It remains to prove the sharp $1/t$ decay of the $L^{\infty}$-norm
$\| |\nabla|^{\delta} \langle \nabla \rangle h(t) \|_{\infty}$. For this and later estimates,
we need to show the time-continuity of the norm $\| x (1-\Delta) e^{-it\langle \nabla \rangle} h(t) \|_{2+\delta}$.
This is done in Section 4. The main idea there is a bootstrap estimate exploiting the finite speed
propagation property of the Klein-Gordon flow. In the last part of Section 4, we complete the
$H^{N^\prime}$ estimate of $h$. To lower the regularity assumption, we first introduce frequency cut-offs
$\chi_{\ge \langle s \rangle^{\delta_0}}$ and $\chi_{<\langle s \rangle^{\delta_0}}$ in \eqref{J22_49}.
For the high frequency part, we estimate it using energy smoothing (recall $N^{\prime}=N-\frac 32$) and
dispersive decay. For the low frequency piece, we use the normal form and obtain a cubic nonlinearity
localized to low frequencies. The $H^{N^{\prime}}$ estimate is used in controlling some  boundary terms
in Section 5.

\texttt{Step 3:} Reduction to low frequencies and the $(2+\delta)$-trick.

This is an important step in controlling the $X$-norm of $h$. We use a multiscale argument and introduce
the parameter $\delta_0=20\delta$.  We then decompose the cubic nonlinear term $f_{\text{cubic}}=\mathcal R f_3$
(see \eqref{J22_50}) into two pieces:
\begin{align*}
\widehat{f_3}(t,\xi) & = \int_0^t \int e^{-is \phi(\xi,\eta,\sigma)}
\cdot \frac{\langle \xi \rangle \cdot \langle \eta \rangle}{\phi_0(\xi,\eta)}
\cdot \frac{\eta}{|\eta|} \notag \\
&\qquad \cdot (m_{\text{low}}(\xi,\eta,\sigma,s) + m_{high}(\xi,\eta,\sigma,s)) \cdot
\widehat{\mathcal R f}(s,\xi-\eta) \notag \\
& \qquad \cdot \widehat{\mathcal R f}(s,\eta-\sigma) \cdot \widehat{\mathcal R f}(s,\sigma) d\sigma d\eta ds \notag \\
& =: \widehat{f_3^{(1)}} + \widehat{f_3^{(2)}},
\end{align*}
where
\begin{align*}
m_{\text{low}}(\xi,\eta,\sigma,s)&= \chi_{|\xi-\eta|\le \langle s \rangle^{\delta_0}}
\cdot \chi_{|\eta-\sigma| \le \langle s \rangle^{\delta_0}} \cdot \chi_{|\sigma| \le \langle s \rangle^{\delta_0}}, \notag\\
m_{high}(\xi,\eta,\sigma,s) &= 1- m_{\text{low}}(\xi,\eta,\sigma,s).
\end{align*}

We first show that the high frequency piece has good decay properties, namely
\begin{align}
\| e^{i\tau \langle \nabla \rangle} \mathcal R f_3^{(2)} (\tau)
\|_{\tilde X_t} \lesssim \| h \|_{X_t}^3. \label{J22_a70}
\end{align}

Thanks to the  frequency cut-off $m_{high}$, we must have either $|\xi-\eta|\gtrsim \langle s \rangle^{\delta_0}$,
$|\eta-\sigma|\gtrsim \langle s \rangle^{\delta_0}$, or $|\sigma| \gtrsim \langle s \rangle^{\delta_0}$.  This
frequency localization coupled with the energy norm and dispersive effects then produce strong decay estimates
for the $\tilde X_t$-norm of $e^{i\tau \langle \nabla \rangle} \mathcal R f_3^{(2)} (\tau)$.  By a delicate analysis
we are able to prove \eqref{J22_a70} under the weak assumption that $N\ge 8$. We emphasize that this is the main place
where the high derivative assumption is needed.

To control the $X$-norm of the low frequency piece, we must estimate several quantities including $\| |\nabla|^{\delta} \langle \nabla \rangle
e^{i\tau \langle \nabla \rangle} \mathcal R f_3^{(1)}(\tau) \|_{\infty}$,
$\|\langle \nabla \rangle e^{i\tau \langle \nabla \rangle} \mathcal R f_3^{(1)} (\tau) \|_{\frac 1 {\delta}}$,
 and $\| x (1-\Delta) \mathcal R f_3^{(1)} (\tau) \|_{2+\delta}$.  To do this
we show that all the above norms can be bounded by the $L^{2-}$ norm of some weighted integral produced from $f_3$.
More precisely, we show that
\begin{align}
\| e^{i\tau \langle \nabla \rangle} \mathcal R f_3^{(1)} (\tau) \|_{\tilde X_t}
\lesssim \| f_{\text{low}}(\tau) \|_{L_{\tau}^{\infty} L_x^{2-\frac{\delta}{100}} ([0,t])}
+ \| h\|_{X_t}^3, \label{J22_a71}
\end{align}
where
\begin{align}
f_{\text{low}}(t) &= \int_0^t \int e^{-is \phi} \cdot \frac{ s \partial_{\xi} \phi} {\phi_0(\xi,\eta)}
\cdot \langle \xi \rangle^{4+2\delta} \cdot \langle \eta \rangle
\cdot \frac{\eta}{|\eta|} \cdot m_{\text{low}}(\xi,\eta,\sigma,s)
 \notag \\
& \qquad  \widehat{\mathcal R f} (s,\xi-\eta) \cdot \widehat{\mathcal R f}(s,\eta-\sigma)
\cdot \widehat{\mathcal R f}(s,\sigma) d\sigma d\eta ds. \label{J22_a72}
\end{align}

We stress that the choice of the norm $\| x (1-\Delta) e^{-it\langle \nabla \rangle} h(t) \|_{2+\delta}$
($2+\delta$ trick) comes from this part of analysis. In particular, when bounding the quantity
$\| x \mathcal R f_3^{(1)}(\tau) \|_{2+\delta}$, we have to control the commutator
\begin{align*}
\|  [x,\mathcal R] f_3^{(1)} (\tau) \|_{2+\delta} \sim \| |\nabla|^{-1} f_3^{(1)}(\tau) \|_{2+\delta}.
\end{align*}
This latter quantity can be bounded by $\| f_{\text{low}}(\tau) \|_{2-\frac{\delta}{100}}$ thanks to the assumption
$\delta>0$.

\texttt{Step 4:} Control of the low frequency piece. The goal is to prove the bound
\begin{align}
\| f_{\text{low}}(\tau) \|_{L_{\tau}^{\infty} L_x^{2-\frac{\delta}{100}}([0,t])}
\lesssim \| h \|_{X_t}^3 + \| h\|_{X_t}^4. \label{J24_1}
\end{align}

The main difficulty in establishing this bound is the slow ($1/\langle s \rangle$) decay in \eqref{J22_a72}.
To see this point, we can perform a rough estimate as follows: the integral in \eqref{J22_a72} can be
written as (see \eqref{T_multiply})
\begin{align}
f_{\text{low}}(t) = \int_0^t s e^{-is\langle \nabla \rangle}
T_{\frac{\partial_{\xi} \phi}{\phi_0(\xi,\eta)} \langle \xi \rangle^{4+\delta}}
\Bigl( P_{\lesssim \langle s \rangle^{\delta_0}} \mathcal R h, \mathcal R
\bigl( P_{\lesssim \langle s \rangle^{\delta_0}} \mathcal R h \cdot P_{\lesssim \langle s \rangle^{\delta_0}}
\mathcal R h\bigr)\Bigr) ds. \notag
\end{align}

Ignoring the linear flow ($e^{-is\langle \nabla \rangle}$) and issues with the multipliers for the moment,
one has
\begin{align}
\| f_{\text{low}}(t) \|_{2-\frac{\delta}{100}}
& \lesssim \int_0^t \langle s \rangle \cdot \| h(s)\|_{2+} \| h(s)\|^2_{\infty-} ds \notag \\
& \lesssim \int_0^t \langle s \rangle^{1-2(1-O(\delta))} ds \cdot \| h\|_{X_t}^3 \notag \\
& \lesssim \int_0^t \langle s \rangle^{-1+O(\delta)} ds \cdot \| h\|_{X_t}^3. \label{J24_2}
\end{align}

Clearly this shows that the decay in $s$ is not enough to make the above time integral converge. To resolve
this difficulty we have to appeal to the specific form of the phase function $\phi=\phi(\xi,\eta,\sigma)$
in \eqref{J22_a72} and exploit some subtle cancelations in various cases. The main goal is to obtain
a strong decay $\langle s \rangle^{-1-\epsilon+O(\delta)}$ with $\epsilon \gg O(\delta)$ in \eqref{J24_2}.
For this we shall use some new ideas and devices which is discussed below.

$\bullet$ \textbf{Hidden derivatives}.  The first observation is that for phases of the form
$\phi(\xi,\eta,\sigma)= \langle \xi \rangle - \langle \xi -\eta\rangle \pm \langle \eta-\sigma \rangle
\pm \langle \sigma \rangle$, we have
\begin{align}
\partial_{\xi} \phi = \frac{\xi}{\langle \xi \rangle} - \frac{\xi-\eta}{\langle \xi -\eta \rangle}
= Q(\xi,\eta) \eta, \label{J26_tmp9200a}
\end{align}
where $Q$ is smooth in $(\xi,\eta)$. For $|\eta| \lesssim \langle s \rangle^{-C\delta_0}$, the factor
$\eta$ in \eqref{J26_tmp9200a} corresponds to a derivative and produces an extra decay $\langle s \rangle^{-C\delta_0}$
which will be enough to make the time integral in \eqref{J24_2} converge. Similarly for the phases
$\phi(\xi,\eta,\sigma)= \langle \xi \rangle +\langle \xi -\eta\rangle \pm \langle \eta-\sigma \rangle
\pm \langle \sigma \rangle$, the factor $\partial_{\xi} \phi$ will also produce an extra decay $\langle s \rangle^{-C\delta_0}$
in the low frequency regime $|\xi| \lesssim \langle s \rangle^{-C\delta_0}$, $|\eta| \lesssim \langle s \rangle^{-C\delta_0}$.

$\bullet$ \textbf{Normal form and the $\eta/|\eta|$ problem}. Consider phases of the form
$\phi(\xi,\eta,\sigma)= \langle \xi \rangle + \langle \xi-\eta \rangle + \langle \eta -\sigma \rangle
\pm \langle \sigma \rangle$. They have the property
\begin{align*}
\phi(\xi,\sigma,\sigma) \gtrsim \frac 1 { \langle \xi \rangle + \langle \xi-\eta\rangle + \langle \eta-\sigma\rangle
+ \langle \sigma \rangle}.
\end{align*}
By using this fact we can integrate by parts in the variable $s$ in \eqref{J22_a72}. Dropping boundary terms, we arrive
at an expression of the form
\begin{align*}
f_{\text{low}}(t) \sim & \int_0^t \int e^{-is \phi}
\cdot \frac{s \partial_{\xi} \phi}{ \phi_0(\xi,\eta)}
\cdot \frac{\langle \xi \rangle^{4+2\delta}}{\phi(\xi,\eta,\sigma)}
\cdot \langle \eta \rangle \cdot \frac{\eta}{|\eta|} \notag \\
& \qquad m_{\text{low}}(\xi,\eta,\sigma,s)
\cdot \partial_s (\widehat{\mathcal R f}(s,\xi-\eta)) \cdot \widehat{\mathcal R f}(s,\eta-\sigma)
\cdot \widehat{\mathcal R f}(s,\sigma) d\sigma d\eta ds \notag \\
& \qquad + \text{similar terms}.
\end{align*}
Note that by \eqref{J22_49} $\partial_{s}(\widehat{\mathcal R f}) \sim O((\mathcal R f)^2)$ which is quadratic in $f$. By this fact
one may hope to get $\langle s \rangle^{-2+O(\delta)}$ decay in \eqref{J24_2}. However this argument is only
correct in the regime $|\eta| \gtrsim \langle s \rangle^{-\delta_0}$. In the low frequency regime
$|\eta| \lesssim \langle s \rangle^{-\delta_0}$, the symbol $\frac 1 {\phi(\xi,\eta,\sigma)} \cdot \frac{\eta}
{|\eta|}$ is no longer smooth and one has to deal with it separately.

$\bullet$ \textbf{Partial normal form transform}.   To solve the $\eta/|\eta|$ problem, we will integrate by parts
using only \emph{part of the phase} to which we refer as \emph{partial normal form transform}. Consider for example
the phase $\phi(\xi,\eta,\sigma)= \langle \xi \rangle + \langle \xi-\eta \rangle + \langle \eta -\sigma\rangle
- \langle \sigma \rangle$. We use the identity
\begin{align*}
e^{-is (\langle \xi \rangle + \langle \xi-\eta \rangle)} = \frac{i}{\langle \xi \rangle + \langle \xi-\eta \rangle}
\frac {\partial} {\partial s} \Bigl( e^{-is(\langle \xi \rangle + \langle \xi-\eta \rangle)} \Bigr)
\end{align*}
to do integration by parts in $s$. When the derivative $\partial_s$ hits the term
$e^{-is(\langle \eta -\sigma\rangle -\langle \sigma\rangle)}$,
we obtain a factor $\langle \eta -\sigma \rangle -\langle \sigma \rangle \approx Q(\eta,\sigma) \eta$ which gains
extra decay $\langle s \rangle^{-C\delta_0}$. When the derivative hits the other terms we obtain a quintic nonlinearity.
Note that in this case all symbols are separable in the sense that they can be written as
\begin{align*}
\tilde m(\xi,\eta,\sigma) = a(\xi,\eta) b(\eta,\sigma)
\end{align*}
for some functions $a$ and $b$. The Riesz factor $\eta/|\eta|$ then causes no problem since we can deal with
the multipliers corresponding to $(\xi,\eta)$ and $(\eta,\sigma)$ separately.

$\bullet$ \textbf{Transformation of phase derivatives and frequency separation}. Consider for example the phase
$\phi(\xi,\eta,\sigma)= \langle \xi \rangle +
\langle \xi-\eta \rangle -\langle \eta-\sigma \rangle -\langle \sigma \rangle$. By Lemma \ref{lem_phase},
we can write for some smooth $Q_1$, $Q_2$,
\begin{align*}
\partial_{\xi} \phi= Q_1(\xi,\eta,\sigma) \partial_{\eta} \phi + Q_2(\xi,\eta,\sigma) \partial_{\sigma} \phi
\end{align*}
and
\begin{align*}
is e^{is\phi} \partial_{\xi} \phi= Q_1(\xi,\eta,\sigma) \partial_{\eta} (e^{is\phi})
+ Q_2(\xi,\eta,\sigma) \partial_{\sigma} (e^{is\phi}).
\end{align*}

Consequently one can integrate by parts in $\eta$ and $\sigma$ respectively which boosts the decay
in $s$ to $\langle s \rangle^{-2+O(\delta)}$.  Note that there is still a subtle issue when we perform
the above argument and integrate by parts in $\eta$.  Namely the $\partial_{\eta}$ derivative may
hit the Riesz term $\eta/|\eta|$ and produces an operator $|\nabla|^{-1}$ which is hard to control
for $|\eta| \lesssim \langle s \rangle^{-\delta_0}$. To solve this problem we have to do a multi-scale
partition of the $(\xi,\eta,\sigma)$-phase space and discuss several subcases (cf. Subcase 3a to
3d in Case 3). In particular for the low frequency regime $|\eta| \lesssim \langle s \rangle^{-\delta_0}$,
we have to discuss several situations and use the hidden derivatives, partial normal form together with several other tricks
to treat these cases (see in particular Subcase 3a to 3c in Case 3). This part of the analysis is quite
involved and uses the nonlinear structure in an essential way.

The above ideas together with some further delicate analysis completes the proof of Theorem \ref{thm_main}.
The rest of this paper is organized as follows. In Section 2 we gather some preliminary linear estimates.
In Section 3 we perform some preliminary transformations and decompose the solution into three parts:
the initial data, the boundary term $g$ and the cubic interaction term $f_{\text{cubic}}$.
In Section 4 we establish local theory, prove continuity of the $X$-norm along the flow and give the
$H^{N^{\prime}}$ estimate
of $h$.
Section 5 is devoted to the estimate of the boundary terms $g$ arising from the normal form transformation.
In Section 6 we control the
high frequency part of cubic interactions. In Section 7 we control the low frequency part of
cubic interactions which is the most delicate part of our analysis.
In Section 8 we complete the proof of our main theorem.

\section*{Acknowledgements}
D. Li is supported in part by NSF Grant-0908032 and NSF DMS-1128155. Y. Wu is partially supported
by NSF of China No. 11101042 and China Postdoctoral Science Foundation No. 20110490018. D. Li is also
supported in part by a Nserc discovery grant. We would like to thank
Prof. Weinan E,  Yan Guo,  Tom Spencer and  Zhouping Xin for their interest in this
work and some helpful remarks. The first author would like to thank Prof. Changxing Miao and the hospitality of
Institute for Applied Physics and Computational Mathematics where the work is conducted.


\section{Preliminaries}

\subsection{Some notations}
We write $X \lesssim Y$ or $Y \gtrsim X$ to indicate $X \leq CY$ for some constant $C>0$.  We use $O(Y)$ to denote any quantity $X$
such that $|X| \lesssim Y$.  We use the notation $X \sim Y$ whenever $X \lesssim Y \lesssim X$.
If $C$ depends upon some additional parameters, we will indicate this with
subscripts; for example, $X \lesssim_u Y$ denotes the assertion that $X \leq C_u Y$ for some $C_u$ depending on $u$. Sometimes
when the context is clear, we will suppress the dependence on $u$ and write $X \lesssim_u Y$ as $X \lesssim Y$.
We will write $C=C(Y_1, \cdots, Y_n)$ to stress that the constant $C$ depends on quantities $Y_1$, $\cdots$, $Y_n$.
We denote by $X\pm$ any quantity of the form $X\pm \epsilon$ for any $\epsilon>0$.

We use the `Japanese bracket' convention $\langle x \rangle := (1 +|x|^2)^{1/2}$. It is convenient to use
the notation $\langle \nabla \rangle =\sqrt{1-\Delta}$ to denote
\begin{align}
 \widehat{\langle \nabla \rangle f} (\xi) = (1+|\xi|^2)^{\frac 12} \hat f(\xi).  \label{not_jnab}
\end{align}
In a similar manner one can define $\langle \nabla \rangle^s$ and $|\nabla|^s$ for any $s\in \mathbb R$.

For any function $f$ on $\mathbb R^d$, we shall use the notation $\| f\|_{L^p}$ or $\| f \|_p$ to denote the
usual Lebesgue norm for $1\le p \le \infty$.

We write $L^q_t L^r_{x}$ to denote the Banach space with norm
$$ \| u \|_{L^q_t L^r_x(\R \times \R^d)} :=
\Bigl(\int_\R \Bigl(\int_{\R^d} |u(t,x)|^r\ dx\Bigr)^{q/r}\ dt\Bigr)^{1/q},$$
with the usual modifications when $q$ or $r$ are equal to infinity,
or when the domain $\R \times \R^d$ is replaced by a smaller
region of spacetime such as $I \times \R^d$.  When $q=r$ we abbreviate $L^q_t L^q_x$ as $L^q_{t,x}$.

We will use $\phi\in C^\infty(\R^d)$ to be a
radial bump function supported in the ball $\{ x \in \R^d: |x| \leq
\frac{25} {24} \}$ and equal to one on the ball $\{ x \in \R^d: |x|
\leq 1 \}$.  For any constant $C>0$, we denote $\phi_{\le C}(x):=
\phi \bigl( \tfrac{x}{C}\bigr)$ and $\phi_{> C}:=1-\phi_{\le C}$. We also
denote $\chi_{|x|>C} = \chi_{>C} =\phi_{>C}$ (resp. $\chi_{|x|\le C}$) sometimes.

We will often need the Fourier multiplier operators defined by the following:
\begin{align}
\mathcal F \Bigl({T_{m(\xi,\eta)}(f,g)} \Bigr)(\xi)& = \int m(\xi,\eta) \hat f(\xi-\eta) \hat g(\eta) d\eta, \notag \\
\mathcal F\Bigl( {T_{m(\xi,\eta,\sigma)}(f,g,h)} \Bigr)(\xi) &= \int
m(\xi,\eta,\sigma) \hat f(\xi-\eta) \hat g(\eta-\sigma) \hat h(\sigma) d\eta d\sigma. \label{T_multiply}
\end{align}
Similarly one can define $T_{m}(f_1,\cdots,f_n)$ for functions $f_1,\cdots,f_n$ and a general symbol
$m=m(\xi,\eta_1,\cdots, \eta_{n-1})$.

\subsection{Basic harmonic analysis}\label{ss:basic}

For each number $N > 0$, we define the Fourier multipliers
\begin{align*}
\widehat{P_{\leq N} f}(\xi) &:= \phi_{\leq N}(\xi) \hat f(\xi)\\
\widehat{P_{> N} f}(\xi) &:= \phi_{> N}(\xi) \hat f(\xi)\\
\widehat{P_N f}(\xi) &:= (\phi_{\leq N} - \phi_{\leq N/2})(\xi) \hat
f(\xi)
\end{align*}
and similarly $P_{<N}$ and $P_{\geq N}$.  We also define
$$ P_{M < \cdot \leq N} := P_{\leq N} - P_{\leq M} = \sum_{M < N' \leq N} P_{N'}$$
whenever $M < N$.  We will usually use these multipliers when $M$ and $N$ are \emph{dyadic numbers} (that is, of the form $2^n$
for some integer $n$); in particular, all summations over $N$ or $M$ are understood to be over dyadic numbers.  Nevertheless, it
will occasionally be convenient to allow $M$ and $N$ to not be a power of $2$.  As $P_N$ is not truly a projection, $P_N^2\neq P_N$,
we will occasionally need to use fattened Littlewood-Paley operators:
\begin{equation}\label{PMtilde}
\tilde P_N := P_{N/2} + P_N +P_{2N}.
\end{equation}
These obey $P_N \tilde P_N = \tilde P_N P_N= P_N$.

Like all Fourier multipliers, the Littlewood-Paley operators commute with the propagator $e^{it\Delta}$, as well as with
differential operators such as $i\partial_t + \Delta$. We will use basic properties of these operators many times,
including

\begin{lem}[Bernstein estimates]\label{Bernstein}
 For $1 \leq p \leq q \leq \infty$,
\begin{align*}
\bigl\| |\nabla|^{\pm s} P_M f\bigr\|_{L^p_x(\R^d)} &\sim M^{\pm s} \| P_M f \|_{L^p_x(\R^d)},\\
\|P_{\leq M} f\|_{L^q_x(\R^d)} &\lesssim M^{\frac{d}{p}-\frac{d}{q}} \|P_{\leq M} f\|_{L^p_x(\R^d)},\\
\|P_M f\|_{L^q_x(\R^d)} &\lesssim M^{\frac{d}{p}-\frac{d}{q}} \| P_M f\|_{L^p_x(\R^d)}.
\end{align*}
\end{lem}

We shall use the following lemma several times which allows us to commute the $L^p$ estimates
with the linear flow $e^{it \langle \nabla \rangle}$. Roughly speaking it says that for $t\gtrsim 1$,
\begin{align*}
\|P_{<t^{C}} e^{it\langle \nabla \rangle} f \|_{p} \lesssim t^{0+} \|f\|_{p}, \quad p=2+ \text{ or } p=2-.
\end{align*}
\begin{lem} \label{lem_rbj_3}
For any $1\le p \le \infty$, $t\ge 0$ and dyadic $M>0$, we have
\begin{align}
\| e^{it\langle \nabla \rangle } P_{<M} g \|_p \lesssim \langle Mt \rangle^{|1-\frac 2p|} \| g\|_p.
\label{rbj_3e1}
\end{align}
Also for any $1\le p \le \infty$, $t\ge 0$,  $s>|1-\frac 2p|$, we have
\begin{align}
\| e^{it \jnab} g \|_p \lesssim \langle t \rangle^{|1-\frac 2p|} \| \jnab^s g \|_p.
\label{rbj_3e2}
\end{align}
In particular for any $0 \le \epsilon <1$, we have
\begin{align}
&\| e^{it \jnab} g \|_{2+\epsilon} \lesssim_{\epsilon} \langle t \rangle^{\frac{\epsilon}{2+\epsilon}}
\| \jnab^{\frac{\epsilon}2} g\|_{2+\epsilon}, \notag \\
&\| e^{it \jnab} g \|_{2-\epsilon} \lesssim_{\epsilon} \langle t \rangle^{\frac{\epsilon}{2-\epsilon}}
\| \jnab^{\epsilon} g \|_{2-\epsilon}.  \label{cor_rbj_3a}
\end{align}
\end{lem}
\begin{proof}
We first prove \eqref{rbj_3e1}.
The idea is to use interpolation between $p=1$, $p=2$ and $p=\infty$. We consider only the case
$p=\infty$. The other
case $p=1$ is similar. To establish the inequality it suffices to bound the
$L_x^1$ norm of the kernel $e^{it\langle \nabla \rangle} P_{<M}$.

Note that $e^{it\langle \nabla \rangle} P_{<M} f = K*f$,
where
\begin{align*}
\hat K(\xi) = e^{it \langle \xi \rangle} \phi(\frac{\xi} M).
\end{align*}
Observe $\| K\|_{L_x^2} \lesssim M$ and for $t>0$,
\begin{align*}
\| |x|^2 K(x) \|_{L_x^2}= \| \partial_{\xi}^2 ( \hat K(\xi) ) \|_{L_{\xi}^2} \lesssim t^2 M + t + \frac 1 M.
\end{align*}
Then
\begin{align*}
\| K\|_{L_x^1} \lesssim \| K\|_{L_x^2}^{\frac 12} \| |x|^2 K \|_{L_x^2}^{\frac 12} \lesssim \langle Mt \rangle.
\end{align*}
The desired inequality then follows from Young's inequality.

Next we show \eqref{rbj_3e2}. By \eqref{rbj_3e1} and the inequality $\langle Mt \rangle
\le \langle M \rangle \langle t \rangle$, we have
\begin{align*}
\| e^{it \langle \nabla \rangle} g \|_p
& \lesssim \| e^{it\langle \nabla \rangle} P_{<1} g \|_p + \sum_{M>1}
\| e^{it \langle \nabla \rangle} P_M g \|_p \notag \\
& \lesssim \langle t \rangle^{|1-\frac 2p|} \| g\|_p
+\sum_{M>1} M^{|1-\frac 2p|} \langle t \rangle^{|1-\frac 2p|}
\| P_M g \|_p \notag \\
& \lesssim \langle t \rangle^{|1-\frac 2p|} \| \langle \nabla \rangle^{s} g \|_p.
\end{align*}

\end{proof}

\begin{lem} \label{lem_rbj_1}
Suppose $m=m(\xi,\eta) \in C^3(\mathbb R^2 \times \mathbb R^2)$ satisfies
\begin{align}
|m| + | \partial_{\xi}^3 m| + | \partial_{\eta}^3 m | \in L^2_{\xi,\eta}(\mathbb R^2 \times
\mathbb R^2). \label{rbj_1e1}
\end{align}
Then
\begin{align}
\| T_m(f,g)\|_r \lesssim \| f\|_{p_1} \| g \|_{p_2}, \label{rbj_1e2}
\end{align}
for any $\frac 1 r = \frac 1 {p_1} + \frac 1 {p_2}$, $1\le r, p_1,p_2 \le \infty$.
\end{lem}

\begin{proof}[Proof of Lemma \ref{lem_rbj_1}]
Let
\begin{align*}
K(x,y) = \frac {1} {(2\pi)^4} \int m(\xi,\eta) e^{i (x\cdot \xi + y \cdot \eta)} d\xi d\eta.
\end{align*}

By \eqref{rbj_1e1}, easy to check that
\begin{align*}
\| K\|_{L_{x,y}^1(\mathbb R^2 \times \mathbb R^2)}
& \lesssim \| (1+|x|^3 + |y|^3) K(x,y) \|_{L_{xy}^2(\mathbb R^2 \times \mathbb R^2) } \notag \\
& \lesssim \| m \|_{L^2_{\xi, \eta}(\mathbb R^2 \times \mathbb R^2)}
+ \| \partial_{\xi}^3 m \|_{L_{\xi,\eta}^2 (\mathbb R^2 \times \mathbb R^2)}
+ \| \partial_{\eta}^3 m \|_{L_{\xi,\eta}^2 (\mathbb R^2 \times \mathbb R^2)}
<\infty.
\end{align*}

Define
\begin{align*}
F(x,y) = \frac 1 {(2\pi)^4} \int m(\xi,\eta)
\hat f(\xi-\eta) \hat g(\eta) e^{i (x \cdot \xi + y \cdot \eta)} d\xi  d\eta.
\end{align*}

By Fourier transform,
\begin{align*}
F(x,y) = \int K(x-x^{\prime}, y -y^{\prime}) h(x^{\prime}, y^{\prime}) dx^{\prime} dy^{\prime},
\end{align*}
where
\begin{align*}
h(x^{\prime}, y^{\prime}) & = \frac 1 {(2\pi)^2} \int \hat f (\xi-\eta) \hat g(\eta) e^{i(x^{\prime} \cdot \xi
+ y^{\prime} \cdot \eta)} d\xi d\eta \notag \\
& = f(x^{\prime}) g(x^{\prime} +y^{\prime}).
\end{align*}

By Young's inequality and H\"older, we then have
\begin{align*}
\| (T_{m}(f,g))(x) \|_{L_x^r} & = \| F(x,0) \|_{L_x^r} \notag \\
& \le \int \| \int  K(x-x^{\prime}, y-y^{\prime}) f(x^{\prime}) g(x^{\prime}+y^{\prime})
d x^{\prime}\|_{L_x^r } dy^{\prime} \notag \\
& \le \int \| K(\cdot, y-y^{\prime} ) \|_{L_x^1} \| f\|_{L_x^{p_1}} \| g \|_{L_x^{p_2}} d y^{\prime} \notag \\
&= \| K\|_{L_{x,y}^1} \| f \|_{p_1} \| g\|_{p_2}.
\end{align*}

\end{proof}

By a similar proof we have
\begin{cor} \label{cor_rbj_1}
Suppose $m=m(\xi,\eta,\sigma) \in C^4 (\mathbb R^2 \times \mathbb R^2 \times \mathbb R^2)$ satisfies
\begin{align}
\| m \|_{L_{\xi,\eta,\sigma}^2} + \| \partial_{\xi}^4  m \|_{L_{\xi,\eta,\sigma}^2}
+ \| \partial_{\eta}^4  m \|_{L_{\xi,\eta,\sigma}^2}
+ \| \partial_{\sigma}^4  m \|_{L_{\xi,\eta,\sigma}^2} \le A <\infty,
\label{cond_three}
\end{align}
then
\begin{align*}
\| T_{m}(f,g,h) \|_r \le C \cdot A \cdot \| f \|_{p_1} \cdot \| g \|_{p_2} \cdot \| h \|_{p_3},
\end{align*}
for any $\frac 1 r = \frac 1 {p_1} + \frac 1 {p_2} + \frac 1{p_3}$, $1\le r, p_1,p_2,p_3\le \infty$. Here
$C>0$ is an absolute constant.
\end{cor}

We shall need to use the following simple  Sobolev embedding lemma.

\begin{lem} \label{lem_J15_2}
Let the numbers $(r,p)$ satisfy $2<r<\infty$, $r>p$, $p \ge (\frac 12 +\frac 1r)^{-1}$. Then
for any smooth $f$ on $\mathbb R^2$,  we have
\begin{align}
\| |\nabla|^{-1} f \|_{r} \lesssim \| \langle x \rangle f \|_p. \label{J15_2_tmp020_e50a}
\end{align}
In particular, for any $2 \le p<r<\infty$, we have
\begin{align*}
\| |\nabla|^{-1} f \|_r \lesssim \| \langle x \rangle f \|_p.
\end{align*}
\end{lem}
\begin{proof}[Proof of Lemma \ref{lem_J15_2}]
We only need to prove \eqref{J15_2_tmp020_e50a}. By Sobolev embedding and H\"older, we have
\begin{align*}
\| |\nabla|^{-1} f \|_r &\lesssim \| f \|_{ (\frac 12 + \frac 1r )^{-1}}  \notag \\
& \lesssim \| \langle x \rangle f \|_{p} \cdot \| \langle x \rangle^{-1} \|_{
\bigl(\frac 12 + \frac 1r -\frac 1p\bigr)^{-1}} \notag \\
& \lesssim \| \langle x \rangle f \|_p.
\end{align*}

\end{proof}

\begin{lem}[Bounds on the phase function] \label{lem_rbj_5}
Let $\psi(x,y) = \frac 1 {\langle x \rangle + \langle y \rangle - \langle x + y \rangle}$
for $x,y \in \mathbb R^2$. Then
\begin{align}
|\partial_x^{\alpha} \partial_y^{\beta}
\psi(x,y) | \lesssim_{\alpha,\beta}
\min\{ \langle x \rangle, \langle y \rangle, \langle x +y \rangle \},
\qquad\forall\, x,y \in \mathbb R^2. \label{rbj_5e0}
\end{align}
\end{lem}

\begin{proof}
Write
\begin{align}
\psi(x,y) & = \frac{\langle x \rangle + \langle y \rangle + \langle x+y \rangle}
{(\langle x \rangle + \langle y \rangle)^2- (\langle x +y\rangle)^2} \notag \\
& = \frac{\langle x \rangle + \langle y \rangle + \langle x+y \rangle}
{1+2(\langle x \rangle \langle y \rangle -x\cdot y )} \notag \\
&=: \frac{\langle x \rangle + \langle y \rangle + \langle x +y \rangle} B.
\label{rbj_5e1}
\end{align}

We first show that
\begin{align}
| \partial_x^{\alpha} \partial_y^{\beta} ( \frac 1 B) | \lesssim_{\alpha,\beta}
\frac 1 B. \quad \label{rbj_5e2}
\end{align}

We begin with the estimate
\begin{align}
\frac{|\partial_x B|}{B} \lesssim 1. \label{rbj_5e3}
\end{align}

This is equivalent to
\begin{align}
| \frac x {\langle x \rangle} \langle y \rangle -y | \lesssim 1 +
 (\langle x \rangle \langle y \rangle - x \cdot y). \label{rbj_5e4}
\end{align}

Denote $\theta = \frac{x\cdot y}{|x||y|}$. It is obvious that \eqref{rbj_5e2} holds
for $-1\le \theta \le 0$. Therefore we only need to consider the case
$0<\theta \le 1$. Taking the square on both sides of \eqref{rbj_5e4}, we see that it suffices
to prove for some $0<\epsilon<1$ the inequality
\begin{align}
\frac{|x|^2}{\langle x \rangle^2} \langle y \rangle^2 +|y|^2
-2 \frac{\langle y \rangle}{\langle x \rangle} |x||y| \theta
\le \frac 1 {\epsilon} \Bigl( 1+ ( \langle x \rangle \langle y \rangle -|x||y|\theta)^2
\Bigr). \label{rbj_5e5}
\end{align}

Now consider the function
\begin{align*}
F(\theta) = |x|^2 |y|^2 \theta^2 - 2 |x||y| \langle y \rangle ( \langle x \rangle - \frac{\epsilon}{\langle x \rangle}
) \theta.
\end{align*}

By using the obvious inequality $$\langle x \rangle - |x| \ge \frac 1 {2 \langle x \rangle},$$
it is not difficult to check that for $0<\epsilon\le \frac 12$,
\begin{align*}
\frac{\langle y \rangle (\langle x \rangle - \frac{\epsilon}{\langle x \rangle})}{|x||y|} >1.
\end{align*}

Since $0\le \theta \le 1$, clearly $F(\theta)$ achieves its minimum at $\theta=1$. Therefore it suffices
to prove \eqref{rbj_5e5} or equivalently \eqref{rbj_5e4} for $\theta=1$.

Consider \eqref{rbj_5e4} for $\theta=1$. We have
\begin{align*}
| \frac x {\langle x \rangle} \langle y \rangle - y| & =
\left | \frac{|x|}{\langle x \rangle} \langle y \rangle - |y| \right|
\notag \\
& \lesssim 1 + |y| \cdot \left| \frac{|x|}{\langle x \rangle}-1 \right| \notag \\
& = 1+ \frac{|y|}{\langle x \rangle}(\langle x \rangle -|x|).
\end{align*}

On the other hand
\begin{align}
\langle x \rangle \langle y \rangle - |x| |y| \ge (\langle x \rangle -|x| ) |y|.
\label{rbj_5e10}
\end{align}

Therefore \eqref{rbj_5e4} holds and consequently \eqref{rbj_5e3} is proved. By using
an estimate similar to \eqref{rbj_5e10}, we have
\begin{align}
\langle x \rangle \langle y \rangle - |x||y| \gtrsim \max\{
\frac{|y|}{\langle x \rangle}, \, \frac{|x|}{\langle y \rangle} \}. \label{rbj_5e10a}
\end{align}

This together with \eqref{rbj_5e3} obviously implies that
\begin{align}
\frac{|\partial_x B| + |\partial_y B| + \frac{\langle x \rangle}{\langle y \rangle}
+ \frac{\langle y \rangle}{\langle x \rangle}} B \lesssim 1. \label{rbj_5e11}
\end{align}

It is easy to check that
\begin{align}
\left| \partial_x^{\alpha} \partial_y^{\beta} B\right| \lesssim_{\alpha,\beta} \frac{\langle y \rangle}{\langle x \rangle}
+ \frac{ \langle x \rangle}{\langle y \rangle}, \quad\forall\, |\alpha|+|\beta|\ge 2. \label{rbj_5e11_tmpa}
\end{align}

The estimate \eqref{rbj_5e2} now follows from \eqref{rbj_5e11}, \eqref{rbj_5e11_tmpa} and an induction argument.

By \eqref{rbj_5e1} and \eqref{rbj_5e2}, we have
\begin{align*}
| \partial_x^{\alpha} \partial_y^{\beta} \psi(x,y)| \lesssim \psi(x,y).
\end{align*}

It remains for us to prove \eqref{rbj_5e0} for $\alpha=\beta=0$.
If $\langle x +y \rangle \ll \langle x \rangle$ or $\langle x +y \rangle \ll \langle y \rangle$,
the estimate is obvious. Without loss of generality assume $\langle y \rangle \ge \langle  x\rangle$
and $\min \{ \langle x \rangle, \langle y \rangle, \langle x+y \rangle \} \sim \langle x \rangle$.
Then by \eqref{rbj_5e10a} and \eqref{rbj_5e1}, we have
\begin{align*}
\psi(x,y) \le \frac{\langle x \rangle + \langle y \rangle}{ 1 + \frac{|y|}{\langle x \rangle} }
\lesssim \langle x \rangle.
\end{align*}

Therefore \eqref{rbj_5e0} is proved.

\end{proof}

We need a simple lemma from vector algebra.

\begin{lem} \label{lem_rbj_4}
For any $x\in \mathbb R^2$, $y\in \mathbb R^2$, we have
\begin{align}
\frac{x}{\langle x \rangle} - \frac{y}{\langle y \rangle} = Q(x,y) (x-y), \label{rbj_4e1}
\end{align}
where $Q(x,y)=Q$ is a matrix given by the expression
\begin{align}
Q_{ij} = \frac 1 {\langle y \rangle}
\Bigl( I- \frac{x (x+y)^T}{\langle x \rangle (\langle x \rangle +\langle y \rangle)}
\Bigr)_{ij} = \frac 1 {\langle y \rangle}
\Bigl( \delta_{ij} - \frac{x_i (x_j+y_j)} { \langle x \rangle (\langle x \rangle +\langle y \rangle)}
\Bigr), \qquad 1\le i,j\le 2. \label{rbj_4e1a}
\end{align}
Denote $\tilde x = (-x_2,x_1)^T$, $\tilde y = (-y_2,y_1)^T$. Then
\begin{align}
Q^{-1} = \langle x \rangle \langle y \rangle (\langle x \rangle +\langle y \rangle)
\bigl(1+ \langle x \rangle \langle y \rangle -x \cdot y \bigr)^{-1}
\Bigl( I-\frac{(\tilde x +\tilde y ) (\tilde x)^T}{ \langle x \rangle(\langle x \rangle +\langle y \rangle)}
\Bigr). \label{rbj_4e2}
\end{align}
We have the pointwise bounds:
\begin{align}
& | \partial_x^{\alpha} \partial_y^{\beta} Q(x,y) | \lesssim_{\alpha,\beta} \langle y \rangle^{-1}, \; \forall\,
\alpha,\beta; \notag \\
& | \partial_x^{\alpha} \partial_y^{\beta} ( Q^{-1}(x,y)) |
\lesssim_{\alpha,\beta} \langle x \rangle^3 + \langle y \rangle^3, \qquad\forall\,\alpha,\beta. \label{rbj_4e3}
\end{align}

\end{lem}
\begin{proof}
We first show \eqref{rbj_4e1}:
\begin{align*}
\frac x {\langle x \rangle} - \frac  y{\langle y \rangle}
& = x ( \frac 1 {\langle x \rangle} - \frac 1 {\langle y \rangle} ) + \frac 1 {\langle y \rangle}
( x-y) \notag \\
& = x \frac{(y+x)^T (y-x)} {\langle x \rangle \langle y \rangle (\langle x \rangle + \langle y \rangle)}
+\frac 1 {\langle y \rangle} (x-y) \notag \\
&= \frac 1 {\langle y \rangle}
\Bigl( I - \frac {x (x+y)^T} {\langle x \rangle (\langle x \rangle + \langle y \rangle )} \Bigr)
(x-y).
\end{align*}
Since $Q$ is a two by two matrix, the expression for $Q^{-1}$ is a straightforward computation.
The bounds \eqref{rbj_4e3} follow easily from \eqref{rbj_4e1a}, \eqref{rbj_4e2} and a similar estimate
as in \eqref{rbj_5e2}.
\end{proof}

We shall need to exploit some subtle cancelations of the phases. The following lemma will be useful
in our nonlinear estimates.

\begin{lem}[Transformation of phase derivatives] \label{lem_phase}
Consider the following phases:
\begin{align}
\phi_1(\xi,\eta,\sigma) &= \langle \xi \rangle + \langle \xi-\eta\rangle -\langle \eta-\sigma\rangle
-\langle \sigma\rangle, \notag \\
\phi_2(\xi,\eta,\sigma) & = \langle \xi \rangle - \langle \xi-\eta \rangle +\langle \eta-\sigma \rangle
-\langle \sigma \rangle, \notag \\
\phi_3(\xi,\eta,\sigma) & = \langle \xi \rangle - \langle \xi -\eta \rangle
-\langle \eta -\sigma \rangle + \langle \sigma \rangle. \notag
\end{align}
There exist smooth matrix functions
$Q_{11}=Q_{11}(\xi,\eta,\sigma)$, $Q_{12}=Q_{12}(\xi,\eta,\sigma)$, $Q_{21}=Q_{21}(\xi,\eta)$,
$Q_{22}=Q_{22}(\eta,\sigma)$, $Q_{31}=Q_{31}(\xi,\eta)$, $Q_{32}=Q_{32}(\eta,\sigma)$ such that
\begin{align*}
\partial_{\xi} \phi_1 &= Q_{11}(\xi,\eta,\sigma) \partial_{\eta} \phi_1 +
Q_{12}(\xi,\eta,\sigma) \partial_{\sigma} \phi_1,  \notag \\
\partial_{\xi} \phi_2 & = Q_{21}(\xi,\eta) Q_{22}(\eta,\sigma) \partial_{\sigma} \phi_2, \notag \\
\partial_{\xi} \phi_3 & = Q_{31}(\xi,\eta) Q_{32} (\eta,\sigma) \partial_{\sigma} \phi_3.
\end{align*}

Moreover we have the point-wise  bounds
\begin{align}
| \partial_{\xi}^{\alpha} \partial_{\eta}^{\beta} \partial_{\sigma}^{\gamma}
Q_{11}(\xi,\eta,\sigma)| + |\partial_{\xi}^{\alpha} \partial_{\eta}^{\beta} \partial_{\sigma}^{\gamma}
Q_{12}(\xi,\eta,\sigma)|&\lesssim_{\alpha,\beta,\gamma}
\langle |\xi|+|\eta|+|\sigma| \rangle^{3}, \quad\forall\, \alpha,\,\beta,\, \gamma;\notag \\
| \partial_{\xi}^{\alpha} \partial_{\eta}^{\beta}
Q_{21}(\xi,\eta)| +| \partial_{\xi}^{\alpha} \partial_{\eta}^{\beta}
Q_{31}(\xi,\eta)|  & \lesssim_{\alpha,\beta} 1, \quad\forall\, \alpha,\,\beta; \notag \\
|\partial_{\eta}^{\alpha} \partial_{\sigma}^{\beta} Q_{22}(\eta,\sigma)|+
|\partial_{\eta}^{\alpha} \partial_{\sigma}^{\beta} Q_{32}(\eta,\sigma)|
& \lesssim_{\alpha,\beta} \langle |\eta|+|\sigma| \rangle^{3},\quad\forall\, \alpha,\,\beta.
\label{9233000}
\end{align}
\end{lem}
\begin{proof}
We prove it for $\phi_1$. The other two cases are simpler. By Lemma \ref{lem_rbj_4}, we write
\begin{align*}
\partial_{\xi} \phi_1 &= \frac{\xi}{\langle \xi\rangle} + \frac{\xi-\eta}{\langle \xi-\eta\rangle}
=\tilde Q_1(\xi,\eta) \cdot (2\xi-\eta), \\
\partial_{\eta} \phi_1 & = \frac{\eta-\xi}{\langle \eta-\xi\rangle}
- \frac{\eta-\sigma}{\langle \eta-\sigma \rangle} = \tilde Q_2(\xi,\eta,\sigma)\cdot (\xi-\sigma), \\
\partial_{\sigma} \phi_1 & = \frac{\eta-\sigma}{\langle \eta-\sigma \rangle} - \frac{\sigma}
{\langle \sigma \rangle} = \tilde Q_3(\eta,\sigma) \cdot (\eta-2\sigma).
\end{align*}
Hence
\begin{align*}
\partial_{\xi} \phi_1 & = \tilde Q_1 \left( 2 {\tilde Q_2}^{-1} \partial_{\eta} \phi_1 - {\tilde Q_3}^{-1}
\partial_{\sigma} \phi_1 \right) \\
& =: Q_{11} \partial_{\eta} \phi_1 + Q_{12} \partial_{\sigma} \phi_1.
\end{align*}

The bound \eqref{9233000} is obvious.
\end{proof}

\section{Preliminary transformations}
Since the function $h=h(t,x)$ is complex-valued, we write it as
\begin{align*}
h(t,x)=h_1(t,x)+ i h_2(t,x).
\end{align*}
By \eqref{e92_537a} and \eqref{eq_vh}, we have
\begin{align*}
u & = \frac{|\nabla|}{\jnab} h_1, \\
\v v &= - \frac {\nabla} {|\nabla|} h_2.
\end{align*}

In Fourier space, \eqref{e92_537b} then takes the form
\begin{align*}
&\hat h(t,\xi) \notag \\
= & e^{it\langle \xi \rangle} \widehat{h_0}(\xi)
- \int_0^t \int e^{i(t-s) \langle \xi \rangle}
\langle \xi \rangle \langle \eta \rangle^{-1}
|\eta| \frac{\xi \cdot (\xi-\eta)}{|\xi||\xi-\eta|}
\widehat{h_1}(s,\eta) \widehat{h_2}(s,\xi-\eta) d\eta ds \notag \\
& \qquad +
\frac {i} 2 \int_0^t \int e^{i(t-s) \langle \xi \rangle}
|\xi| \frac{|\eta||\xi-\eta|}{\langle \eta \rangle \langle \xi-\eta \rangle}
\widehat{h_1}(s,\eta) \widehat{h_1}(s,\xi-\eta) d\eta ds \notag \\
& \qquad - \frac i 2
\int_0^t \int e^{i(t-s) \langle \xi \rangle}
|\xi| \frac{\eta \cdot(\xi-\eta)}{|\eta| |\xi-\eta|}
\widehat{h_2}(s,\eta) \widehat{h_2}(s,\xi-\eta) d\eta ds.
\end{align*}

Denote
\begin{align*}
f(t)=e^{-it \langle \nabla \rangle} h(t).
\end{align*}
Then after a tedious calculation,
\begin{align}
\hat f(t,\xi) &= \hat h_0(\xi) +
\int_0^t \int e^{-is(\langle \xi \rangle -\langle \eta \rangle -\langle \xi-\eta \rangle)}
\Bigl( \psa \notag \\
& \qquad + \psb + \psc \Bigr)
\hat f (s,\eta) \hat f(s,\xi-\eta) d\eta ds \notag \\
& \;\;
+\int_0^t \int e^{-is(\langle \xi \rangle -\langle \eta \rangle +\langle \xi-\eta \rangle)}
\Bigl( -\psa \notag \\
& \qquad +  \psb -\psc \Bigr)
\hat f (s,\eta) \overline{\hat f(s,\eta-\xi)} d\eta ds \notag \\
& \;\;
+\int_0^t \int e^{-is(\langle \xi \rangle +\langle \eta \rangle -\langle \xi-\eta \rangle)}
\Bigl( \psa \notag \\
& \qquad +  \psb -\psc \Bigr)
\overline{\hat f (s,-\eta)} {\hat f(s,\xi-\eta)} d\eta ds \notag \\
& \;\;
+\int_0^t \int e^{-is(\langle \xi \rangle +\langle \eta \rangle +\langle \xi-\eta \rangle)}
\Bigl( -\psa \notag \\
& \qquad +  \psb +\psc \Bigr)
\overline{\hat f (s,-\eta)} \;\overline{\hat f(s,\eta-\xi)} d\eta ds. \label{eq_complicated}
\end{align}
Here $\overline{\hat f}$ denote the complex conjugate of $\hat f$. Note that
\begin{align*}
\overline{\hat f(t,-\xi)} &= e^{it\langle \xi \rangle} \hat{\bar h}(t,\xi), \\
\hat f(t,\xi) &= e^{-it\langle \xi \rangle} \hat h(t,\xi).
\end{align*}
To simplify matters, we shall write \eqref{eq_complicated} collectively as
\begin{align}
\hat f(t,\xi) = \hat h_0(\xi) + \int_0^t \int e^{-is\phi_0(\xi,\eta)} m_0(\xi,\eta)
\hat f(s,\xi-\eta) \hat f(s,\eta) d\eta ds, \label{e927_546a}
\end{align}
where
\begin{align}
\phi_0(\xi,\eta) = \langle \xi \rangle \pm \langle \xi-\eta\rangle \pm \langle \eta \rangle, \label{phase1}
\end{align}
and $m_0(\xi,\eta)$ is given by (after some symmetrization between $\eta$ and $\xi-\eta$)
\begin{align*}
m_0(\xi,\eta) & = const\cdot \mveca {\xi} {\eta}
+ const \cdot \langle \xi \rangle \frac{\xi \cdot (\xi-\eta)} {|\xi| |\xi-\eta|} \frac{|\eta|}{\langle \eta \rangle}
\notag \\
&\qquad+ const\cdot |\xi| \cdot \frac{|\eta|}{\langle \eta\rangle}
\cdot \frac{|\xi-\eta|} {\langle \xi-\eta\rangle} + const \cdot \mvecc {\xi} {\eta} \notag \\
& := \sum_{i=1}^4 m_i(\xi,\eta).
\end{align*}

Here and in the rest of this paper we shall abuse slightly the notations and
denote $\hat f(t,\xi)$ to be either itself or its complex conjugate (i.e. $\overline{\hat f(t,-\xi)}$,
see \eqref{eq_complicated}).
Note that in the expression
of $m_0(\xi,\eta)$ there are four types of symbols. For $w=(w_1,w_2) \in \mathbb R^2$, define
\begin{align*}
r_1(w)= \frac{w_1}{|w|}, \; r_2(w) = \frac{w_2}{|w|}, \; r_3(w)= \frac{|w|}{\langle w \rangle}.
\end{align*}
We write $m_0(\xi,\eta)$ collectively as
\begin{align}
m_0(\xi,\eta) = \sum_{1\le j,k,l\le 3} a_{jkl} \cdot
\langle \xi \rangle \cdot r_j(\xi) r_k(\xi-\eta) r_l(\eta), \label{p10_J11_e1}
\end{align}
where $a_{jkl}$ are some constant coefficients.
For example
\begin{align*}
m_3(\xi,\eta) & = const \cdot \langle \xi \rangle \cdot \frac {|\xi|}{\langle \xi \rangle}
\cdot \frac{|\xi-\eta|} {\langle \xi -\eta \rangle} \cdot \frac{|\eta|} {\langle \eta \rangle} \notag \\
& = const \cdot \langle \xi \rangle \cdot r_3(\xi) r_3(\xi-\eta) r_3(\eta).
\end{align*}

Although the frequency variables $(\xi,\eta)$ are vectors, this fact will play no role in our analysis.
The actual value of the constants $a_{jkl}$  will also not be important.
Therefore we shall suppress the subscript notations and summation in \eqref{p10_J11_e1},
pretend everything is scalar valued
and regard $m_0(\xi,\eta)$ as any one of the summand in \eqref{p10_J11_e1}.
Observe $m_0(\xi,\eta)$ is symmetric in the sense that
\begin{align}
m_0(\xi,\eta)=m_0(\xi,\xi-\eta). \label{m0_sym}
\end{align}

The nice feature of Klein-Gordon is (cf. Lemma \ref{lem_rbj_5})
\begin{align*}
|\phi_0(\xi,\eta)| \gtrsim 1/
\langle |\xi| +|\eta| \rangle, \qquad \text{ for any $(\xi,\eta)$.}
\end{align*}

By the simple identity
\begin{align*}
e^{-is \phi_0(\xi,\eta)} = \frac 1 {-i \phi_0(\xi,\eta)} \frac {\partial}{\partial s}
\Bigl( e^{-is \phi_0(\xi,\eta)} \Bigr),
\end{align*}
we can then integrate
by parts in the time variable $s$ in \eqref{e927_546a}. By \eqref{m0_sym},
\begin{align}
&\int_0^t \int e^{-is \phi_0(\xi,\eta)} \frac{m_0(\xi,\eta)}{\phi_0(\xi,\eta)} \partial_s \hat f(s,\xi-\eta) \hat f(s,\eta) d \eta \notag \\
= & \int_0^t \int e^{-is \phi_0(\xi,\eta)}
\frac{ m_0(\xi,\eta)}{\phi_0(\xi,\eta)} \partial_s \hat
f(s,\eta) \hat f(s,\xi-\eta) d \eta \label{not_change}
\end{align}
using the change of variable $\eta\to \xi-\eta$. In the above equality we have
abused again the notation and denote $\phi_0(\xi,\eta)=\phi_0(\xi,\xi-\eta)$ since
 it will remain the same form as \eqref{phase1}. By \eqref{e927_546a}, we have
 \begin{align}
 \partial_s \hat f (s,\eta)
 = \int e^{-is \phi_0(\eta,\sigma)}
 m_0(\eta,\sigma)
 \hat f(s,\eta-\sigma) \hat f (s,\sigma) d\sigma. \label{p10_J11_e2}
 \end{align}

Integrating by parts in the time variable $s$ in \eqref{e927_546a},
using \eqref{p10_J11_e2} and \eqref{not_change}, we obtain
\begin{align}
\hat f(t,\xi)& = \widehat{\tilde h_0}(\xi) + \hat g(t,\xi) \notag \\
& \qquad\quad+ \int_0^t \int e^{-is \phi(\xi,\eta,\sigma)}
m_1(\xi,\eta,\sigma) \hat f(s,\xi-\eta) \hat f(s,\eta-\sigma) \hat f(s,\sigma) d\sigma d\eta ds\notag\\
&=: \widehat{\tilde  h_0}(\xi) +\hat g(t,\xi) + \hat f_{\text{cubic}}(t,\xi), \label{e926_552}
\end{align}
where $\tilde h_0$ collects the contribution from the boundary term $s=0$ and data $h_0$:
\begin{align}
\widehat{\tilde h_0}(\xi) &= \widehat{h_0}(\xi) + \int \frac{m_0(\xi,\eta)} {i \phi_0(\xi,\eta)}
\hat h_0(\xi-\eta) \hat h_0(\eta) d\eta \notag \\
&=\widehat{h_0}(\xi) - \hat g(0,\xi); \label{e27_111}
\end{align}
the term $g$ denotes the boundary term arising from $s=t$:
\begin{align}
\hat g(t,\xi) = \int e^{-it \phi_0(\xi,\eta)} \cdot \frac{m_0(\xi,\eta)} {-i \phi_0(\xi,\eta)}
\hat f(t,\xi-\eta) \hat f(t,\eta) d\eta; \label{e27_111a}
\end{align}
$m_1(\xi,\eta,\sigma)$ is given by
\begin{align*}
m_1(\xi,\eta,\sigma) = \frac{m_0(\xi,\eta) m_0(\eta,\sigma)} {i\phi_0(\xi,\eta)}   ;
\end{align*}
 and also
\begin{align*}
\phi(\xi,\eta,\sigma)= \langle \xi \rangle \pm \langle \xi -\eta \rangle
\pm \langle \eta-\sigma \rangle \pm \langle \sigma \rangle.
\end{align*}

Note that
\begin{align*}
&m_0(\xi,\eta) m_0(\eta,\sigma) \notag \\
 =\; &
\sum_{1\le j,k,l,j^{\prime},k^{\prime},l^{\prime}  \le 3}
\langle \xi \rangle \langle \eta \rangle r_j(\xi) r_{k}(\xi-\eta)
r_l(\eta) r_{j^{\prime}}(\eta) r_{k^{\prime}}(\eta-\sigma) r_{l^{\prime}}(\sigma) \notag\\
=\; & \sum_{1\le j,k,l,j^{\prime},k^{\prime},l^{\prime} \le 3}
\langle \xi \rangle \langle \eta \rangle r_{l}(\eta) r_{j^{\prime}}(\eta)
r_j(\xi) r_k(\xi-\eta) r_{k^{\prime}}(\eta-\sigma) r_{l^{\prime}}(\sigma).
\end{align*}

We shall abuse slightly the notations and denote
\begin{align*}
&\widehat{\mathcal R f}(\xi) = r(\xi) \hat f(\xi), \qquad
r(\xi)=r_1(\xi),\; r_2(\xi),\; r_3(\xi), \; \text{or $r_j(\xi)r_{j^{\prime}}(\xi)$, }\\
&\frac {\eta}{|\eta|} = r_l(\eta) r_{j^{\prime}}(\eta).
\end{align*}

The notations $\mathcal R$ and $\frac{\eta}{|\eta|}$ suggest that the functions
$r_j$ and $r_{l}r_{j^{\prime}}$ are essentially the symbols of some Riesz-type
operators or better. Their estimates are the same and the actual form plays no role
in the proof. By adopting the above notations we can simplify greatly the presentation
and also the analysis. In this notation, we shall write
\begin{align}
\widehat{f_{\text{cubic}}}(t,\xi) = const \cdot \widehat{\mathcal R f_3}(t,\xi),
\notag
\end{align}
and
\begin{align}
\widehat{f_3}(t,\xi) & = \int_0^t  \int e^{-is \phi(\xi,\eta,\sigma)}
\cdot \frac{\langle \xi \rangle \cdot \langle \eta \rangle}{\phi_0(\xi,\eta)}
 \widehat{\mathcal R f}(s,\xi-\eta) \notag \\
 & \qquad \frac{\eta}{|\eta|}
 \Bigl( \widehat{\mathcal R f}(s,\eta-\sigma) \widehat{\mathcal R f}(s,\sigma)
 \Bigr) d\sigma d\eta ds. \label{p12_J11_e1}
 \end{align}

 In a similar way, we write the boundary terms as
 \begin{align}
 \hat g (t,\xi)& = const \cdot \widehat{\mathcal R g_1}(t,\xi), \notag \\
\widehat{g_1}(t,\xi) & = \int
e^{-it \phi_0(\xi,\eta)} \cdot \frac{\langle \xi \rangle}{\phi_0(\xi,\eta)}
\widehat{\mathcal R f}(t,\xi-\eta) \widehat{\mathcal R f}(t,\eta) d\eta.
\label{p12_J11_e2}
\end{align}

\section{Local theory, continuity of $X$-norm and $H^{N^\prime}$-estimate}
\label{sec_local}

We recall that
\begin{align}
\partial_th& = i\langle \nabla \rangle h
-  \frac{\langle \nabla \rangle \nabla}{|\nabla|} \cdot (u \v v) +
\frac i2 |\nabla| ( u^2 +|\v v|^2) , \label{e92tmp_537b}
\end{align}
where $h=h_1+i h_2$, and
$$
u=\frac{|\nabla|}{\langle \nabla\rangle}h_1,\qquad \v
v=-\frac{\nabla}{|\nabla|}h_2.
$$
\begin{thm} \label{thm_local}
For any $k\ge 4$, $h_0\in H^k(\R^2)$, there exists
$T_0=T_0(\|h_0\|_{H^k})>0$, and a unique smooth local solution $h\in
C_t^0H^k([0,T_0]\times \R^2)$ to \eqref{e92_537b}.

Moreover, if $h_0 \in H^7(\mathbb R^2)$ and $\|x(1-\Delta)h_0\|_{2+\delta}<\infty$, then
$$
\tilde a(t):=\|x(1-\Delta)e^{-it\langle
\nabla\rangle}h(t)\|_{2+\delta}<\infty,
$$
for any $0\le t\le T_0$, and $\tilde a(t)$ is a continuous function of $t$.

We also have
\begin{align*}
\| h (\tau) \|_{C_{\tau}^0 H^{N^{\prime}} ([0,t])} \lesssim \|h_0\|_{H^{N^{\prime}}} + \| h \|_{X_t}^2 + \| h\|_{X_t}^3.
\end{align*}

\end{thm}
The rest of this section is devoted to the proof of this theorem. We
begin with the $H^k$-local well-posedness theory which is quite
standard. We sketch the details here for the sake of completeness.

\subsection{Energy estimates}

Let $m$ be an integer. By \eqref{e92tmp_537b}, we compute
\begin{align}
\frac12\frac d{dt}\int\partial^mh \>\partial^m\bar{h}&=-\int
\partial^m\Big(\frac{\lnr}{|\nabla|}\nabla\cdot(u\v v)\Big)
\partial^m\frac{\lnr}{|\nabla|}u \notag\\
&\quad+\frac12\int\partial^m|\nabla|\big(u^2+|\v
v|^2\big)\partial^m\left(\frac{\nabla}{|\nabla|}\cdot \v v\right)\notag\\
&=-\int
\partial^m\nabla\cdot(u\v v)
\partial^m\frac{1-\Delta}{-\Delta}u\notag\\
&\quad+\frac12\int\partial^m\big(u^2+|\v
v|^2\big)\partial^m\left({\nabla}\cdot \v v\right).\label{J18_e2}
\end{align}
\underline{$L^2$-estimate}. Taking $m=0$ in \eqref{J18_e2}, we get
\begin{align*}
\frac12\frac d{dt}\left(\|h(t)\|_{L^2}^2\right)&=\int u\v v\cdot \nabla\frac{1-\Delta}{-\Delta}u+\frac12\int(u^2+|\v v|^2)(\nabla \cdot
\v v)\\
&\lesssim
\|u\|_{\infty}\|\v v\|_{L^2}\Big\|\frac{\lnr^2}{|\nabla|}u\Big\|_2+\|\nabla\cdot
\v v\|_\infty\big(\|u\|_2^2+\|\v v\|_2^2\big)\\
&\lesssim \big(\|u\|_\infty+\|\nabla\cdot\v
v\|_{\infty}\big)\|h\|_{H^k}^2.
\end{align*}
\underline{$H^k$-estimate}. Taking $m=k$ in \eqref{J18_e2},  we
have
\begin{align}
\frac12\frac d{dt}\left(\|\partial^k
h(t)\|_{L^2}^2\right)=&-\int\partial^k \nabla\cdot( u\v v)\>
(\partial^k (-\Delta)^{-1}u)\label{J18_e3a}\\
&-\int\partial^k \nabla\cdot( u\v v)\>
\partial^k u\label{J18_e3b}\\
&+\frac12\int\partial^k(u^2)\partial^k(\nabla\cdot \v
v)\label{J18_e3c}\\
&+\frac12\int\partial^k(|\v v|^2)\partial^k(\nabla\cdot \v
v)\label{J18_e3d}.
\end{align}
For \eqref{J18_e3a}, we estimate it as
\begin{align*}
\eqref{J18_e3a}=& - \int\partial^k \nabla u\cdot\v v\> \big(\partial^k (-\Delta)^{-1}u\big)
-\int u\>(\partial^k \nabla\cdot \v v)\>\partial^k (-\Delta)^{-1}u\\
&+\sum\limits_{1\le l \le k} O\Big(\int
\partial^lu\>\partial^{k+1-l}\v v\>\partial^k (-\Delta)^{-1}u
\Big)\\
=&  \;\frac 12 \int \big| \partial^k (-\Delta)^{-1} \nabla u\big|^2 (\nabla\cdot \v v)
+\int u\>\big(\partial^k\v v\cdot \nabla\partial^k(-\Delta)^{-1}u\big)\\
& \;+ O\Big( \int (-\Delta)^{-1} \partial^{k+2} u \cdot \partial \v v \partial^k (-\Delta)^{-1} u \Big) \notag \\
&\;+\sum\limits_{1\le l \le k} O\Big(\int
\partial^lu\>\partial^{k+1-l}\v v\>\partial^k (-\Delta)^{-1}u
\Big)\\
\lesssim & \; \|\nabla\cdot \v
v\|_\infty\>\|u\|_{H^k}^2+\|u\|_\infty\>\|\v
v\|_{H^k}\>\|u\|_{H^k} + \| \partial \v v \|_{\infty} \| u \|_{H^k}^2 \notag \\
&\; +\sum\limits_{l=1}^k \|\partial^l
u\|_{\frac{2(k-1)}{l-1}}\>\|\partial^{k+1-l}\v
v\|_{\frac{2(k-1)}{k-l}}\>\|u\|_{H^k}\\
\lesssim &\; \big(\|u\|_\infty+\|\partial  \v
v\|_{\infty}\big)(\|u\|_{H^k}^2+\|\v
v\|_{H^k}^2)\\
&+\sum\limits_{l=1}^k \|\partial^ku\|_2^{\frac{l-1}{k-1}}\|\partial
u\|_\infty^{\frac{k-l}{k-1}}\>\|\partial^k\v
v\|_2^{\frac{k-l}{k-1}}\|\partial \v
v\|_\infty^{\frac{l-1}{k-1}}\>\|u\|_{H^k}\\
&\lesssim \big(\|u\|_\infty+\|\partial u\|_\infty+\|\partial \v
v\|_\infty \big)\big(\|u\|_{H^k}^2+\|\v
v\|_{H^k}^2\big).
\end{align*}

For \eqref{J18_e3b}, we write it as
\begin{align*}
\eqref{J18_e3b}=&-\int\big(\partial^k \nabla \cdot\v v\big)\>
u\>\big(\partial^ku\big) -\int\big(\partial^k \nabla u \cdot\v
v\big)\>\big(\partial^ku\big) \\
&\quad + \sum_{1\le l \le k} O \Big( \int \partial^l u \partial^{k+1-l} \v v \partial^k u dx \Bigr) \notag \\
=&-\int\big(\partial^k \nabla \cdot\v v\big)\>
u\>\partial^ku +\frac12\int\nabla \cdot\v v
\big(\partial^ku \big)^2+\cdots,
\end{align*}
where ``$\cdots$" denote terms which  can be estimated in a similar
way as that in \eqref{J18_e3a}.

Similarly,
$$
\eqref{J18_e3c}=\int (\partial^k\nabla\cdot \v v)u\partial^k
u+\cdots.
$$
Also, using the fact that $\text{curl} \v v=0$,
$$
\eqref{J18_e3d} = \frac 1 4 \int |\partial^k \v v |^2 (\nabla \cdot \v v)+ \cdots.
$$
Collecting all the estimates, we obtain
$$
\frac12\frac d{dt}\left(\|h(t)\|_{H^k}^2\right)\lesssim
\big(\| u\|_\infty+\|\partial u\|_\infty+\|\partial \v
v\|_{\infty}\big)\|h(t)\|_{H^k}^2.
$$
This concludes the energy estimates.

\subsection{Continuity of $X$-norm along the flow.}

Now we  show that $$\tilde a(t)=\|x(1-\Delta)e^{-it\langle
\nabla\rangle}h(t)\|_{2+\delta}$$  is a continuous function of $t$ (so
that we can use the continuity argument later). Without loss
of generality we shall assume $0\le t \le 1$.

\texttt{Step 1.} For any dyadic $R$, define
$$
A_R=\Big\|\chi_{\frac R2\le |x|\le 2R}\left( \aligned
    &u\\
    &\v v
   \endaligned
  \right)
  \Big\|_p+\Big\|\chi_{\frac R2\le |x|\le 2R}\left( \aligned
    &\nabla u\\
    &\nabla \v v
   \endaligned
  \right)
  \Big\|_p,
$$
where we fix some $p$ such that $2+\delta<p<2(2+\delta)$. Here
$$
\chi_{\frac R2 \le |x| \le 2R} = \chi_{|x| \le 2R} - \chi_{|x| \le \frac R2}.
$$

 We first show that
\begin{align}\label{J18_ee0a}
A_R\lesssim \frac 1R,  \qquad \mbox { for }  R\ge R_0,
\end{align}
and $R_0$ is sufficiently large.

\underline{Linear flow estimate}. For $0\le t\le 1$, by Lemma \ref{lem_rbj_3} and
Lemma \ref{lem_J15_2}, we have
\begin{align*}
&\Big\|\chi_{\frac R2\le |x|\le 2R}\frac{|\nabla|}{\langle
  \nabla\rangle}e^{it\langle \nabla\rangle}h_0
  \Big\|_p + \Big\|\chi_{\frac R2\le |x|\le 2R}\frac{\nabla}{|
  \nabla|}e^{it\langle \nabla\rangle}h_0
  \Big\|_p\\
\lesssim &\frac 1R \Bigl( \Big\|x\frac{|\nabla|}{\langle
  \nabla\rangle}e^{it\langle \nabla\rangle}h_0
  \Big\|_p + \Big\|x\frac{\nabla}{|
  \nabla|}e^{it\langle \nabla\rangle}h_0
  \Big\|_p \Bigr) \\
\lesssim &\frac 1R \Big(\Big\| |\nabla|^{-1} e^{it\langle \nabla\rangle}h_0
  \Big\|_p+\langle t\rangle \|e^{it\langle \nabla\rangle}h_0
  \|_p+\|e^{it\langle \nabla\rangle}(xh_0)
  \|_p\Big)\\
\lesssim &\frac 1R\langle
  t\rangle^{|1-\frac2p|}\Big(\| |\nabla|^{-1} \jnab^{|1-\frac 2p|+} h_0\|_p+\langle t\rangle
  \|\langle \nabla\rangle h_0\|_p+ \|\langle
  \nabla\rangle (xh_0)\|_p\Big)\\
\lesssim &\frac 1R \Big(\|\langle x\rangle h_0\|_{2+\delta}+
  \| h_0\|_{H^3}+\|x\Delta h_0\|_{2+\delta}\Big).
\end{align*}

Similarly,
\begin{align*}
&\Big\|\chi_{\frac R2\le |x|\le 2R}\frac{|\nabla|}{\langle
  \nabla\rangle}\nabla e^{it\langle \nabla\rangle}h_0
  \Big\|_p + \Big\|\chi_{\frac R2\le |x|\le 2R}\frac{\nabla}{|
  \nabla|}\nabla e^{it\langle \nabla\rangle}h_0
  \Big\|_p\\
\lesssim &\frac 1R\Bigl(  \Big\|x\frac{|\nabla|}{\langle
  \nabla\rangle}\nabla e^{it\langle \nabla\rangle}h_0
  \Big\|_p +  \Big\|x\frac{\nabla}{|
  \nabla|}\nabla e^{it\langle \nabla\rangle}h_0
  \Big\|_p \Bigr) \\
\lesssim &\frac 1R \Big(\|e^{it\langle \nabla\rangle}h_0
  \Big\|_p+ \|\nabla e^{it\langle \nabla\rangle}h_0
  \|_p+\|\nabla e^{it\langle \nabla\rangle}(xh_0)\|_p
  \Big)\\
\lesssim &\frac 1R \Big(\|h_0\|_{H^3}+
  \|\nabla\langle \nabla\rangle^{(1-\frac 2p)+} (xh_0)\|_p\Big).
\end{align*}
Now note that by Sobolev embedding,
$$
\|\nabla \langle \nabla\rangle^{(1-\frac 2p)+} (xh_0)\|_{p}\lesssim
\|\nabla \langle \nabla\rangle^{(1-\frac 2p)+} \langle
\nabla\rangle^{\frac{2}{2+\delta}-\frac 2p} (xh_0)\|_{2+\delta}.
$$
Since
$$
2+\frac{2}{2+\delta}-\frac 4p<2,
$$
we get
$$
\|\nabla \langle \nabla\rangle^{(1-\frac 2p)+} (xh_0)\|_{p}\lesssim
\|h_0\|_{H^3}+\|x\Delta h_0\|_{2+\delta}.
$$
So the contribution from the linear flow $\lesssim \frac 1R$.

\underline{Nonlinear flow estimate.}

Denote
\begin{align*}
\mathcal {N}_u(t)&= \int_0^t e^{i(t-s)\langle \nabla\rangle}
\Big[-\nabla\cdot (u\v v)+\frac i2 \frac{-\Delta}{\langle
\nabla\rangle}(u^2+|\v v|^2)\Big]\,ds;\\
\mathcal {N}_{\v v}(t)&= \int_0^t e^{i(t-s)\langle \nabla\rangle}
\left[\frac{\nabla}{|\nabla|}\Big(\frac{\jnab}{|\nabla|}
\nabla\cdot(u\v v)\Big) -\frac i2\nabla(u^2+|\v v|^2)\right]\,ds.
\end{align*}

We discuss two cases.

\underline{Low frequency piece.}

First note that by using the finite speed propagation of the Klein-Gordon
propagators $\cos \tau \lnr, \frac{\sin \tau\lnr}{\lnr}$, we have for all
$0\le \tau \le1$ and $R\ge 100$,
\begin{align}\label{J18_ee0}
\chi_{\frac R2\le|\cdot |\le 2R} \cos\tau\lnr=\chi_{\frac R2\le|\cdot|\le 2R}
\cos\tau\lnr\big[\chi_{\frac 25 R\le|\cdot |\le \frac 52 R}\big];\notag\\
\chi_{\frac R2\le|\cdot |\le 2R} \frac{\sin \tau\lnr}{\lnr}=\chi_{\frac
R2\le|\cdot |\le 2R} \frac{\sin \tau\lnr}{\lnr}\big[\chi_{
\frac 25R \le|\cdot|\le \frac 52 R}\big].
\end{align}
Consider the operators
\begin{align*}
K_{<1}^{(1)}f&=\chi_{\frac
25 R \le|x|\le \frac 52 R}\nabla P_{<1}(\widetilde{\chi}f),\notag\\
K_{<1}^{(2)}f&=\chi_{\frac
25 R \le|x|\le \frac 52 R}\frac{\Delta}{\lnr} P_{<1}(\widetilde{\chi}f),\notag\\
K_{<1}^{(3)}f&=\chi_{\frac
25 R \le|x|\le \frac 52 R}\frac{\nabla}{|\nabla|}\frac{\nabla}{|\nabla|}\langle
\nabla\rangle P_{<1}(\widetilde{\chi}f),
\end{align*}
where $\widetilde{\chi}=\chi_{\le \frac R4}$ or $\chi_{\ge 4R}$. We
claim that
\begin{equation}\label{J18_ee1}
\|K_{<1}^{j}f\|_p\lesssim \frac 1R
\|f\|_{(\frac12+\frac1p)^{-1}},\quad \mbox{for any}\quad j=1,2,3.
\end{equation}
Indeed, we shall prove it for $j=3$ and
$\widetilde{\chi}=\chi_{\le \frac R4}$. The others are similar. For
any dyadic $N<1$, it is not difficult to check that for some $\Psi(\xi)=\phi_{\le 1}(\xi)-\phi_{\le \frac 12}(\xi)$
\begin{align*}
\left[\mathcal
{F}^{-1}\Big(\frac{\nabla}{|\nabla|}\frac{\nabla}{|\nabla|}\lnr P_N\Big)\right](z)&=\int
e^{i\xi\cdot z}\Psi(\frac\xi
N)\frac{\xi}{|\xi|}\frac{\xi}{|\xi|}\langle
\xi\rangle\,d\xi\\
&=N^2\int e^{i\xi\cdot Nz}\Psi(\xi
)\frac{\xi}{|\xi|}\frac{\xi}{|\xi|}\langle N\xi\rangle\,d\xi\\
&=N^2\tilde{\phi}(N, z),
\end{align*}
where $\tilde{\phi}\in C^\infty$ satisfies
$$
|\tilde{\phi}(N, z)|\lesssim_k \langle N z \rangle^{-k},\quad\mbox{ for
any } z\in \mathbb R^2,\, N<1.
$$
We then have
\begin{align*}
\|K^{(3)}_{< 1} f\|_p&\lesssim \sum\limits_{N<1} \|\chi_{\frac
25 R \le|x|\le \frac 52 R}\frac{\nabla}{|\nabla|}\frac{\nabla}{|\nabla|}\langle
\nabla\rangle P_{N}(\widetilde{\chi}f)\|_p\\
&\lesssim \sum\limits_{N<1}\langle NR\rangle^{-10} N
\|f\|_{(\frac12+\frac1p)^{-1}}\\
&\lesssim \frac1R \|f\|_{(\frac12+\frac1p)^{-1}}.
\end{align*}
This settles the estimate \eqref{J18_ee1}. By using \eqref{J18_ee0}
and \eqref{J18_ee1}, we have
\begin{align*}
&\|\chi_{\frac R2\le|x|\le 2R} P_{< 1}\mathcal
{N}_{u}\|_p+\|\chi_{\frac R2\le|x|\le 2R} P_{< 1}\mathcal {N}_{\v
v}\|_p\\
&\qquad +\|\chi_{\frac R2\le|x|\le 2R} \nabla P_{< 1}\mathcal
{N}_{u}\|_p+\|\chi_{\frac R2\le|x|\le 2R} \nabla P_{< 1}\mathcal
{N}_{\v v}\|_p\\
\le &\eta(T_0)\|\chi_{\frac R4\le|x|\le
4R}u^2\|_p+\eta(T_0)\|\chi_{\frac R4\le|x|\le 4R}|\v v|^2\|_p+\frac CR\\
\le &\eta(T_0)\Big(\|\chi_{\frac R4\le|x|\le 4R}u\|_p+\|\chi_{\frac
R4\le|x|\le 4R}\v v\|_p\Big)+\frac CR.
\end{align*}
Here $\eta(T_0) \to 0$ as we take $T_0\to 0$.

\underline{High frequency piece.}

By \eqref{J18_ee0} and a similar computation as in the low frequency
case, we have
\begin{align*}
&\|\chi_{\frac R2\le|x|\le 2R} P_{>1}\mathcal
{N}_{u}\|_p+\|\chi_{\frac R2\le|x|\le 2R} P_{>1}\mathcal {N}_{\v
v}\|_p\\
&\qquad +\|\chi_{\frac R2\le|x|\le 2R} \nabla P_{>1}\mathcal
{N}_{u}\|_p+\|\chi_{\frac R2\le|x|\le 2R} \nabla P_{>1}\mathcal
{N}_{\v v}\|_p\\
\le
&\eta(T_0)\Big[\left\|\lnr^{3}(\chi_{\frac R4\le|x|\le
4R}u\>\chi_{\frac R4\le|x|\le 4R}\v v)\right\|_2\\
&\qquad
+\left\|\lnr^{3} \Big[(\chi_{\frac R4\le|x|\le 4R}u)^2 \Bigr]\right\|_2+
\left\|\lnr^{3}\Big[ (\chi_{\frac R4\le|x|\le 4R}\v v)^2 \Big]\right\|_2\Big]\\
&\qquad +\sum\limits_{N>1}(NR)^{-10}N^4(\|u\|_{H^4}^2+\|\v v\|_{H^4}^2)\\
\le &\eta(T_0)\Big(\|\chi_{\frac R4\le|x|\le 4R}u\|_p+\|\chi_{\frac
R4\le|x|\le 4R}\v v\|_p\Big)+\frac CR.
\end{align*}
Collecting the estimates, we obtain
$$
A_R\lesssim \eta(T_0)\Big(\|\chi_{\frac R4\le|x|\le
4R}u\|_p +\|\chi_{\frac R4\le|x|\le 4R}\v v\|_p \Big) +\frac CR.
$$
Now denote
$$
a_m=\Big\|\chi_{2^{m-1}\le |x|\le 2^{m+1}}\left( \aligned
    & u\\
    & \v v
   \endaligned
  \right)
  \Big\|_p+\Big\|\chi_{2^{m-1}\le |x|\le 2^{m+1}}\left( \aligned
    &\nabla u\\
    &\nabla \v v
   \endaligned
  \right)
  \Big\|_p.
$$
Clearly by choosing $T_0$ sufficiently small, we have
\begin{equation}\label{J18_ee}
a_m\le \frac 18\big(a_{m-1}+a_m+a_{m+1}\big)+C \cdot 2^{-m}.
\end{equation}
Note that $a_m\lesssim 1$ for any $m$. Iterating \eqref{J18_ee}
gives us
$$
a_m\lesssim 2^{-m}.
$$
Therefore \eqref{J18_ee0a} is proved.

\texttt{Step 2.} We show
$$
\|x(1-\Delta)e^{-it\langle
\nabla\rangle}h(t)\|_{2+\delta}
\mbox{ is continuous in $t$.}
$$
We first prove that
\begin{equation}\label{J18_ee4}
\Big\|x\left( \aligned
    &u\\
    &\v v
   \endaligned
  \right)
  \Big\|_\infty\lesssim 1.
\end{equation}
This is equivalent to
$$
\Big\|\chi_{|x|\sim R}\left( \aligned
    &u\\
    &\v v
   \endaligned
  \right)
  \Big\|_\infty\lesssim \frac1R,\qquad \mbox{ for any }R\ge 100.
$$
From Step 1 and Sobolev embedding, we have
\begin{align*}
\Big\|\chi_{|x|\sim R}\left( \aligned
    &u\\
    &\v v
   \endaligned
  \right)
\Big\|_\infty
&\lesssim\Big\|\chi_{|x|\sim R}\left( \aligned
    &u\\
    &\v v
   \endaligned
  \right)
  \Big\|_p+\Big\|\nabla \Big[ \chi_{|x|\sim R}\left( \aligned
    &u\\
    &\v v
   \endaligned
  \right)\Big]
  \Big\|_p\\
&\lesssim \frac1R.
\end{align*}
Hence \eqref{J18_ee4} holds.

To continue we need a simple lemma.
\begin{lem}\label{lem-J18_ee5}
For any $s\ge 0$,
\begin{equation}\label{J18_ee5}
\|x\lnr^s(fg)\|_{2+\delta}\lesssim
\|xf\|_\infty\>\|g\|_{H^{s+3}}+\|xg\|_\infty\>\|f\|_{H^{s+3}}
+\|f\|_{H^{s+3}}\>\|g\|_{H^{s+3}}.
\end{equation}
\end{lem}
\begin{proof}
We write
\begin{align*}
\widehat{\lnr^s(fg)}(\xi)&=\langle\xi\rangle^{s}\int\chi_{\frac{\langle\eta\rangle}{\langle\xi-\eta\rangle}\le
1}\hat f(\xi-\eta)\hat
g(\eta)\,d\eta+\langle\xi\rangle^{s}\int\chi_{\frac{\langle\eta\rangle}{\langle\xi-\eta\rangle}>
1}\hat f(\xi-\eta)\hat g(\eta)\,d\eta\\
=&\langle\xi\rangle^{s}\int\chi_{\frac{\langle\xi-\eta\rangle}{\langle\eta\rangle}\le
1}\hat f(\eta)\hat
g(\xi-\eta)\,d\eta+\langle\xi\rangle^{s}\int\chi_{\frac{\langle\eta\rangle}{\langle\xi-\eta\rangle}>
1}\hat f(\xi-\eta)\hat g(\eta)\,d\eta.
\end{align*}
Differentiating  in $\xi$ gives
\begin{align}
\mathcal  F\Big( {x\lnr^s(fg)} \Big)(\xi)=&O(\langle\xi\rangle^{s-1})\int\chi_{\frac{\langle\xi-\eta\rangle}{\langle\eta\rangle}\le
1}\hat f(\eta)\hat
g(\xi-\eta)\,d\eta\label{J18_ee6a}\\
&+\langle\xi\rangle^{s}\int\partial_{\xi}\chi_{\frac{\langle\xi-\eta\rangle}{\langle\eta\rangle}\le
1}\hat f(\eta)\hat
g(\xi-\eta)\,d\eta\label{J18_ee6b}\\
&+\langle\xi\rangle^{s}\int\chi_{\frac{\langle\xi-\eta\rangle}{\langle\eta\rangle}\le
1}\hat f(\eta)\widehat
{xg}(\xi-\eta)\,d\eta\label{J18_ee6c}\\
&+\langle\xi\rangle^{s}\int\chi_{\frac{\langle\eta\rangle}{\langle\xi-\eta\rangle}>
1}\widehat {xf}(\xi-\eta) \hat g(\eta)\,d\eta\label{J18_ee6d}\\
&+\cdots, \notag
\end{align}
where $``\cdots"$ denote similar terms.

It is not difficult to show that
\begin{align*}
\|\mathcal {F}^{-1}(\eqref{J18_ee6a})\|_{2+\delta}+\|\mathcal
{F}^{-1}(\eqref{J18_ee6b})\|_{2+\delta}\lesssim
\|f\|_{H^{s+3}}\>\|g\|_{H^{s+3}}.
\end{align*}
We shall only estimate \eqref{J18_ee6c}. The estimate of
\eqref{J18_ee6d} is similar. By Lemma \ref{lem_rbj_1}, we have
\begin{align*}
\|\mathcal {F}^{-1}(\eqref{J18_ee6c})\|_{2+\delta}
&\lesssim
\left\|T_{\chi_{\frac{\langle\xi-\eta\rangle}{\langle\eta\rangle}\le
1}\>\langle\xi\rangle^s\langle\eta\rangle^{-(s+2)-}}\big(\lnr^{(s+2)+}f,xg\big)\right\|_{2+\delta}\\
&\lesssim \|\lnr^{(s+2)+}f\|_{2+\delta}\>\|xg\|_\infty\\
&\lesssim \|xg\|_\infty\>\|f\|_{H^{s+3}}.
\end{align*}
The lemma is proved.
\end{proof}

By \eqref{e92_537b}, observe that
\begin{align*}
(1-\Delta)e^{-it\lnr }h& =(1-\Delta)h_0+\int_0^t e^{-is\lnr }(1-\Delta) \Bigl(- \frac{\langle \nabla \rangle \nabla}{|\nabla|} \cdot (u \v v)\\
&\qquad +
\frac i2 |\nabla| ( u^2 +|\v v|^2) \Bigr) \,ds.
\end{align*}
By Lemma \ref{lem-J18_ee5} and \eqref{J18_ee4}, we have
\begin{align*}
&\left\|x\Big((1-\Delta)e^{-it\lnr} h(t)\Big)-x(1-\Delta)h_0\right\|_{2+\delta}\\
\lesssim&
|t|\|u\|_{L^\infty_tH^6}\>\|\v v\|_{L^\infty_tH^6}+\int_0^t \|x\lnr^4(u\v v)\|_{2+\delta}\,ds\\
&\qquad+\int_0^t \big(\|x\lnr^4(u^2)\|_{2+\delta}+\|x\lnr^4(|\v v|^2)\|_{2+\delta}\big)\,ds\\
\lesssim &
|t|+|t|\Big(\|xu\|_\infty\big(\|u\|_{H^7}+\|\v v\|_{H^7}\big)\\
&\qquad+\|x\v v\|_\infty\big(\|u\|_{H^7}+\|\v v\|_{H^7}\big)+\big(\|u\|_{H^7}+\|\v v\|_{H^7}\big)^2\Big)\\
\lesssim & |t|.
\end{align*}
Clearly this gives continuity in $t$.

\subsection{$H^{N^{\prime}}$ estimate of $h$}

By \eqref{J22_49},  we decompose $f$ as
\begin{align}
\hat f(t,\xi) & = \widehat{h_0}(\xi) +
\int_0^t \int e^{-is\phi_0} \frac {\xi}{|\xi|}  \langle \xi \rangle \chi_{|\xi-\eta| \le \langle s \rangle^{10\delta}}
\chi_{|\eta| \le \langle s \rangle^{10\delta}}
\widehat{\mathcal R f}(s,\xi-\eta) \widehat{\mathcal R f}(s,\eta) d\eta ds \label{JL4_1a} \\
& \qquad + \int_0^t \int e^{-is \phi_0} \frac {\xi}{|\xi|}  \langle \xi \rangle
\chi_{|\xi-\eta| > \langle s \rangle^{10\delta}}
\chi_{|\eta| \le \langle s \rangle^{10\delta}}
\widehat{\mathcal R f}(s,\xi-\eta) \widehat{\mathcal R f}(s,\eta) d\eta ds \label{JL4_1b} \\
& \qquad + \int_0^t \int e^{-is \phi_0} \frac {\xi}{|\xi|}
\langle \xi \rangle \chi_{|\eta|> \langle s \rangle^{10\delta}}
\widehat{\mathcal Rf}(s,\xi-\eta) \widehat{\mathcal Rf}(s,\eta) d\eta ds.  \label{JL4_1c}
\end{align}

For \eqref{JL4_1c}, we compute
\begin{align*}
\| \mathcal F^{-1}  &( \eqref{JL4_1c} ) \|_{H^{N^{\prime}}} \notag \\
& \lesssim \int_0^t \| P_{> \langle s \rangle^{10\delta}} \mathcal R h(s) \cdot \mathcal
R h(s) \|_{H^{N^{\prime}+1}} ds \notag \\
& \lesssim \int_0^t \Bigl( \| P_{>\langle s \rangle^{10\delta}} h(s) \|_{H^{N^{\prime}+1}}
\cdot \| \mathcal R h(s) \|_{\infty}
+ \| h (s) \|_{H^{N^{\prime}+1}} \| P_{>\langle s \rangle^{10\delta}} \mathcal R h(s) \|_{\infty}
\Bigr) ds \notag \\
& \lesssim \int_0^t \langle s \rangle^{-5\delta}
\cdot \| h(s) \|_{H^{N^{\prime} + \frac 32}} \cdot \| \langle \nabla \rangle
h(s) \|_{\frac 1 {\delta}} ds \notag \\
& \lesssim \int_0^t \langle s \rangle^{-5\delta+\delta - (1-2\delta)} ds
\| h\|_{X_t}^2 \lesssim \| h\|_{ X_t}^2.
\end{align*}
Here we used the fact $N^{\prime}=N-\frac 32$.

Similarly
\begin{align*}
\| \mathcal F^{-1} ( \eqref{JL4_1b} ) \|_{H^{N^{\prime}}} \lesssim \| h \|_{ X_t}^2.
\end{align*}

For \eqref{JL4_1a}, we use the identity
\begin{align*}
e^{-is \phi_0} = \frac {i}{\phi_0} \frac {\partial}{\partial s } \left( e^{-is \phi_0} \right)
\end{align*}
to integrate by parts in $s$, and this gives
\begin{align}
\eqref{JL4_1a} & = i \int \frac{e^{-is \phi_0}}{\phi_0} \frac {\xi}{|\xi|}  \langle \xi \rangle
\chi_{|\xi-\eta| \le \langle s \rangle^{10\delta}}
\chi_{|\eta| \le \langle s \rangle^{10\delta}}
\widehat{\mathcal R f}(s,\xi-\eta) \widehat{\mathcal R f}(s,\eta) d\eta  \Biggr|_{s=0}^{s=t}
\label{JL4_2a} \\
& \quad - \int_0^t \int \frac{i e^{-is\phi_0}} {\phi_0} \frac {\xi}{|\xi|}  \langle \xi \rangle
\partial_{s} \bigl( \chi_{|\xi-\eta| \le \langle s \rangle^{10\delta}}
\chi_{|\eta| \le \langle s \rangle^{10\delta}}  \bigr) \widehat{\mathcal R f}(s,\xi-\eta)
\widehat{\mathcal R f}(s,\eta) d\eta ds \label{JL4_2b} \\
& \quad - \int_0^t \int \frac{i e^{-is\phi_0}} {\phi_0} \frac {\xi}{|\xi|}
\langle \xi \rangle \chi_{|\xi-\eta| \le \langle s \rangle^{10\delta}}
\chi_{|\eta| \le \langle s \rangle^{10\delta}}
\partial_s ( \widehat{\mathcal R f}(s,\xi-\eta)
\widehat{\mathcal R f}(s,\eta) ) d\eta ds. \label{JL4_2c}
\end{align}

For \eqref{JL4_2a}, we have
\begin{align*}
\| \mathcal F^{-1} ( \eqref{JL4_2a} ) \|_{H^{N^{\prime}}}
& \lesssim \| T_{\frac 1 {\phi_0} } ( P_{\le 1} \mathcal R h_0, P_{\le 1} \mathcal R h_0 ) \|_{H^{N^{\prime}+1}} \notag \\
& \qquad + \| T_{\frac 1 {\phi_0}} ( P_{\le \langle t \rangle^{10\delta}} \mathcal R h(t), P_{\le \langle t \rangle^{10\delta}}
\mathcal R h(t) ) \|_{H^{N^{\prime}+1}} \notag \\
& \lesssim \| h_0\|_2^2 + \| P_{\le \langle t \rangle^{10\delta}} h(t) \|_{H^{N^{\prime}+4+\delta}} \cdot
\| \mathcal R h(t) \|_{\infty} \notag \\
& \lesssim \| h_0 \|_2^2 + \langle t \rangle^{40\delta-(1-2\delta)} \| h \|_{X_t}^2 \lesssim  \| h\|_{ X_t}^2.
\end{align*}

For \eqref{JL4_2b}, we note that
\begin{align*}
\partial_s ( \chi_{|\xi-\eta| \le \langle s \rangle^{10\delta}}
\chi_{|\eta| \le \langle s \rangle^{10\delta}} ) &=
\chi^{(1)}_{|\xi-\eta| \lesssim \langle s \rangle^{10\delta}} \cdot
\chi_{|\eta| \le \langle s \rangle^{10\delta}} \cdot \langle s
\rangle^{-1} \notag \\
& \quad+ \chi_{|\xi-\eta| \le \langle s \rangle^{10\delta}} \cdot
\chi^{(2)}_{|\eta| \lesssim \langle s \rangle^{10\delta}} \cdot \langle s
\rangle^{-1},
\end{align*}
where $\chi^{(1)}$, $\chi^{(2)}$ ar some modified cut-offs. Therefore
\begin{align*}
\| \mathcal F^{-1} ( \eqref{JL4_2b} ) \|_{H^{N^{\prime}}}
& \lesssim \int_0^t \langle s \rangle^{-1}
\cdot \| \langle \nabla \rangle^{N^{\prime}+4+\delta}
P_{\le \langle s \rangle^{10\delta}} h(s) \|_2 \cdot \| \mathcal R P_{\le \langle s \rangle^{10\delta}}
h(s) \|_{\infty} ds \notag \\
& \lesssim \int_0^t \langle s \rangle^{-1-(1-2\delta) +40\delta} ds \cdot \| h\|_{ X_t}^2 \notag \\
& \lesssim \| h \|^2_{ X_t}.
\end{align*}

For \eqref{JL4_2c}, we observe that (see \eqref{p10_J11_e2})
\begin{align*}
e^{is \langle \nabla \rangle} \mathcal F^{-1} \bigl( \partial_s ( \widehat{\mathcal R f}(s) ) \bigr)
= \mathcal R \langle \nabla \rangle ( \mathcal R h(s) \cdot \mathcal R h(s) ).
\end{align*}
Therefore
\begin{align*}
\| \mathcal F^{-1} ( \eqref{JL4_2c} ) \|_{H^{N^{\prime}}}
& \lesssim \int_0^t \langle s \rangle^{-2(1-2\delta) +50\delta} ds \cdot \| h \|_{ X_t}^3 \notag \\
& \lesssim \| h\|_{X_t}^3.
\end{align*}

\section{Estimates of the boundary term $g$} \label{sec_g100}
In this section we control the boundary term $g$ coming from integration by parts in the time variable $s$
(see \eqref{e27_111a}).

We have the following
\begin{prop}  \label{prop_e926_310}
\begin{align*}
\| \langle \tau \rangle (1-\Delta) e^{i \tau \jnab} g(\tau)\|_{L_{\tau}^\infty L_x^{\frac 1 {\delta}}([0,t])}
+ \|  x (1-\Delta) g (\tau)\|_{L_{\tau}^\infty L_x^{2+\delta}([0,t])} \lesssim \|h\|_{ X_t}^2.
\end{align*}
\end{prop}

By Proposition \ref{prop_e926_310} and Sobolev embedding, it is easy to show that
\begin{align*}
\| \langle \tau \rangle e^{i\tau \jnab} g(\tau) \|_{L_{\tau,x}^{\infty} ([0,t])}
\lesssim \| h\|_{X_t}^2.
\end{align*}

The rest of this section is devoted to the proof of Proposition \ref{prop_e926_310}.
We begin with a simple lemma.

\begin{lem} \label{lem-J15-1}
For any $1\le s^{\prime}\le 7$, $t\ge 0$, we have
\begin{align}
\| \jnab^{s^{\prime}} h(t) \|_{\frac{16}{s^{\prime}}} \lesssim \langle t \rangle^{-(1-\frac{s^{\prime}}8 -\delta)}
\| h \|_{X_t}. \label{JL10_e1a} \\
\end{align}
Similarly for any $1\le s^{\prime} \le 6$, $t\ge 0$, we have
\begin{align}
\| \jnab^{s^{\prime}} h(t) \|_{\frac {13} {s^{\prime}}} \lesssim \langle t\rangle^{-(1-\frac{s^{\prime}}{6.5})} \| h \|_{X_t}.
\label{JL10_e1b}
\end{align}
\end{lem}

\begin{proof}[Proof of Lemma \ref{lem-J15-1}]
Observe that by interpolation we have
\begin{align*}
\| \jnab^{s^{\prime}} P_{<1} h(t) \|_{\frac {16}{s^{\prime}}} \lesssim \| h(t) \|_{\frac {16}{s^{\prime}}} \lesssim \langle t \rangle^{-(1-\frac {s^{\prime}} 8)} \| h \|_{X_t}.
\end{align*}

On the other hand, for any dyadic $M\ge 1$,
\begin{align*}
\| \jnab^{s^{\prime}} P_M h(t) \|_{\frac{16}{s^{\prime}}} & \lesssim M^{-(1-\frac{s^{\prime}}8)}
\Big( M^8 \| P_M h(t) \|_2 \Big)^{\frac {s^{\prime}} 8} \Big( M \| P_M h(t) \|_{\infty} \Big)^{1 -\frac{s^{\prime}} 8} \notag \\
& \lesssim M^{-(1-\frac{s^{\prime}} 8)} \cdot \langle t \rangle^{-(1-\frac {s^{\prime}}8 - \delta)} \| h \|_{X_t}.
\end{align*}

Summing in $M$ gives \eqref{JL10_e1a}.

The estimate of \eqref{JL10_e1b} is similar except that we use $\|h(t) \|_{H^{6.5}} \lesssim 1$ for all $t\ge 0$.

\end{proof}

 We begin
with the estimate of $\| (1-\Delta) e^{it \jnab} g(t) \|_{\frac 1 {\delta}}$.  By
\eqref{p12_J11_e2}, Lemma \ref{lem_rbj_1}, Lemma \ref{lem_rbj_5} and Lemma \ref{lem-J15-1},  we have
\begin{align*}
\| (1-\Delta) e^{it \jnab} g(t) \|_{\frac 1 {\delta}}
& \lesssim \| T_{\frac{\langle \xi \rangle^3}{\phi_0}} (\mathcal R h(t), \mathcal R h(t) ) \|_{\frac 1 {\delta}}
\notag \\
& \lesssim \| \jnab^{5+\delta} \mathcal R h(t) \|_{\infty} \cdot \| \jnab \mathcal R h(t) \|_{\frac 1 {\delta}}
\notag \\
& \lesssim \; \| \jnab^6 h(t) \|_{\frac {13} 6} \cdot \| \jnab h(t) \|_{\frac 1 {\delta} } \notag \\
& \lesssim \frac 1 {\langle t \rangle} \| h \|_{X_{t}}^2.
\end{align*}

It remains to control $\| x (1-\Delta) g(t) \|_{2+\delta}$. By \eqref{p12_J11_e2}, we have
\begin{align*}
\| x(1-\Delta) g \|_{2+\delta} \lesssim \| x (1-\Delta) \mathcal R g_1 \|_{2+\delta}.
\end{align*}
Note that
\begin{align*}
\partial_{\xi} \left( \frac{\xi}{|\xi|}
\langle \xi \rangle^2 \widehat{g_1}(\xi) \right) \sim
\frac{\langle \xi \rangle^2}{|\xi|} \widehat{g_1}(\xi)
+ \frac {\xi}{|\xi|}
\langle \xi \rangle \widehat{g_1}(\xi) +
\frac{\xi}{|\xi|} \langle \xi \rangle^2
\widehat{x g_1}(\xi).
\end{align*}
Therefore by Lemma \ref{lem_J15_2},
\begin{align*}
& \| x (1-\Delta) \mathcal R g_1 \|_{2+\delta} \notag \\
\lesssim & \| |\nabla|^{-1} \langle \nabla \rangle^2 g_1\|_{2+\delta}
+ \| \jnab g_1 \|_{2+\delta} + \| \jnab^2 ( x g_1 ) \|_{2+\delta} \notag \\
\lesssim & \| g_1 \|_{H^2} + \| x \jnab^2 g_1 \|_2 + \| \jnab^2 (x g_1) \|_{2+\delta} \notag \\
\lesssim & \| g_1 \|_{H^2} + \| \jnab^{2+\delta} (x g_1) \|_2.
\end{align*}
It is easy to check that $\|g_1 \|_{H^2} \lesssim \| h \|_{X_t}^2$.  We only need to estimate
$ \jnab^{2+\delta} (x g_1)$. We decompose $g_1$ as
\begin{align}
\widehat{g_1}(t,\xi) & = \int e^{-it \phi_0} \frac{\langle \xi
\rangle}{\phi_0} \chi_{\frac{|\xi-\eta|}{\langle \eta \rangle} \le
1} \widehat{\mathcal R f}(t,\xi-\eta) \widehat{\mathcal R f}(t,\eta)
d\eta \label{JL7_e1a}\\
& \quad+ \int e^{-it \phi_0}
\frac{\langle \xi \rangle}{\phi_0}
\chi_{\frac{|\xi-\eta|}{\langle \eta \rangle} > 1}
\widehat{\mathcal R f}(t,\xi-\eta)
\widehat{\mathcal R f}(t,\eta) d\eta. \label{JL7_e1b}
\end{align}

We shall only estimate the contribution of \eqref{JL7_e1a}. The term
\eqref{JL7_e1b} can be dealt in the same way as \eqref{JL7_e1a} using the
change of variable $\eta \to \xi-\eta$.

Now we have
\begin{align}
 & \langle \xi \rangle^{2+\delta} \widehat{x g_1}(t,\xi) \notag \\
 = &\; (-it) \cdot
 \int \partial_{\xi} \phi_0 e^{-it \phi_0}
 \cdot \frac{\langle \xi \rangle^{3+\delta}}{\phi_0}
 \cdot \chi_{\frac{|\xi-\eta|}{\langle \eta \rangle} \le 1}
 \widehat {\mathcal R f}(t,\xi-\eta)
 \widehat{\mathcal R f}(t,\eta) d\eta \label{JL7_e2a} \\
 & \qquad + \int e^{-it \phi_0}
 \langle \xi \rangle^{2+\delta}
 \partial_{\xi} \biggl( \frac{\langle \xi \rangle}{\phi_0}
 \chi_{\frac{|\xi-\eta|}{\langle \eta \rangle} \le 1} \biggr)
 \widehat{\mathcal R f}(t,\xi-\eta) \widehat{\mathcal R f}(t,\eta) d\eta
 \label{JL7_e2b} \\
 & \qquad + \int e^{-it \phi_0} \frac{\langle \xi \rangle^{3+\delta}}{\phi_0}
 \chi_{\frac{|\xi-\eta|}{\langle \eta \rangle} \le 1}
 \partial_{\xi} ( \widehat{\mathcal R f}(t,\xi-\eta) ) \widehat{\mathcal R f}(t,\eta) d\eta
 \label{JL7_e2c} \\
 & \qquad + \cdots, \notag
 \end{align}
 where $\cdots$ denote similar terms.

By Lemma \ref{lem_rbj_1} and Lemma \ref{lem-J15-1},  we estimate \eqref{JL7_e2a} as
\begin{align*}
\| \mathcal F^{-1} ( \eqref{JL7_e2a} ) \|_2 &
\lesssim |t| \left\| T_{\frac{\langle \xi \rangle^{3+\delta}}{\phi_0}
\chi_{\frac{|\xi-\eta|}{\langle \eta \rangle} \le 1} \partial_{\xi} \phi_0}
(\mathcal R h(t), \mathcal R h(t) ) \right\|_2 \notag \\
& \lesssim \| \jnab^{5+2\delta} h(t) \|_{\frac{13}6} \cdot \| \jnab h(t) \|_{\frac{13}{0.5}} \notag\\
& \lesssim\; |t| \cdot \langle t \rangle^{-(1-\frac 6{6.5})} \cdot \langle t \rangle^{-\frac{6}{6.5}}
\| h\|_{X_t}^2
\notag \\
& \lesssim \| h\|_{X_t}^2.
\end{align*}

Similarly
\begin{align*}
\| \mathcal F^{-1} ( \eqref{JL7_e2b} ) \|_2 \lesssim \|h \|_{X_t}^2.
\end{align*}

For \eqref{JL7_e2c}, we note that by Lemma \ref{lem_rbj_3} and Lemma \ref{lem_J15_2},
\begin{align*}
 & \| \jnab^{2-20\delta} e^{it\jnab} \mathcal F^{-1} ( \partial_{\xi} (\widehat {\mathcal R f} ) ) \|_{2+2\delta}
 \notag \\
 \lesssim & \langle t \rangle^{\delta} \Bigl( \| \jnab^{2-19\delta} |\nabla|^{-1} f \|_{2+2\delta}
 + \| \jnab^{2-19\delta} \mathcal R (xf) \|_{2+2\delta} \Bigr)\notag \\
 \lesssim & \;  \langle t \rangle^{\delta} \Bigl( \| x (1-\Delta) f \|_{2+\delta} + \| f \|_{H^2} \Bigr).
 \end{align*}

Therefore
\begin{align*}
\| \mathcal F^{-1} ( \eqref{JL7_e2c} ) \|_{2}
&\lesssim \| h\|_{X_t} \cdot \langle t \rangle^{\delta} \cdot \| \jnab^{4+22\delta} h(t)
\|_{(\frac 12 -\frac 1{2+2\delta})^{-1}} \notag \\
& \lesssim \|h\|_{ X_t}^2.
\end{align*}

The proposition is proved.

\section{Reduction to low frequency}
In this section we control the high frequency part of the solution. The main result of
this section is
\begin{prop}
\begin{equation*}
\|e^{i\tau \lnr}f_{\text{cubic}} (\tau) \|_{\tilde X_t}\lesssim \|h\|_{X_t}^3+\|f_{\text{low}} (\tau)
\|_{L_{\tau}^{\infty} L_x^{2-\frac{\delta}{100}} ([0,t]) },
\end{equation*}
where
\begin{align}
\hat f_{\text{low}}(t,\xi) =&  \int_0^t \int e^{-is\phi} \frac{s
\>\partial_{\xi} \phi}{\phi_0(\xi,\eta)}\langle \xi
\rangle^{4+2\delta}\langle\eta\rangle \>
m_{\text{low}}(\xi,\eta,\sigma)\notag\\
&\quad \> \widehat{\mathcal R f}(s,\xi-\eta)\frac{\eta}{|\eta|}
\Bigl( \widehat{\mathcal Rf}(s,\eta-\sigma) \widehat{\mathcal R
f}(s, \sigma) \Bigr) d\sigma d\eta ds \label{sec6_J19_e1}
\end{align}
and
\begin{align*}
m_{\text{low}}(\xi,\eta,\sigma)  = \chi_{|\xi-\eta| \le \langle s
\rangle^{\delta_0}}\chi_{|\eta-\sigma| \le \langle s
\rangle^{\delta_0}} \chi_{|\sigma| \le \langle s
\rangle^{\delta_0}}.
\end{align*}
Here $\delta_0 =20\delta$.
\end{prop}

The rest of this section is devoted to the proof of this proposition.

\underline{Estimate of $\||\nabla|^\delta\lnr
(e^{it\lnr}f_{\text{cubic}})\|_\infty$ and $\|\lnr
(e^{it\lnr}f_{\text{cubic}})\|_{\frac 1 {\delta}} $.}

By using the dispersive inequality and noting that $f_{\text{cubic}}=const \cdot \mathcal R f_3$ (see \eqref{p12_J11_e1}), we have
\begin{align*}
&\||\nabla|^\delta\lnr (e^{it\lnr}f_{\text{cubic}}(t) )\|_\infty \notag \\
\lesssim
&\sum\limits_{M<1}M^\delta \|P_M e^{it\lnr}f_3 (t) \|_\infty+
\sum\limits_{M\ge1}M^{1+\delta} \|P_M e^{it\lnr}f_3(t) \|_\infty\\
\lesssim & \frac{1}{\langle t\rangle}\|f_3\|_1+\frac{1}{\langle
t\rangle}\sum\limits_{M\ge 1}M^{3+\delta}\|P_Mf_3\|_1\\
\lesssim &\frac{1}{\langle t\rangle}\|\lnr^{3+2\delta}f_3\|_1.
\end{align*}
Similarly
\begin{align*}
\|\lnr (e^{it\lnr}f_{\text{cubic}} (t) )\|_{\frac 1 {\delta}} \lesssim &\|\lnr
(e^{it\lnr}f_{3} (t) )\|_{\frac 1 {\delta}}  \\
\lesssim &\langle t\rangle^{-(1-2\delta)}\|\lnr^{3}f_3\|_{ (1-\delta)^{-1} }\\
\lesssim &\langle t\rangle^{-(1-2\delta )}\|\lnr^{3+2\delta}f_3\|_1.
\end{align*}
Since
$$
\|\lnr^{3+2\delta}f_3\|_1\lesssim \|\langle x\rangle
\big(\lnr^{3+2\delta}f_3\big)\|_{2-\frac{\delta}{100}},
$$
we obtain
\begin{align*}
&\||\nabla|^\delta\lnr (e^{it\lnr}f_{\text{cubic}} (t) )\|_\infty + \|\lnr
(e^{it\lnr}f_{\text{cubic}} (t) )\|_{\frac 1 {\delta}}  \notag \\
\lesssim &\|\langle
x\rangle(\lnr^{3+2\delta}f_3 (t) )\|_{2-\frac{\delta}{100}}.
\end{align*}

\underline{Estimate of $\|x (1-\Delta) f_{\text{cubic}}\|_{2+\delta}$.}

By Lemma \ref{lem_J15_2}, we have
\begin{align*}
&\| x (1-\Delta) f_{\text{cubic}} \|_{2+\delta} \notag \\
\lesssim &\; \| x \mathcal R \jnab^2 f_3 \|_{2+\delta} \notag \\
 \lesssim & \;\| |\nabla|^{-1} \jnab^2 f_3 \|_{2+\delta} + \| \jnab f_3 \|_{2+\delta}
+ \| \jnab^2 (x f_3 ) \|_{2+\delta} \notag \\
 \lesssim &\; \| |\nabla|^{-1} \jnab^{3+2\delta} f_3 \|_{2+\delta} +
\| \langle x \rangle \jnab^{3+2\delta} f_3 \|_{2-\frac{\delta}{100}} \notag \\
 \lesssim &\; \| \langle x \rangle \jnab^{3+2\delta} f_3 \|_{2 -\frac{\delta}{100}}.
\end{align*}

\underline{Estimate of $\|\langle x \rangle \lnr^{3+2\delta}f_3\|_{2-\frac{\delta}{100}}$.}

We shall only estimate $\| x  \lnr^{3+2\delta}f_3\|_{2-\frac{\delta}{100}}$. The estimate
of $\| \jnab^{3+2\delta} f_3 \|_{2-\frac {\delta}{100}}$ is simpler and omitted.

Observe that by \eqref{p12_J11_e1},
\begin{align*}
\mathcal F \left( { \lnr^{3+2\delta}f_3} \right)(\xi)
=&\int_0^t \int e^{-is\phi} \frac{1}{\phi_0(\xi,\eta)}\langle \xi
\rangle^{4+2\delta}\langle\eta\rangle \\
&\quad\widehat{\mathcal R f}(s,\xi-\eta)\frac{\eta}{|\eta|} \>
\Bigl( \widehat{\mathcal Rf}(s,\eta-\sigma) \widehat{\mathcal R
f}(s,\sigma) \Bigr) d\sigma d\eta ds.
\end{align*}
Differentiating in $\xi$ gives us
\begin{align}
&\mathcal F \left(  (-i){x \lnr^{3+2\delta}f_3 } \right) \notag \\
=&\;\partial_\xi\Bigl( \mathcal F \left(  { \lnr^{3+2\delta}f_3} \right) (\xi)  \Bigr)\notag\\
=&\;\int_0^t \int e^{-is\phi}\>(-is\partial_\xi \phi) \frac{1}{\phi_0(\xi,\eta)}\langle \xi
\rangle^{4+2\delta}\langle\eta\rangle \notag\\
&\qquad\widehat{\mathcal R f}(s,\xi-\eta)\frac{\eta}{|\eta|} \>
\Bigl( \widehat{\mathcal Rf}(s,\eta-\sigma) \widehat{\mathcal R
f}(s,\sigma) \Bigr) d\sigma d\eta ds\label{aa-J12-300a}\\
&\;+\int_0^t \int e^{-is\phi} \partial_\xi\Big(\frac{\langle \xi
\rangle^{4+2\delta}}{\phi_0(\xi,\eta)}\Big)\langle\eta\rangle \notag\\
&\qquad\widehat{\mathcal R f}(s,\xi-\eta)\frac{\eta}{|\eta|} \>
\Bigl( \widehat{\mathcal Rf}(s,\eta-\sigma) \widehat{\mathcal R
f}(s,\sigma) \Bigr) d\sigma d\eta ds\label{aa-J12-300b}\\
&\;+\int_0^t \int e^{-is\phi} \frac{\langle \xi
\rangle^{4+2\delta}}{\phi_0(\xi,\eta)}\langle\eta\rangle \notag\\
&\quad\partial_\xi\widehat{\mathcal R f}(s,\xi-\eta)\frac{\eta}{|\eta|} \>
\Bigl( \widehat{\mathcal Rf}(s,\eta-\sigma) \widehat{\mathcal R
f}(s,\sigma) \Bigr) d\sigma d\eta ds. \label{aa-J12-300c}
\end{align}
We first deal with \eqref{aa-J12-300a}. We have
\begin{align}
\eqref{aa-J12-300a}&=\int_0^t \int(-is)\> e^{-is\phi}\frac{\partial_\xi \phi}{\phi_0(\xi,\eta)}\langle \xi
\rangle^{4+2\delta}\langle\eta\rangle\chi_{|\xi-\eta|>\lsr} \notag\\
&\quad\quad\widehat{\mathcal R f}(s,\xi-\eta)\frac{\eta}{|\eta|} \>
\Bigl( \widehat{\mathcal Rf}(s,\eta-\sigma) \widehat{\mathcal R
f}(s,\sigma) \Bigr) d\sigma d\eta ds\label{aa-J12-301a}\\
&\quad+ \int_0^t \int  (-is)\>e^{-is\phi}\frac{\partial_\xi \phi}{\phi_0(\xi,\eta)}\langle \xi
\rangle^{4+2\delta}\langle\eta\rangle\chi_{|\xi-\eta|\le \lsr} \notag\\
&\quad\quad\widehat{\mathcal R f}(s,\xi-\eta)\frac{\eta}{|\eta|} \>
\Bigl( \widehat{\mathcal Rf}(s,\eta-\sigma) \widehat{\mathcal R
f}(s,\sigma) \Bigr) d\sigma d\eta ds.\label{aa-J12-301b}
\end{align}
For \eqref{aa-J12-301a}, we further decompose it as
\begin{align}
\eqref{aa-J12-301a}&=\int_0^t \int(-is)\> e^{-is\phi}\frac{\partial_\xi \phi}{\phi_0(\xi,\eta)}\langle \xi
\rangle^{4+2\delta}\langle\eta\rangle\chi_{\frac{\langle\eta\rangle}{\langle\xi-\eta\rangle}\le 1} \chi_{|\xi-\eta|>\lsr} \notag\\
&\quad\quad\widehat{\mathcal R f}(s,\xi-\eta)\frac{\eta}{|\eta|} \>
\Bigl( \widehat{\mathcal Rf}(s,\eta-\sigma) \widehat{\mathcal R
f}(s,\sigma) \Bigr) d\sigma d\eta ds\label{aa-J12-302a}\\
&\quad+ \int_0^t \int(-is)\> e^{-is\phi}\frac{\partial_\xi \phi}{\phi_0(\xi,\eta)}\langle \xi
\rangle^{4+2\delta}\langle\eta\rangle\chi_{\frac{\langle\eta\rangle}{\langle\xi-\eta\rangle}> 1} \chi_{|\xi-\eta|>\lsr} \notag\\
&\quad\quad\widehat{\mathcal R f}(s,\xi-\eta)\frac{\eta}{|\eta|} \>
\Bigl( \widehat{\mathcal Rf}(s,\eta-\sigma) \widehat{\mathcal R
f}(s,\sigma) \Bigr) d\sigma d\eta ds.\label{aa-J12-302b}
\end{align}

We estimate \eqref{aa-J12-302a} as
\begin{align}
&\left\|\mathcal F^{-1}\eqref{aa-J12-302a}\right\|_{2-\frac{\delta}{100}} \notag \\
&\lesssim \int_0^t s\>\Big\|e^{is\lnr}\left( T_{\frac{\partial_\xi \phi}{\phi_0(\xi,\eta)}\langle \xi
\rangle^{4+2\delta}\langle\eta\rangle\chi_{\frac{\langle\eta\rangle}{\langle\xi-\eta\rangle}\le 1} }\Big(\mathcal RP_{>\lsr}h, \mathcal R(\mathcal Rh \>\mathcal Rh)\Big)\right)\Big\|_{2-\frac{\delta}{100}}\,ds.
\end{align}

By Lemma \ref{lem_rbj_3}, the operator
$$
\|\lnr^{-\frac{\delta}{100}}e^{is\lnr}\|_{L^{2-\frac{\delta}{100}}_x\rightarrow L^{2-\frac{\delta}{100}}_x}\lesssim \langle s \rangle^{\frac{\delta}{100}}.
$$
Therefore by Lemma \ref{lem_rbj_1} and Lemma \ref{lem_rbj_5}, we have
\begin{align*}
&\Big\|e^{is\lnr}\left( T_{\frac{\partial_\xi \phi}{\phi_0(\xi,\eta)}\langle \xi
\rangle^{4+2\delta}\langle\eta\rangle\chi_{\frac{\langle\eta\rangle}{\langle\xi-\eta\rangle}\le 1} }\Big(\mathcal RP_{>\lsr}h, \mathcal R(\mathcal Rh \>\mathcal Rh)\Big)\right)\Big\|_{2-\frac{\delta}{100}}\\
=&\Big\|\lnr^{-\frac{\delta}{100}}e^{is\lnr}\left( T_{\frac{\partial_\xi \phi}{\phi_0(\xi,\eta)}\langle \xi
\rangle^{4+2\delta+\frac{\delta}{100}}\langle\eta\rangle\chi_{\frac{\langle\eta\rangle}{\langle\xi-\eta\rangle}\le 1} }\Big(\mathcal RP_{>\lsr}h, \mathcal R(\mathcal Rh \>\mathcal Rh)\Big)\right)\Big\|_{2-\frac{\delta}{100}}\\
\lesssim
& \langle s \rangle^{\frac{\delta}{100}}\Big\|
 T_{\frac{\partial_\xi \phi}{\phi_0(\xi,\eta)}\langle \xi
\rangle^{4+2\delta+\frac{\delta}{100}}\langle\eta\rangle
\chi_{\frac{\langle\eta\rangle}{\langle\xi-\eta\rangle}\le 1} }
\Big(\mathcal RP_{>\lsr}h, \mathcal R(\mathcal Rh \>\mathcal Rh)\Big)\Big\|_{2-\frac{\delta}{100}}\\
\lesssim & \langle s \rangle^{\frac{\delta}{100}}\Big\|T_{\frac{\partial_\xi \phi}{\phi_0(\xi,\eta)}\langle \xi
\rangle^{4+2\delta+\frac{\delta}{100}}\langle\eta\rangle\chi_{\frac{\langle\eta\rangle}{\langle\xi-\eta\rangle}\le 1} \langle\xi-\eta\rangle^{-(7+3\delta)}\langle\eta\rangle^{-1}}\Big(\lnr^{7+3\delta}\mathcal RP_{>\lsr}h, \\
&\quad\lnr\mathcal R(\mathcal Rh \>\mathcal Rh)\Big)\Big\|_{2-\frac{\delta}{100}}  \\
\lesssim & \langle s \rangle^{\frac{\delta}{100}}\|\lnr^{7+3\delta}
\mathcal RP_{>\lsr}
h\|_{(\frac1{2-\frac{\delta}{100}}-2\delta)^{-1}}\|\lnr\mathcal R(\mathcal Rh \>\mathcal Rh)\|_{\frac{1}{2\delta}}\\
\lesssim & \langle s \rangle^{\frac{\delta}{100}}\|\lnr^{7+7\delta} P_{>\lsr}h\|_{2}\|\lnr h \|_{\frac{1}{\delta}}^2\\
\lesssim & \langle s \rangle^{\frac{\delta}{100}}\langle s \rangle^{-\delta_0(N-7-7\delta)}
\langle s \rangle^{\delta}\>\langle s \rangle^{-2(1-2\delta)}\|h\|_{X_t}^3\\
\lesssim &\langle s \rangle^{-2+7\delta-\delta_0(N-7)}\|h\|_{X_t}^3 \notag \\
\lesssim &  \langle s \rangle^{-2-}\|h\|_{X_t}^3,
\end{align*}
where we have used the fact that $N>7$ and $7\delta<\delta_0(N-7)$.
This clearly implies that
\begin{align*}
\left\|\mathcal F^{-1}\eqref{aa-J12-302a}\right\|_{2-\frac{\delta}{100}}
&\lesssim \int_0^t\langle s \rangle^{-1-}\,ds \>\|h\|_{X_t}^3\\
&\lesssim \|h\|_{X_t}^3.
\end{align*}

Similarly
\begin{align*}
\left\|\mathcal F^{-1}\eqref{aa-J12-302b}\right\|_{2-\frac{\delta}{100}}
\lesssim& \int_0^t\langle s \rangle^{1+\frac{\delta}{100}}\|\lnr^{7+3\delta}
P_{>\lsr}(\mathcal Rh \>\mathcal Rh )\|_{(\frac1{2-\frac{\delta}{100}}-\delta)^{-1}}
\|\lnr h\|_{\frac{1}{\delta}}\,ds\\
\lesssim& \int_0^t\langle s \rangle^{-1-}\,ds \>\|h\|_{X_t}^3\\
\lesssim& \|h\|_{X_t}^3.
\end{align*}
Therefore
$$
\left\|\mathcal F^{-1}\eqref{aa-J12-301a}\right\|_{2-\frac{\delta}{100}}\lesssim \|h\|_{X_t}^3.
$$

For \eqref{aa-J12-301b}, we decompose it as
\begin{align}
\eqref{aa-J12-301b}
=&\int_0^t \int (-is)\>e^{-is\phi}\frac{\partial_\xi \phi}{\phi_0(\xi,\eta)}\langle \xi
\rangle^{4+2\delta}\langle\eta\rangle\chi_{|\xi-\eta|\le \lsr}
\chi_{|\eta-\sigma|\le \lsr}\chi_{|\sigma|\le \lsr}\notag\\
&\quad\quad\widehat{\mathcal R f}(s,\xi-\eta)\frac{\eta}{|\eta|} \>
\Bigl( \widehat{\mathcal Rf}(s,\eta-\sigma) \widehat{\mathcal R
f}(s,\sigma) \Bigr) d\sigma d\eta ds\label{aa-J13-1a}\\
&+\int_0^t \int  (-is)\>e^{-is\phi}\frac{\partial_\xi \phi}{\phi_0(\xi,\eta)}\langle \xi
\rangle^{4+2\delta}\langle\eta\rangle\chi_{|\xi-\eta|\le \lsr}
\chi_{|\eta-\sigma|\le \lsr}\chi_{|\sigma|> \lsr}\notag\\
&\quad\quad\widehat{\mathcal R f}(s,\xi-\eta)\frac{\eta}{|\eta|} \>
\Bigl( \widehat{\mathcal Rf}(s,\eta-\sigma) \widehat{\mathcal R
f}(s,\sigma) \Bigr) d\sigma d\eta ds\label{aa-J13-1b}\\
&+\int_0^t \int (-is)\>e^{-is\phi}\frac{\partial_\xi \phi}{\phi_0(\xi,\eta)}\langle \xi
\rangle^{4+2\delta}\langle\eta\rangle\chi_{|\xi-\eta|\le \lsr}
\chi_{|\eta-\sigma|> \lsr}\notag\\
&\quad\quad\widehat{\mathcal R f}(s,\xi-\eta)\frac{\eta}{|\eta|} \>
\Bigl( \widehat{\mathcal Rf}(s,\eta-\sigma) \widehat{\mathcal R
f}(s,\sigma) \Bigr) d\sigma d\eta ds.\label{aa-J13-1c}
\end{align}

The estimate of \eqref{aa-J13-1b} is similar to that in \eqref{aa-J12-302a}. We have
\begin{align*}
 &\left\|\mathcal F^{-1}\eqref{aa-J13-1b}\right\|_{2-\frac{\delta}{100}} \notag \\
\lesssim& \int_0^t\langle s \rangle^{1+\frac{\delta}{100}}\|\lnr h\|_{\frac{1}{\delta}}
\|\lnr^{7+3\delta}(P_{\le\lsr}\mathcal Rh \>P_{>\lsr}\mathcal Rh)\|_{(\frac1{2-\frac{\delta}{100}}-\delta)^{-1}}\,ds\\
\lesssim& \int_0^t\langle s \rangle^{1+\frac{\delta}{100}}\|\lnr h\|_{\frac{1}{\delta}}\Big(
\big\|\lnr^{7+3\delta}P_{\le\lsr}\mathcal Rh \big\|_{(\frac1{2-\frac{\delta}{100}}-2\delta)^{-1}}\|P_{>\lsr} h\|_{\frac{1}{\delta}}\\
&\quad+\big\|\lnr^{7+3\delta}P_{>\lsr}\mathcal Rh \big\|_{(\frac1{2-\frac{\delta}{100}}-2\delta)^{-1}}\|P_{\le\lsr} h\|_{\frac{1}{\delta}}\Big)\,ds\\
\lesssim& \int_0^t\langle s \rangle^{1+\frac{\delta}{100}}
\|\lnr h\|_{\frac{1}{\delta}}\Big(\big\|\lnr^{N}h \big\|_{2}\langle s \rangle^{-\delta_0}\|\lnr h\|_{\frac{1}{\delta}}\\
&\quad+\langle s \rangle^{-\delta_0(N-7-7\delta)}\big\|\lnr^N h\big\|_2\|\lnr h\|_{\frac{1}{\delta}}\Big)\,ds\\
\lesssim &\int_0^t\langle s \rangle^{1+\frac{\delta}{100}-(1-2\delta)}\big(\langle s\rangle^{-\delta_0}
+\langle s\rangle^{-\delta_0(N-7-7\delta)}\big)\langle s \rangle^\delta \langle
 s \rangle^{-(1-2\delta)}\,ds \|h\|_{X_t}^3\\
\lesssim& \int_0^t\langle s \rangle^{-1-}\,ds \>\|h\|_{X_t}^3\\
\lesssim& \|h\|_{X_t}^3.
\end{align*}

The estimate of \eqref{aa-J13-1c} is the same as \eqref{aa-J13-1b}, and we have
$$
\left\|\mathcal F^{-1}\eqref{aa-J13-1c}\right\|_{2-\frac{\delta}{100}}
\lesssim \|h\|_{X_t}^3.
$$

The piece \eqref{aa-J13-1a} is exactly in the form given by \eqref{sec6_J19_e1}.
 Hence we have finished the estimate of \eqref{aa-J12-301b} and consequently the estimate of \eqref{aa-J12-300a}.

We now estimate \eqref{aa-J12-300b}. Note that
\begin{align*}
\partial_\xi\Big(\frac{\langle \xi\rangle^{4+2\delta}}{\phi_0}\Big)&\sim \frac{\langle \xi\rangle^{3+2\delta}}{\phi_0}+\langle\xi\rangle^{4+2\delta}
\partial_\xi\Big(\frac{1}{\phi_0}\Big)\\
&\sim \langle\xi\rangle^{4+2\delta}\left[\frac{1}{\langle \xi\rangle}\>\frac{1}{\phi_0}+\partial_\xi\Big(\frac{1}{\phi_0}\Big)\right].
\end{align*}
By Lemma \ref{lem_rbj_5}, obviously
\begin{align*}
\left|\partial_\xi^\alpha\partial_\eta^\beta\left[\frac{1}{\langle \xi\rangle}\>\frac{1}{\phi_0}+\partial_\xi\Big(\frac{1}{\phi_0}\Big)\right]\right|
\lesssim_{\alpha,\beta} \min\{\langle\xi-\eta\rangle,\langle\eta\rangle\}, \quad \forall\, \xi,\eta\in\R^2.
\end{align*}
We then write
\begin{align*}
\partial_\xi\Big(\frac{\langle \xi\rangle^{4+2\delta}}{\phi_0}\Big)=&\chi_{\frac{\langle\xi-
\eta\rangle}{\langle\eta\rangle}\le 1}\partial_\xi\Big(\frac{\langle \xi\rangle^{4+2\delta}}{\phi_0}\Big)
+\chi_{\frac{\langle\xi-
\eta\rangle}{\langle\eta\rangle}> 1}\partial_\xi\Big(\frac{\langle \xi\rangle^{4+2\delta}}{\phi_0}\Big)\\
=:&\;\widetilde{m_1}(\xi,\eta)+\widetilde{m_2}(\xi,\eta).
\end{align*}
It is not difficult to check that the functions
\begin{align*}
&\langle \xi \rangle^{\delta} \tilde m_1(\xi,\eta) \langle \eta \rangle
\langle \eta \rangle^{-(6+3\delta)} \cdot \langle \eta \rangle^{-(1+\delta)} \cdot \langle \xi-\eta\rangle^{-1}, \notag \\
& \langle \xi \rangle^{\delta} \tilde m_2(\xi,\eta) \langle \eta \rangle \langle \xi-\eta\rangle^{-(6+3\delta)}
\cdot \langle \xi-\eta \rangle^{-(1+\delta)} \cdot \langle \eta \rangle^{-1},
\end{align*}
satisfy \eqref{rbj_1e1}. Therefore, by Lemma \ref{lem_rbj_1}, we have
\begin{align*}
&\left\|\mathcal F^{-1}\eqref{aa-J12-300b}\right\|_{2-\frac{\delta}{100}}\\
\lesssim&\int_0^t \langle s \rangle^{\frac{\delta}{100}}\Big\|
T_{\langle \xi \rangle^{\delta} \tilde m_1(\xi,\eta) \langle \eta \rangle
\langle \eta \rangle^{-(6+3\delta)} \cdot \langle \eta \rangle^{-(1+\delta)} \cdot \langle \xi-\eta\rangle^{-1} }
\Big(\lnr \mathcal Rh, \mathcal R\lnr^{7+4\delta}(\mathcal Rh \>\mathcal Rh)\Big)\Big\|_{2-\frac{\delta}{100}} ds\\
&\quad +\int_0^t \langle s \rangle^{\frac{\delta}{100}}\Big\|
T_{\langle \xi \rangle^{\delta} \tilde m_2(\xi,\eta) \langle \eta \rangle \langle \xi-\eta\rangle^{-(6+3\delta)}
\cdot \langle \xi-\eta \rangle^{-(1+\delta)} \cdot \langle \eta \rangle^{-1}}
\Big(\lnr^{7+4\delta} \mathcal Rh, \mathcal R\lnr(\mathcal Rh \>\mathcal Rh)\Big)\Big\|_{2-\frac{\delta}{100}} ds\\
\lesssim& \int_0^t\langle s \rangle^{\frac{\delta}{100}}\|\lnr h\|_{\frac{1}{\delta}}\|\lnr^{7+4\delta}
(\mathcal Rh \>\mathcal Rh)\|_{(\frac1{2-\frac{\delta}{100}}-\delta)^{-1}}\,ds\\
&\quad+\int_0^t\langle s \rangle^{\frac{\delta}{100}}\|\lnr^{7+4\delta} h\|_{(\frac1{2-\frac{\delta}{100}}-2\delta)^{-1}}\|\lnr(\mathcal Rh \>\mathcal Rh)\|_{\frac{1}{2\delta}}\,ds\\
\lesssim&\int_0^t \langle s\rangle^{-1-}\,ds\>\|h\|_{X_t}^3
\lesssim \|h\|_{X_t}^3.
\end{align*}
Finally we estimate \eqref{aa-J12-300c}. We decompose it as
\begin{align}
\eqref{aa-J12-300c}=&\int_0^t \int e^{-is\phi} \frac{\langle \xi
\rangle^{4+2\delta}}{\phi_0(\xi,\eta)}\langle\eta\rangle \chi_{|\xi-\eta|\le \lsr}\notag\\
&\quad\partial_\xi\widehat{\mathcal R f}(s,\xi-\eta)\frac{\eta}{|\eta|} \>
\Bigl( \widehat{\mathcal Rf}(s,\eta-\sigma) \widehat{\mathcal R
f}(s,\sigma) \Bigr) d\sigma d\eta ds\label{J15_e1a}\\
&+\int_0^t \int e^{-is\phi} \frac{\langle \xi
\rangle^{4+2\delta}}{\phi_0(\xi,\eta)}\langle\eta\rangle \chi_{|\xi-\eta|> \lsr}\notag\\
&\quad\partial_\xi\widehat{\mathcal R f}(s,\xi-\eta)\frac{\eta}{|\eta|} \>
\Bigl( \widehat{\mathcal Rf}(s,\eta-\sigma) \widehat{\mathcal R
f}(s,\sigma) \Bigr) d\sigma d\eta ds.\label{J15_e1b}
\end{align}

For \eqref{J15_e1a}, we note that by Lemma \ref{lem_rbj_5} the function
$$
\widetilde{m}(\xi,\eta)=\frac{\langle \xi
\rangle^{4+3\delta}}{\phi_0(\xi,\eta)}\langle\eta\rangle \chi_{|\xi-\eta|\le \lsr}\>\langle\xi-\eta\rangle^{-(6+14\delta)}\langle\eta\rangle^{-(6+14\delta)}
$$
satisfies \eqref{rbj_1e1}. Therefore, by Lemma \ref{lem_rbj_3} and Lemma \ref{lem_rbj_1}, we have
\begin{align*}
&\left\|\mathcal F^{-1}\eqref{J15_e1a}\right\|_{2-\frac{\delta}{100}}\\
\lesssim&\int_0^t \langle s \rangle^{\frac{\delta}{100}}\Big\|
T_{\widetilde{m}(\xi,\eta)}\Big(\lnr^{6+4\delta} P_{\le \lsr}
e^{is\lnr}\mathcal F^{-1}\big(\partial_\xi(\widehat{\mathcal Rf})\big),\\
&\qquad\lnr^{6+4\delta}\mathcal R(\mathcal Rh \>\mathcal Rh)\Big)\Big\|_{2-\frac{\delta}{100}}ds \\
\lesssim& \int_0^t\langle s \rangle^{\frac{\delta}{100}}\left\|\lnr^{6+4\delta} P_{\le \lsr}e^{is\lnr}
\left(\mathcal F^{-1}\big(\partial_\xi(\widehat{\mathcal Rf})\big)\right)\right\|_{(\frac1{2-\frac{\delta}{100}}-2\delta)^{-1}}\\
&\qquad\cdot\big\|\lnr^{6+4\delta}(\mathcal Rh \>\mathcal Rh)\big\|_{\frac{1}{2\delta}}\,ds.
\end{align*}

To continue we need a lemma.

\begin{lem}\label{lem-J19-1}
For any dyadic $M\ge 1$, and $2+\delta<p<\infty$, we have
$$
\left\|P_{<M} e^{it\lnr}\mathcal {F}^{-1}\big(\partial_\xi(\mathcal
{R}f)\big)\right\|_p\lesssim M^{1+\frac{2}{2+\delta}-\frac
4p}\langle t\rangle^{1-\frac2p}\|\langle x\rangle f\|_{2+\delta}.
$$
\end{lem}
\begin{proof}[Proof of Lemma \ref{lem-J19-1}]
By Lemma \ref{lem_rbj_3} and Lemma \ref{lem_J15_2}, we have
\begin{align*}
\left\|P_{<M} e^{it\lnr}\mathcal {F}^{-1}\big(\partial_\xi(\mathcal
{R}f)\big)\right\|_p &\lesssim M^{1-\frac
2p}\langle t\rangle^{1-\frac2p}\big(\||\nabla|^{-1}f\|_p+\|P_{<M}(xf)\|_p\big)\\
&\lesssim M^{1-\frac
2p}\langle t\rangle^{1-\frac2p}\big(\|\langle x\rangle
f\|_{2+\delta}+M^{2(\frac 1{2+\delta}-\frac
1p)}\|xf\|_{2+\delta}\big)\\
&\lesssim M^{1+\frac{2}{2+\delta}-\frac 4p}\langle
t\rangle^{1-\frac2p}\|\langle x\rangle f\|_{2+\delta}.
\end{align*}
\end{proof}

By Lemma \ref{lem-J19-1}, we have
\begin{align*}
&\left\|\lnr^{6+4\delta} P_{\le \lsr}e^{is\lnr}\left(\mathcal F^{-1}
\big(\partial_\xi(\widehat{\mathcal Rf})\big)\right)\right\|_{(\frac1{2-\frac{\delta}{100}}-2\delta)^{-1}}\\
\lesssim & \; \langle s \rangle^{\delta_0(6+4\delta)} \langle s \rangle^{ \delta_0 ( 1+\frac 2{2+\delta}
-4(\frac 1 {2-\frac{\delta}{100}} -2\delta) )} \notag \\
& \qquad \cdot \langle s \rangle^{1-2(\frac 1{2-\frac{\delta}{100}} -2\delta)} \| \langle x \rangle f \|_{2+\delta}
\notag \\
\lesssim & \langle s \rangle^{7\delta_0+4\delta} \| \langle x \rangle f \|_{2+\delta}.
\end{align*}

By Sobolev embedding and Lemma \ref{lem-J15-1},
\begin{align*}
&\big\|\lnr^{6+4\delta}(\mathcal Rh \>\mathcal Rh)\big\|_{\frac{1}{2\delta}}\\
\lesssim&\|\lnr^{6+4\delta}h\|_{\frac1\delta}\|h\|_{\frac1\delta}\\
\lesssim&\|\lnr^7h\|_{\frac{16}{7}}\|h\|_{\frac1\delta}\\
\lesssim & \langle s \rangle^{-\frac18}\langle s \rangle^{-(1-2\delta)}
\langle s \rangle^{\delta} \|h\|_{X_t} ^2=\langle s \rangle^{-\frac98+3\delta}\|h\|_{X_t}^2.
\end{align*}
Therefore
\begin{align*}
\left\|\mathcal F^{-1}\eqref{J15_e1a}\right\|_{2-\frac{\delta}{100}}
\lesssim&\int_0^t \langle s \rangle^{\frac{\delta}{100}+4\delta+7\delta_0-\frac98+3\delta}\,ds\>\|h\|_{X_t}^3\\
\lesssim& \|h\|_{X_t}^3.
\end{align*}
For \eqref{J15_e1b}, we decompose it as
\begin{align}
\eqref{J15_e1b}=&\int_0^t \int e^{-is\phi} \frac{\langle \xi
\rangle^{4+2\delta}}{\phi_0(\xi,\eta)}\langle\eta\rangle \chi_{|\xi-\eta|> \lsr}\chi_{\frac{\langle\xi-\eta\rangle}{\langle\eta\rangle}\le 1}\notag\\
&\quad\partial_\xi\widehat{\mathcal R f}(s,\xi-\eta)\frac{\eta}{|\eta|} \>
\Bigl( \widehat{\mathcal Rf}(s,\eta-\sigma) \widehat{\mathcal R
f}(s,\sigma) \Bigr) d\sigma d\eta ds\label{J15_e1ba}\\
&+\int_0^t \int e^{-is\phi} \frac{\langle \xi
\rangle^{4+2\delta}}{\phi_0(\xi,\eta)}\langle\eta\rangle \chi_{|\xi-\eta|> \lsr}\chi_{\frac{\langle\xi-\eta\rangle}{\langle\eta\rangle}> 1}\notag\\
&\quad\partial_\xi\widehat{\mathcal R f}(s,\xi-\eta)\frac{\eta}{|\eta|} \>
\Bigl( \widehat{\mathcal Rf}(s,\eta-\sigma) \widehat{\mathcal R
f}(s,\sigma) \Bigr) d\sigma d\eta ds. \label{J15_e1bb}
\end{align}

For \eqref{J15_e1ba}, we note that by Lemma \ref{lem_rbj_5} the function
$$
\widetilde{m}(\xi,\eta)=\frac{\langle \xi
\rangle^{4+3\delta}}{\phi_0(\xi,\eta)}\langle\eta\rangle \chi_{|\xi-\eta|> \lsr}
\chi_{\frac{\langle\xi-\eta\rangle}{\langle\eta\rangle}\le 1}\>
\langle\xi-\eta\rangle^{-2+10\delta}\langle\eta\rangle^{-(6+14\delta)}
$$
satisfies \eqref{rbj_1e1}. Therefore by Lemma \ref{lem_rbj_3} and Lemma \ref{lem_rbj_1}, we have
\begin{align*}
&\left\|\mathcal F^{-1}\eqref{J15_e1ba}\right\|_{2-\frac{\delta}{100}}\\
\lesssim&\int_0^t \langle s \rangle^{\frac{\delta}{100}}\Big\| T_{\widetilde{m}(\xi,\eta)}\Big(\lnr^{2-10\delta}e^{is\lnr}\left(\mathcal F^{-1}\big(\partial_\xi(\widehat{\mathcal Rf})\big)\right),\\
&\qquad\lnr^{6+14\delta}\mathcal R(\mathcal Rh \>\mathcal Rh)\Big)\Big\|_{2-\frac{\delta}{100}} ds\\
\lesssim& \int_0^t\langle s \rangle^{\frac{\delta}{100}}\left\|\lnr^{2-10\delta}e^{is\lnr}\left(\mathcal F^{-1}\big(\partial_\xi(\widehat{\mathcal Rf})\big)\right)\right\|_{(\frac1{2-\frac{\delta}{100}}-2\delta)^{-1}}\\
&\qquad\cdot\big\|\lnr^{6+14\delta}(\mathcal Rh \>\mathcal Rh)\big\|_{\frac{1}{2\delta}}\,ds.
\end{align*}
By Lemma \ref{lem_rbj_3} and Lemma \ref{lem_J15_2}, we have
\begin{align*}
&\left\|\lnr^{2-10\delta}e^{is\lnr}
\left(\mathcal F^{-1}\big(\partial_\xi(\widehat{\mathcal Rf})\big)
\right)\right\|_{(\frac1{2-\frac{\delta}{100}}-2\delta)^{-1}}\\
\lesssim&\langle s \rangle^{4\delta}
\Big[\|\lnr^{2-6\delta}|\nabla|^{-1}f\|_{(\frac1{2-\frac{\delta}{100}}-2\delta)^{-1}}
+ \|\lnr^{2-6\delta}\mathcal (xf)\|_{(\frac1{2-\frac{\delta}{100}}-2\delta)^{-1}} \Big]\\
\lesssim &\langle s \rangle^{4\delta}\|\langle x\rangle\lnr^2f\|_{2+\delta}.
\end{align*}
On the other hand by Lemma \ref{lem-J15-1},
\begin{align*}
\big\|\lnr^{6+14\delta}(\mathcal Rh \>\mathcal Rh)\big\|_{\frac{1}{2\delta}}
\lesssim&\|\lnr^{6+14\delta}h\|_{\frac1\delta}\|h\|_{\frac1\delta}\\
\lesssim&\|\lnr^7h\|_{\frac{16}{7}}\|h\|_{\frac1\delta}\\
\lesssim & \langle s \rangle^{-\frac18}\langle s \rangle^{-(1-2\delta)}
\langle s \rangle^{\delta} \|h\|_{X_t}^2 \notag \\
=& \; \langle s \rangle^{-\frac98+3\delta}\|h\|_{X_t}^2.
\end{align*}
Therefore
\begin{align*}
\left\|\mathcal F^{-1}\eqref{J15_e1ba}\right\|_{2-\frac{\delta}{100}}
\lesssim&\int_0^t \langle s \rangle^{\frac{\delta}{100}+4\delta-\frac98+3\delta}\,ds\>\|h\|_{X_t}^3\\
\lesssim& \;  \|h\|_{X_t}^3.
\end{align*}
For \eqref{J15_e1bb}, we use the identity
$$
\partial_\xi\big(\widehat{\mathcal R f}(s,\xi-\eta)\big)=-\partial_\eta\big(\widehat{\mathcal R f}(s,\xi-\eta)\big)
$$
to integrate by parts in $\eta$. This gives
\begin{align}
\eqref{J15_e1bb}=&\int_0^t\int (-is\partial_\eta \phi) e^{-is\phi}\frac{\langle\xi\rangle^{4+2\delta}}{\phi_0}\langle\eta\rangle\chi_{|\xi-\eta|>\lsr}\chi_{
\frac{\langle\xi-\eta\rangle}{\langle\eta\rangle}>1}\notag\\
&\qquad \widehat{\mathcal Rf}(s,\xi-\eta)\>\frac{\eta}{|\eta|}\Bigl( \widehat{\mathcal Rf}(s,\eta-\sigma) \widehat{\mathcal R
f}(s,\sigma) \Bigr) d\sigma d\eta ds\label{J15_e2a}\\
& +\int_0^t\int e^{-is\phi}\langle\xi\rangle^{4+2\delta}\partial_\eta\Big(\frac{1}{\phi_0}\langle\eta\rangle\chi_{|\xi-\eta|>\lsr}\chi_{
\frac{\langle\xi-\eta\rangle}{\langle\eta\rangle}>1}\Big)\notag\\
&\qquad \widehat{\mathcal Rf}(s,\xi-\eta)\>\frac{\eta}{|\eta|}\Bigl( \widehat{\mathcal Rf}(s,\eta-\sigma) \widehat{\mathcal R
f}(s,\sigma) \Bigr) d\sigma d\eta ds\label{J15_e2b}\\
&+\int_0^t\int e^{-is\phi}\frac{\langle\xi\rangle^{4+2\delta}}{\phi_0}\langle\eta\rangle\chi_{|\xi-\eta|>\lsr}\chi_{
\frac{\langle\xi-\eta\rangle}{\langle\eta\rangle}>1}\notag\\
&\qquad \widehat{\mathcal Rf}(s,\xi-\eta)\>O\Bigl(\frac{1}{|\eta|}\Bigr)\Bigl( \widehat{\mathcal Rf}(s,\eta-\sigma) \widehat{\mathcal R
f}(s,\sigma) \Bigr) d\sigma d\eta ds\label{J15_e2c}\\
&+\int_0^t\int e^{-is\phi}\frac{\langle\xi\rangle^{4+2\delta}}{\phi_0}\langle\eta\rangle\chi_{|\xi-\eta|>\lsr}\chi_{
\frac{\langle\xi-\eta\rangle}{\langle\eta\rangle}>1}\notag\\
&\qquad \widehat{\mathcal Rf}(s,\xi-\eta)\>\frac{\eta}{|\eta|}\partial_\eta\Bigl(\widehat{\mathcal Rf}(s,\eta-\sigma) \widehat{\mathcal R
f}(s,\sigma) \Bigr) d\sigma d\eta ds.\label{J15_e2d}
\end{align}
The estimate of \eqref{J15_e2a} is exactly the same as that of \eqref{aa-J12-301a}.
The only change is that $\partial_\xi\phi$ is now replace by
$\partial_\eta\phi$. But in the estimates there only the boundedness of $\partial_\xi \phi$
 (and its derivatives) are used. Therefore we have
$$
\big\|\mathcal F^{-1}(\eqref{J15_e2a})\big\|_{2-\frac{\delta}{100}}\lesssim \|h\|_{X_t}^3.
$$

The estimate of \eqref{J15_e2b} is similar to the estimate of \eqref{aa-J12-300b}, and we have
$$
\big\|\mathcal F^{-1}(\eqref{J15_e2b})\big\|_{2-\frac{\delta}{100}}\lesssim \|h\|_{X_t}^3.
$$
For \eqref{J15_e2c}, we can decompose
$$
O(\frac{1}{|\eta|})=O(\frac{1}{|\eta|})\chi_{|\eta|<1}+O(\frac{1}{|\eta|})\chi_{|\eta|\ge1}.
$$
The piece corresponding to $O(\frac{1}{|\eta|})\chi_{|\eta|\ge1}$ is estimated in the same way as in \eqref{J15_e2b}. For the low frequency piece, we note that the function
$$
\widetilde{m}(\xi,\eta)=\frac{\langle\xi\rangle^{4+3\delta}}{\phi_0}\langle\eta\rangle\chi_{|\xi-\eta|>\lsr}\chi_{
\frac{\langle\xi-\eta\rangle}{\langle\eta\rangle}>1}\chi_{|\eta|<1}\cdot\langle\xi-\eta\rangle^{-(5+4\delta)}
$$
satisfies \eqref{rbj_1e1}. Therefore by Lemma \ref{lem_rbj_3} and Lemma \ref{lem_rbj_1}, we have
\begin{align*}
&\left\|\mathcal F^{-1}\eqref{J15_e2c}\right\|_{2-\frac{\delta}{100}}\\
\lesssim&\|h\|_{X_t}^3+
\int_0^t \langle s \rangle^{\frac{\delta}{100}}\Big\| T_{\widetilde{m}(\xi,\eta)}\Big(\lnr^{5+4\delta} P_{\gtrsim \lsr}\mathcal Rh,|\nabla|^{-1}(\mathcal Rh,\mathcal Rh)\Big)\Big\|_{2-\frac{\delta}{100}}\\
\lesssim&\|h\|_{X_t}^3+ \int_0^t\langle s \rangle^{\frac{\delta}{100}}\left\|\lnr^{5+4\delta} P_{\gtrsim \lsr}e^{is\lnr}h\right\|_2\\
&\qquad\cdot\big\||\nabla|^{-1}(\mathcal Rh \>\mathcal Rh)\big\|_{(\frac1{2-\frac{\delta}{100}}-\frac12)^{-1}}\,ds\\
\lesssim&\|h\|_{X_t}^3+ \int_0^t\langle s \rangle^{\frac{\delta}{100}-\delta_0}\|h\|_{H^{N^{\prime}}}\>\|\mathcal Rh\>\mathcal Rh\|_{2-\frac{\delta}{100}}\,ds\\
\lesssim&\|h\|_{X_t}^3+ \int_0^t\langle s \rangle^{\frac{\delta}{100}-\delta_0}\langle s \rangle^{-(1-2\delta)}\,ds\>\|h\|_{X_t}^3\\
\lesssim&\|h\|_{X_t}^3.
\end{align*}
Finally, we deal with \eqref{J15_e2d}. We decompose it further as
\begin{align}
\eqref{J15_e2d}=&\int_0^t\int e^{-is\phi}\frac{\langle\xi\rangle^{4+2\delta}}{\phi_0}\langle\eta\rangle\chi_{|\xi-\eta|>\lsr}\>\chi_{
\frac{\langle\xi-\eta\rangle}{\langle\eta\rangle}>1}\notag\\
&\qquad \widehat{\mathcal Rf}(s,\xi-\eta)\>\frac{\eta}{|\eta|}\partial_\eta\Bigl(\chi_{\frac{\langle\sigma\rangle}{\langle\eta-\sigma\rangle}\le 1}\widehat{\mathcal Rf}(s,\eta-\sigma) \widehat{\mathcal R
f}(s,\sigma) \Bigr) d\sigma d\eta ds\label{J15_e3a}\\
&+\int_0^t\int e^{-is\phi}\frac{\langle\xi\rangle^{4+2\delta}}{\phi_0}\langle\eta\rangle\chi_{|\xi-\eta|>\lsr}\chi_{
\frac{\langle\xi-\eta\rangle}{\langle\eta\rangle}>1}\notag\\
&\qquad \widehat{\mathcal Rf}(s,\xi-\eta)\>\frac{\eta}{|\eta|}\partial_\eta\Bigl(\chi_{\frac{\langle\sigma\rangle}{\langle\eta-\sigma\rangle}> 1}\widehat{\mathcal Rf}(s,\eta-\sigma) \widehat{\mathcal R
f}(s,\sigma) \Bigr) d\sigma d\eta ds.\label{J15_e3b}
\end{align}
We only need to estimate \eqref{J15_e3a}. The piece \eqref{J15_e3b} can be estimated similarly after the change of variable $\sigma\mapsto \eta-\sigma$. Now
\begin{align}
\eqref{J15_e3a}=&\int_0^t\int e^{-is\phi}\frac{\langle\xi\rangle^{4+2\delta}}{\phi_0}\langle\eta\rangle\chi_{|\xi-\eta|>\lsr}\>\chi_{
\frac{\langle\xi-\eta\rangle}{\langle\eta\rangle}>1}\notag\\
&\qquad \widehat{\mathcal Rf}(s,\xi-\eta)\>\frac{\eta}{|\eta|}\Bigl(\partial_\eta\Big(\chi_{\frac{\langle\sigma\rangle}{\langle\eta-\sigma\rangle}\le 1}\Big)\widehat{\mathcal Rf}(s,\eta-\sigma) \widehat{\mathcal R
f}(s,\sigma) \Bigr) d\sigma d\eta ds\label{J15_e4a}\\
&+\int_0^t\int e^{-is\phi}\frac{\langle\xi\rangle^{4+2\delta}}{\phi_0}\langle\eta\rangle\chi_{|\xi-\eta|>\lsr}\chi_{
\frac{\langle\xi-\eta\rangle}{\langle\eta\rangle}>1}\notag\\
&\qquad \widehat{\mathcal Rf}(s,\xi-\eta)\>\frac{\eta}{|\eta|}\Bigl(\chi_{\frac{\langle\sigma\rangle}{\langle\eta-\sigma\rangle}\le 1}\partial_\eta\big[\widehat{\mathcal Rf}(s,\eta-\sigma)\big] \widehat{\mathcal R
f}(s,\sigma) \Bigr) d\sigma d\eta ds.\label{J15_e4b}
\end{align}

We first deal with \eqref{J15_e4a}. Note that the function
$$
\widetilde{m}(\xi,\eta)=\frac{\langle\xi\rangle^{4+3\delta}}{\phi_0}\langle\eta\rangle\>\chi_{|\xi-\eta|>\lsr}\>\chi_{
\frac{\langle\xi-\eta\rangle}{\langle\eta\rangle}>1}\>\langle\xi-\eta\rangle^{-(7+4\delta)}\langle\eta\rangle^{-1}
$$
satisfies \eqref{rbj_1e1}. Therefore by Lemma \ref{lem_rbj_3} and Lemma \ref{lem_rbj_1}, we have
\begin{align*}
&\left\|\mathcal F^{-1}\eqref{J15_e4a}\right\|_{2-\frac{\delta}{100}}\\
\lesssim&\int_0^t \langle s \rangle^{\frac{\delta}{100}}\Big
\| T_{\widetilde{m}(\xi,\eta)}\Big(\lnr^{7+4\delta} \mathcal Rh,\lnr\mathcal RT_{\partial_\eta
\big(\chi_{\frac{\langle\sigma\rangle}{\langle\eta-\sigma\rangle}\le 1}\big)}(\mathcal Rh,\mathcal Rh)\Big)\Big\|_{2-\frac{\delta}{100}}\,ds\\
\lesssim&
\int_0^t\langle s \rangle^{\frac{\delta}{100}}
\|\lnr^{7+4\delta}h\|_{(\frac1{2-\frac{\delta}{100}}-2\delta)^{-1}}\>
\big\|\lnr T_{\partial_\eta\big(\chi_{\frac{\langle\sigma\rangle}{\langle\eta-\sigma\rangle}\le 1}\big)}
(\mathcal Rh,\mathcal Rh)\big\|_{\frac1{2\delta}}\,ds.
\end{align*}
Now note that
\begin{align*}
&\big\|\lnr T_{\partial_\eta\big(\chi_{\frac{\langle\sigma\rangle}{\langle\eta-\sigma\rangle}\le 1}\big)}
(\mathcal Rh,\mathcal Rh)\big\|_{\frac1{2\delta}}\\
\lesssim &
\big\| T_{\partial_\eta\big(\chi_{\frac{\langle\sigma\rangle}{\langle\eta-\sigma\rangle}\le 1}\big)}
(\mathcal Rh,\mathcal Rh)\big\|_{\frac1{2\delta}}+\big\|\nabla T_{\partial_\eta
\big(\chi_{\frac{\langle\sigma\rangle}{\langle\eta-\sigma\rangle}\le 1}\big)}(\mathcal Rh,\mathcal Rh)\big\|_{\frac1{2\delta}}\\
\lesssim &
\big\| T_{\partial_\eta\big(\chi_{\frac{\langle\sigma\rangle}{\langle\eta-\sigma\rangle}\le 1}\big)}
(\mathcal Rh,\mathcal Rh)\big\|_{\frac1{2\delta}}+\big\| T_{\partial_\eta
\big(\chi_{\frac{\langle\sigma\rangle}{\langle\eta-\sigma\rangle}\le 1}\big)}(\nabla\mathcal Rh,\mathcal Rh)\big\|_{\frac1{2\delta}}\\
&\quad+\big\| T_{\partial_\eta\big(\chi_{\frac{\langle\sigma\rangle}
{\langle\eta-\sigma\rangle}\le 1}\big)}(\mathcal Rh,\nabla\mathcal Rh)\big\|_{\frac1{2\delta}}.
\end{align*}
It is not difficult to check that
$$
\left|\partial_\eta^\alpha\partial_\sigma^\beta\left(\partial_\eta\Big(\chi_{\frac{\langle\sigma\rangle}{\langle\eta-\sigma\rangle}\le 1}\Big)\right)\right|\lesssim \big(\langle\eta\rangle+\langle\sigma\rangle\big)^{-(|\alpha|+|\beta|)}.
$$
Therefore, $\partial_\eta\Big(\chi_{\frac{\langle\sigma\rangle}{\langle\eta-\sigma\rangle}\le 1}\Big)$ is a standard Coifman-Meyer multiplier, and we have
$$
\big\|\lnr T_{\partial_\eta\big(\chi_{\frac{\langle\sigma\rangle}{\langle\eta-\sigma\rangle}\le 1}\big)}
(\mathcal Rh,\mathcal Rh)\big\|_{\frac1{2\delta}}\lesssim \|\lnr h\|_{\frac 1\delta}^2\lesssim \langle
 s \rangle^{-2(1-2\delta)}\|h\|_{X_t}^2.
$$
Hence,
\begin{align*}
\left\|\mathcal F^{-1}\eqref{J15_e4a}\right\|_{2-\frac{\delta}{100}}
\lesssim\int_0^t \langle s \rangle^{\frac{\delta}{100}+\delta}\langle s \rangle^{-2(1-2\delta)}\,ds\>\|h\|_{X_t}^3
\lesssim\|h\|_{X_t}^3.
\end{align*}
It remains to estimate \eqref{J15_e4b}. Note that the function
$$
\widetilde{m}(\xi,\eta)=\frac{\langle\xi\rangle^{4+2\delta}}{\phi_0}\langle\eta\rangle\>\chi_{|\xi-\eta|>\lsr}\>\chi_{
\frac{\langle\xi-\eta\rangle}{\langle\eta\rangle}>1}\>
\langle\xi-\eta\rangle^{-(6+15\delta)}\langle\eta\rangle^{-2+10\delta}
$$
satisfies \eqref{rbj_1e1}. Therefore by Lemma \ref{lem_rbj_3} and Lemma \ref{lem_rbj_1}, we have
\begin{align*}
&\left\|\mathcal F^{-1}\eqref{J15_e4b}\right\|_{2-\frac{\delta}{100}}\\
\lesssim&\int_0^t \langle s \rangle^{\frac{\delta}{100}}\Big\| T_{\widetilde{m}(\xi,\eta)}\Big(\lnr^{6+15\delta} \mathcal Rh,\lnr^{2-10\delta}\mathcal RT_{\chi_{\frac{\langle\sigma\rangle}{\langle\eta-\sigma\rangle}\le 1}}\\
&\qquad\left(e^{is\lnr}\left(\mathcal F^{-1}\big(\partial_\eta(\widehat{\mathcal Rf})\big)\right),\mathcal Rh\right)\Big)\Big\|_{2-\frac{\delta}{100}}\,ds\\
\lesssim&\int_0^t\langle s \rangle^{\frac{\delta}{100}}\|\lnr^{6+15\delta}h\|_{\frac1{\delta}}\\
&\qquad\>\big\|\lnr^{2-10\delta} T_{\chi_{\frac{\langle\sigma\rangle}{\langle\eta-\sigma\rangle}\le 1}}
\left(e^{is\lnr}\mathcal F^{-1}\big(\partial_\eta(\widehat{\mathcal Rf})\big),
\mathcal Rh\right)\big\|_{(\frac1{2-\frac{\delta}{100}}-\delta)^{-1}}\,ds.
\end{align*}

Now we make a Littlewood-Paley decomposition and write
\begin{align}
&\big\|\lnr^{2-10\delta} T_{\chi_{\frac{\langle\sigma\rangle}{\langle\eta-\sigma\rangle}\le 1}}\left(e^{is\lnr}\left(\mathcal F^{-1}\big(\partial_\eta(\widehat{\mathcal Rf})\big)\right),\mathcal Rh\right)\big\|_{(\frac1{2-\frac{\delta}{100}}-\delta)^{-1}}\notag\\
\lesssim &\big\|\lnr^{2-10\delta}
T_{\chi_{\frac{\langle\sigma\rangle}{\langle\eta-\sigma\rangle}\le 1}}\left(P_{<1}e^{is\lnr}\left(\mathcal F^{-1}\big(\partial_\eta(\widehat{\mathcal Rf})\big)\right),\mathcal Rh\right)\big\|_{(\frac1{2-\frac{\delta}{100}}-\delta)^{-1}}\label{J15_e5a}\\
&+\quad\!\!\sum\limits_{M\ge 1}\big\|\lnr^{2-10\delta}
 T_{\chi_{\frac{\langle\sigma\rangle}{\langle\eta-\sigma\rangle}\le 1}}
 \left(P_{M}e^{is\lnr}\left(\mathcal F^{-1}
 \big(\partial_\eta(\widehat{\mathcal Rf})\big)\right),
 \mathcal Rh\right)\big\|_{(\frac1{2-\frac{\delta}{100}}-\delta)^{-1}}.\label{J15_e5b}
\end{align}
For the low frequency piece \eqref{J15_e5a}, we note that by the cut-off
$\chi_{\frac{\langle\sigma\rangle}{\langle\eta-\sigma\rangle}\le 1}$ and $P_{<1}$,
$$
\langle \eta\rangle\le \langle \sigma\rangle+\langle\eta-\sigma\rangle\lesssim \langle\eta-\sigma\rangle\lesssim 1.
$$
Therefore, using the fact that $\chi_{\frac{\langle\sigma\rangle}{\langle\eta-\sigma\rangle}\le 1}$ is a Coifman-Meyer multiplier, we have
\begin{align*}
\eqref{J15_e5a}\lesssim&\big\| P_{\lesssim 1}
T_{\chi_{\frac{\langle\sigma\rangle}{\langle\eta-\sigma\rangle}\le 1}}\left(P_{<1}e^{is\lnr}
\left(\mathcal F^{-1}\big(\partial_\eta(\widehat{\mathcal Rf})\big)\right),\mathcal Rh\right)\big\|_{(\frac1{2-\frac{\delta}{100}}-\delta)^{-1}}\\
\lesssim&\big\|P_{<1}e^{is\lnr}\mathcal F^{-1}\big(\partial_\eta(\widehat{\mathcal Rf})\big)\big\|_{(\frac1{2-\frac{\delta}{100}}-2\delta)^{-1}}\|h\|_{\frac1\delta}\\
\lesssim&\langle s \rangle^{4\delta}\|\langle x\rangle f \|_{2+\delta}\langle s
\rangle^{-1+2\delta}\|h\|_{X_t},
\end{align*}
where in the last inequality, we have used Lemma \ref{lem-J19-1}.
Hence,
\begin{align*}
\eqref{J15_e5a}\lesssim\langle s \rangle^{-1+6\delta}\|h\|_{X_t}^2.
\end{align*}
For \eqref{J15_e5b}, thanks to the localization $P_M$ and $\chi_{\frac{\langle \sigma \rangle}{\langle \eta -\sigma\rangle} \le 1}$,
it follows easily that
$$
\langle \eta\rangle\le \langle \sigma\rangle+\langle\eta-\sigma\rangle\lesssim \langle\eta-\sigma\rangle\lesssim M.
$$
Therefore by Lemma \ref{lem-J19-1},
\begin{align*}
\eqref{J15_e5b}\lesssim&\sum\limits_{M\ge 1}M^{2-10\delta}\big\|
 T_{\chi_{\frac{\langle\sigma\rangle}{\langle\eta-\sigma\rangle}\le 1}}
 \left(P_{M}e^{is\lnr}\left(\mathcal F^{-1}\big(\partial_\eta(\widehat{\mathcal Rf})\big)\right),\mathcal Rh\right)\big\|_{(\frac1{2-\frac{\delta}{100}}-\delta)^{-1}}\\
\lesssim &\sum\limits_{M\ge 1}M^{2-10\delta}
\big\|P_{M}e^{is\lnr}\left(\mathcal F^{-1}\big(\partial_\eta(\widehat{\mathcal Rf})\big)\right)\big\|_{(\frac1{2-\frac{\delta}{100}}-2\delta)^{-1}}\|h\|_{\frac1\delta}\\
\lesssim &\sum\limits_{M\ge 1}M^{2-10\delta}\langle s \rangle^{4\delta}M^{4\delta}\left[\big\|P_{M}|\nabla|^{-1}f
\big\|_{(\frac1{2-\frac{\delta}{100}}-2\delta)^{-1}}+\|P_{M}(xf)\|_{(\frac1{2-\frac{\delta}{100}}-2\delta)^{-1}}\right]
\|h\|_{\frac1\delta}\\
\lesssim &\sum\limits_{M\ge 1}M^{-2\delta}\langle s \rangle^{-1+6\delta}
\left[\big\| \jnab^2 f\big\|_2+\|\lnr^2(xf)\|_{2+\delta}\right]
\|h\|_{X_t}\\
\lesssim&\langle s \rangle^{-1+6\delta}\|h\|_{X_t}^2.
\end{align*}

Collecting the estimates and using Lemma \ref{lem-J15-1}, we obtain
\begin{align*}
&\left\|\mathcal F^{-1}\eqref{J15_e4b}\right\|_{2-\frac{\delta}{100}}\\
\lesssim&\int_0^t \langle s \rangle^{\frac{\delta}{100}}\|\lnr^{6+15\delta}h\|_{\frac1\delta}\langle s
\rangle^{-1+6\delta}\,ds\>\|h\|_{X_t}^2\\
\lesssim&\int_0^t \langle s \rangle^{-1+\frac{\delta}{100}+6\delta}\|\lnr^7h\|_{\frac{16}{7}}\,ds\>\|h\|_{X_t}^2\\
\lesssim&\int_0^t \langle s \rangle^{-1+\frac{\delta}{100}+7\delta} \langle s\rangle^{-\frac18}\,ds\>
\|h\|_{X_t}^3\\
\lesssim &\|h\|_{X_t}^3.
\end{align*}

\section{Control of cubic interactions: the low frequency piece}
In the previous section, we controlled the high frequency part of
the cubic interaction term. In this section, we analyze in detail
the low frequency piece.  The main result of this section is the
following

\begin{prop}\label{prop-J19-prop1}
We have
$$
\|f_{\text{low}} (\tau) \|_{L_{\tau}^\infty L_x^{2-\frac{\delta}{100}} ([0,t])}\lesssim
\|h\|_{X_t}^3+\|h\|_{X_t}^4,
$$
where
\begin{align}
\hat f_{\text{low}}(t,\xi) =&  \int_0^t \int e^{-is\phi} \frac{s
(\partial_{\xi} \phi)}{\phi_0(\xi,\eta)}\langle \xi
\rangle^{4+2\delta}\langle\eta\rangle \>
m_{\text{low}}(\xi,\eta,\sigma)\notag\\
&\quad\frac{\eta}{|\eta|} \> \widehat{\mathcal R f}(s,\xi- \eta)
 \widehat{\mathcal Rf}(s,\eta-\sigma) \widehat{\mathcal R
f}(s,\sigma)  d\sigma d\eta ds \label{J19_e1}
\end{align}
and
\begin{align} \label{J19_e1aa}
m_{\text{low}}(\xi,\eta,\sigma)  = \chi_{|\xi-\eta| \le \langle s
\rangle^{\delta_0}}
\chi_{|\eta-\sigma| \le \langle s
\rangle^{\delta_0}} \chi_{|\sigma| \le \langle s
\rangle^{\delta_0}}.
\end{align}
\end{prop}
The rest of this section is devoted to the proof of this
proposition. The analysis will depend on the explicit form of the
phase function $\phi(\xi,\eta,\sigma)$. We discuss several cases.

\texttt{Case 1:}
\begin{align}
\phi(\xi,\eta,\sigma) = \langle \xi \rangle - \langle \xi -\eta
\rangle + \langle \eta -\sigma \rangle -\langle \sigma \rangle.
\label{e926_509}
\end{align}

By Lemma \ref{lem_phase}, we have
\begin{align*}
\partial_{\xi} \phi = Q_1(\xi,\eta) Q_2(\eta,\sigma) \partial_{\sigma} \phi,
\end{align*}
where
\begin{align}
&|\partial_{\xi}^{\alpha} \partial_{\eta}^{\beta} Q_1(\xi,\eta) | \lesssim_{\alpha,\beta} 1, \notag \\
&|\partial_{\eta}^{\alpha} \partial_{\sigma}^{\beta} Q_2(\eta,\sigma) | \lesssim_{\alpha,\beta}
\; \langle |\eta| +|\sigma| \rangle^3.  \label{J19_e1a}
\end{align}

We now write
\begin{align}
s \partial_{\xi} \phi e^{-is \phi} = i Q_1(\xi,\eta)
Q_2(\eta,\sigma) \partial_{\sigma}\Bigl(e^{-is \phi} \Bigr).
\label{J19_e2}
\end{align}

Plugging \eqref{J19_e2} into \eqref{J19_e1} and integrating by parts
in $\sigma$, we then obtain
\begin{align}
 & \hat f_{\text{low}}(t,\xi) \notag\\
 = & -i\int_0^t \int
 e^{-is\phi}  \frac{Q_1(\xi,\eta)}{\phi_0(\xi,\eta)}\partial_{\sigma} \Bigl( Q_2(\eta,\sigma) \chi_{|\eta-\sigma| \le \langle s
\rangle^{\delta_0}} \chi_{|\sigma| \le \langle s
\rangle^{\delta_0}} \Bigr)\notag \\
 & \qquad \cdot \langle \xi
\rangle^{4+2\delta}\langle\eta\rangle \> \chi_{|\xi-\eta| \le
\langle s \rangle^{\delta_0}} \widehat{\mathcal
Rf}(s,\xi-\eta) \frac{\eta}{|\eta|}
 \Bigl( \widehat{\mathcal R f}(s,\eta-\sigma) \widehat{\mathcal Rf}(s,\sigma) \Bigr) d\sigma d\eta ds \label{J19_e2a} \\
 & \quad - i\int_0^t \int e^{-is\phi}  \frac{Q_1(\xi,\eta)}{\phi_0(\xi,\eta)} Q_2(\eta,\sigma) \chi_{|\eta-\sigma| \le \langle s
\rangle^{\delta_0}} \chi_{|\sigma| \le \langle s
\rangle^{\delta_0}} \chi_{|\xi-\eta| \le
\langle s \rangle^{\delta_0}}\notag \\
 & \qquad \cdot \langle \xi
\rangle^{4+2\delta}\langle\eta\rangle \> \frac{\eta}{|\eta|} \widehat{\mathcal
Rf}(s,\xi-\eta)
 \bigl( \partial_{\sigma}\widehat{\mathcal R f}(s,\eta-\sigma)  \bigr)\>\widehat{\mathcal Rf}(s,\sigma)  d\sigma d\eta ds
 \label{J19_e2b}\\
& \quad - i \int_0^t \int e^{-is\phi} \frac{Q_1(\xi,\eta)}{\phi_0(\xi,\eta)}
Q_2(\eta,\sigma) \chi_{|\eta-\sigma| \le \langle s
\rangle^{\delta_0}} \chi_{|\sigma| \le \langle s
\rangle^{\delta_0}}\chi_{|\xi-\eta| \le
\langle s \rangle^{\delta_0}} \notag \\
 & \qquad \cdot \langle \xi
\rangle^{4+2\delta}\langle\eta\rangle \> \frac{\eta}{|\eta|} \widehat{\mathcal
Rf}(s,\xi-\eta)
  \widehat{\mathcal R f}(s,\eta-\sigma) \>\partial_{\sigma}\widehat{\mathcal Rf}(s,\sigma)
 d\sigma d\eta ds.
 \label{J19_e2c}
 \end{align}

We first estimate \eqref{J19_e2a}. By Lemma \ref{lem_rbj_3}, we have
\begin{align}
\|&\mathcal
{F}^{-1}(\eqref{J19_e2a})\|_{2-\frac{\delta}{100}}\lesssim \int^t_0
\langle
s\rangle^{\frac{\delta}{100}+\delta_0\frac{\delta}{100}+\delta_0(4+2\delta)} \notag \\
&\cdot\left\|T_{\frac{Q_1(\xi,\eta)}{\phi_0(\xi,\eta)}\>\langle
\eta\rangle}\Big( P_{\le \langle s \rangle^{\delta_0}}\mathcal
{R}h,\mathcal {R}T_{\partial_{\sigma}\big( Q_2(\eta,\sigma)
\chi_{|\eta-\sigma| \le \langle s \rangle^{\delta_0}}
\chi_{|\sigma| \le \langle s \rangle^{\delta_0}} \big)} (\mathcal
{R}h,\mathcal {R}h) \Big)\right\|_{2-\frac{\delta}{100}} ds. \label{J19_e2a_ta}
\end{align}
By \eqref{J19_e1a} and Lemma \ref{lem_rbj_5}, it is easy to check
that the functions
\begin{align*}
\widetilde{m_1}(\xi,\eta)&=\frac{Q_1(\xi,\eta)}{\phi_0(\xi,\eta)}
 \langle \eta\rangle
\langle\xi-\eta\rangle^{-2-\frac{\delta}{200}}\langle\eta\rangle^{-2-\frac{\delta}{200}};\\
\widetilde{m_2}(\eta,\sigma)&=\partial_{\sigma}\Bigl( Q_2(\eta,\sigma)
\chi_{|\eta-\sigma| \le \langle s \rangle^{\delta_0}}
\chi_{|\sigma| \le \langle s \rangle^{\delta_0}} \Bigr)
\langle\eta-\sigma\rangle^{-4-\frac{\delta}{200}}\langle\sigma\rangle^{-4-\frac{\delta}{200}}
\end{align*}
satisfy \eqref{rbj_1e1}. By Lemma \ref{lem_rbj_1}, we have
\begin{align*}
&\|T_{\frac{Q_1(\xi,\eta)}{\phi_0(\xi,\eta)}\>\langle
\eta\rangle}\big(P_{\le \langle s \rangle^{\delta_0}}\mathcal
{R}h,\mathcal {R}T_{\partial_{\sigma}\big( Q_2(\eta,\sigma)
\chi_{|\eta-\sigma| \le \langle s \rangle^{\delta_0}}
\chi_{|\sigma| \le \langle s \rangle^{\delta_0}} \big)} (\mathcal
{R}h,\mathcal {R}h)\big)\|_{2-\frac{\delta}{100}}\\
=& \|T_{\widetilde{m_1}(\xi,\eta)}\big(P_{\le \langle s
\rangle^{\delta_0}}\lnr^{2+\frac{\delta}{200}}\mathcal {R}h,P_{\lesssim
\langle s \rangle^{\delta_0}}\mathcal
{R}\lnr^{2+\frac{\delta}{200}}\\
&\quad T_{\widetilde{m_2}(\eta,\sigma)} (P_{\le \langle s
\rangle^{\delta_0}}\lnr^{4+\frac{\delta}{200}}\mathcal {R}h,P_{\lesssim
\langle s
\rangle^{\delta_0}}\lnr^{4+\frac{\delta}{200}}\mathcal {R}h)\big)\|_{2-\frac{\delta}{100}}\\
\lesssim &\|P_{\le \langle s
\rangle^{\delta_0}}\lnr^{2+\frac{\delta}{200}}\mathcal
{R}h\|_{(\frac1{2-\frac{\delta}{100}}-2\delta)^{-1}}\langle s
\rangle^{(2+\frac{\delta}{200})\delta_0}\>\|\lnr^{4+\frac{\delta}{200}}P_{\le
\lsr}h\|_{\frac1{\delta}}^2\\
\lesssim & \langle s \rangle^{(2+\frac{\delta}{200})\delta_0+ 2(3+\frac{\delta}{200})\delta_0 -2(1-2\delta)}
\| h \|_{X_t}^3 \notag \\
\lesssim & \langle s \rangle^{(8+\frac{\delta}{100} ) \delta_0 -2 +5\delta} \| h \|_{X_t}^3.
\end{align*}
Plugging the above estimate into \eqref{J19_e2a_ta}, we obtain
\begin{align*}
\|\mathcal
{F}^{-1}(\eqref{J19_e2a})\|_{2-\frac{\delta}{100}}&\lesssim \int_0^t
\langle s
\rangle^{ (12+\frac{\delta}{50} +2\delta) \delta_0 + (5+\frac 1{100} )\delta-2} ds
\|h\|_{X_t}^3\\
&\lesssim \|h\|_{X_t}^3.
\end{align*}
The estimate of \eqref{J19_e2b} is similar. By Lemma \ref{lem-J19-1}, we have
for some $\tilde m_3(\eta,\sigma)$ similar
to $\tilde m_2(\eta,\sigma)$,
\begin{align*}
& \|\mathcal
{F}^{-1}(\eqref{J19_e2b})\|_{2-\frac{\delta}{100}} \notag \\
\lesssim& \int^t_0
\langle
s\rangle^{\frac{\delta}{100}+
\delta_0\frac{\delta}{100}+\delta_0(4+2\delta)}\|T_{\widetilde{m_1}(\xi,\eta)}\big(P_{\le
\langle s \rangle^{\delta_0}}\lnr^{2+\frac{\delta}{200}}\mathcal
{R}h,P_{\lesssim \langle s \rangle^{\delta_0}}\\
&\quad\mathcal
{R}\lnr^{2+\frac{\delta}{200}}T_{\widetilde{m_3}(\eta,\sigma)}
(P_{\le \langle s
\rangle^{\delta_0}}\lnr^{4+\frac{\delta}{200}}e^{is\lnr}{\mathcal F}^{-1}\big(\partial_\sigma(\widehat{\mathcal
{R}f})\big),\\
&\quad P_{\le \langle s
\rangle^{\delta_0}}\lnr^{4+\frac{\delta}{200}}\mathcal
{R}h)\big)\|_{2-\frac{\delta}{100}}\,ds\\
\lesssim& \int_0^t \langle s
\rangle^{\frac{\delta}{100}+\delta_0\frac{\delta}{100}+(4+2\delta)\delta_0}
\|\lnr^{2+\frac{\delta}{200}}P_{\le
\lsr}h\|_{\frac1{\delta}}\\
&\quad\langle s \rangle^{(2+\frac{\delta}{200})\delta_0}\langle s
\rangle^{(4+\frac{\delta}{200})\delta_0}\|P_{\le \langle s
\rangle^{\delta_0}}e^{is\lnr}{\mathcal F}^{-1}\big(\partial_\sigma(\widehat{\mathcal
{R}f}\big) \big)\|_{(\frac1{2-\frac{\delta}{100}}-2\delta)^{-1}}\\
&\quad\|\lnr^{4+\frac{\delta}{200}}P_{\le
\lsr}h\|_{\frac1{\delta}}\,ds\\
\lesssim & \int_0^t\langle s
\rangle^{\frac{\delta}{100}+\delta_0(4+2\delta+\frac{\delta}{100})}
\langle s
\rangle^{    \delta_0(4+\frac{\delta}{100} ) -2(1-2\delta)}\langle s
\rangle^{ \delta_0 (6+\frac{\delta}{100} )}
\\
&\quad \langle s
\rangle^{[1+\frac2{2+\delta}-4(\frac{1}{2-\frac{\delta}{100}}-2\delta)]\delta_0+1-2(\frac{1}{2-\frac{\delta}{100}}-2\delta)}\,ds
\cdot \|h\|_{X_t}^3\\
\lesssim &\int_0^t\langle s\rangle^{-1-}\,ds\>\|h\|_{X_t}^3\lesssim
\|h\|_{X_t}^3.
\end{align*}
Similarly,
$$
\|\mathcal
{F}^{-1}(\eqref{J19_e2c})\|_{2-\frac{\delta}{100}}\lesssim\|h\|_{X_t}^3.
$$
This concludes Case 1.

\texttt{Case 2}:
\begin{align*}
\phi(\xi,\eta,\sigma) = \langle \xi \rangle - \langle \xi-\eta
\rangle -\langle \eta-\sigma\rangle +\langle \sigma \rangle.
\end{align*}
This is exactly the same as  Case 1 after the change of variable
$\sigma \to \eta -\sigma$.

\texttt{Case 3}:
\begin{align}
\phi(\xi,\eta, \sigma) = \langle \xi \rangle + \langle \xi-\eta
\rangle -\langle \eta-\sigma \rangle -\langle \sigma \rangle.
\label{J19_e4}
\end{align}

For this case, we will have to exploit some delicate cancelations of
the phases. Let $N_1=4$. We now introduce several frequency cut-offs and
write \eqref{J19_e1} as
\begin{align*}
\eqref{J19_e1}=& \sum_{i=1}^4 \int_0^t \int e^{-is\phi} \frac{s \>\partial_{\xi}
\phi}{\phi_0(\xi,\eta)}\langle \xi
\rangle^{4+2\delta}\langle\eta\rangle m_i(\xi,\eta,\sigma,s)\>
\notag\\
&\quad\frac{\eta}{|\eta|} \> \widehat{\mathcal R f}(s,\xi-\eta)
\Bigl( \widehat{\mathcal Rf}(s,\eta-\sigma) \widehat{\mathcal R
f}(s, \sigma) \Bigr) d\sigma d\eta ds\\
=:&\sum\limits_{i=1}^4 I_i,
\end{align*}
where
\begin{align*}
m_1(\xi,\eta,\sigma,s)&=\chi_{|\xi-\eta|\le \lsr}
\chi_{|\eta-\sigma|\le \lsr}\chi_{|\sigma|\le \lsr}\chi_{|\eta|\le
\langle s\rangle^{-\delta_0}}\chi_{|\xi|\le \langle
s\rangle^{-\frac{\delta_0}{N_1}}};\\
m_2(\xi,\eta,\sigma,s)&=\chi_{|\xi-\eta|\le \lsr}
\chi_{|\eta-\sigma|\le \lsr}\chi_{|\sigma|\le \lsr}\chi_{|\eta|\le
\langle s\rangle^{-\delta_0}}\chi_{|\xi|> \langle
s\rangle^{-\frac{\delta_0}{N_1}}}\chi_{|\sigma|\le 2\langle
s\rangle^{-{\delta_0}}};\\
m_3(\xi,\eta,\sigma,s)&=\chi_{|\xi-\eta|\le \lsr}
\chi_{|\eta-\sigma|\le \lsr}\chi_{|\sigma|\le \lsr}\chi_{|\eta|\le
\langle s\rangle^{-\delta_0}}\chi_{|\xi|> \langle
s\rangle^{-\frac{\delta_0}{N_1}}}\chi_{|\sigma|> 2\langle
s\rangle^{-{\delta_0}}};\\
m_4(\xi,\eta,\sigma,s)&=\chi_{|\xi-\eta|\le \lsr}
\chi_{|\eta-\sigma|\le \lsr}\chi_{|\sigma|\le \lsr}\chi_{|\eta|>
\langle s\rangle^{-\delta_0}}.
\end{align*}

\texttt{Subcase 3a}: estimate of $I_1$.

By \eqref{J19_e4}, we have
$$
\partial_\xi\phi=\frac{\xi}{\langle \xi\rangle}+\frac{\xi-\eta}{\langle
\xi-\eta\rangle}.
$$
Since on the support of $m_1(\xi,\eta,\sigma,s)$ both $\xi$ and $\eta$ are localized to low frequencies, we
gain one derivative by using the above identity. Therefore
\begin{align*}
\|\mathcal {F}^{-1}(I_1)\|_{2-\frac{\delta}{100}}&\lesssim
\int_0^t\langle s
\rangle^{\frac{\delta}{100}+1-\frac{\delta_0}{N_1}}\left\|P_{\lesssim \langle
s\rangle^{-\frac{\delta_0}{N_1}}}h\right\|_{(\frac1{2-\frac{\delta}{100}}-2\delta)^{-1}}\|P_{\le
\lsr} h\|_{\frac1\delta}^2 ds \\
&\lesssim \int_0^t \langle
s\rangle^{\frac{\delta}{100}+1-\frac{\delta_0}{N_1}-2(1-2\delta)}\,ds
\|h\|_{X_t}^3\\
&\lesssim \int_0^t\langle s \rangle^{-1-}\,ds \|h\|_{X_t}^3\lesssim
\|h\|_{X_t}^3,
\end{align*}
where we require that $\frac{\delta_0}{N_1}>4.01\delta$.

\texttt{Subcase 3b}: estimate of $I_2$.

Note that in this subcase we have $|\xi|\ge \langle
s\rangle^{-\frac{\delta_0}{N_1}},|\eta|\le \frac{25}{24} \langle
s\rangle^{-\delta_0},  |\sigma|\le 2\cdot \frac{25}{24} \cdot \langle s\rangle^{-{\delta_0}}$
on the support of $m_2(\xi,\eta,\sigma,s)$.
Hence
\begin{align}
\langle \xi \rangle +\langle \xi-\eta\rangle-2&\ge \langle \xi
\rangle-1=\frac{|\xi|^2}{\langle \xi \rangle+1}\gtrsim \langle
s\rangle^{-\frac{2\delta_0}{N_1}}, \quad \mbox{ if } |\xi|\le 3;\notag\\
\langle \xi \rangle +\langle \xi-\eta\rangle-2&\gtrsim \langle \xi-\eta\rangle, \quad \mbox{ if } |\xi|> 3;\notag\\
\langle \eta-\sigma \rangle +\langle \sigma\rangle-2&=\langle
\eta-\sigma \rangle-1+\langle \sigma\rangle-1\notag\\
&=\frac{(\eta-\sigma)\cdot(\eta-\sigma)}{\langle \eta-\sigma
\rangle+1}+\frac{\sigma\cdot\sigma}{\langle \sigma
\rangle+1}. \label{J19_e5}
\end{align}
We now perform a \textbf{partial normal form} transform. Namely, we
write
$$
e^{-is\phi}=e^{-is(\langle \xi \rangle+\langle
\xi-\eta\rangle-2)}\>e^{is(\langle \eta-\sigma \rangle +\langle
\sigma\rangle-2)}.
$$
Using the identity
$$
e^{-is(\langle \xi \rangle+\langle
\xi-\eta\rangle-2)}=\frac{i}{\langle \xi \rangle+\langle
\xi-\eta\rangle-2}\partial_s(e^{-is(\langle \xi \rangle+\langle
\xi-\eta\rangle-2)})
$$
and integrating by parts in the time variable $s$, we obtain
\begin{align}
I_2=& \int_0^t \int \frac{i}{\langle \xi \rangle+\langle
\xi-\eta\rangle-2}\partial_s(e^{-is(\langle \xi \rangle+\langle
\xi-\eta\rangle-2)})e^{is(\langle \eta-\sigma \rangle +\langle
\sigma\rangle-2)} \notag \\
&\quad\frac{s \>\partial_{\xi}
\phi}{\phi_0(\xi,\eta)}\langle \xi
\rangle^{4+2\delta}\langle\eta\rangle m_2(\xi,\eta,\sigma,s)\>
\notag\\
&\quad\frac{\eta}{|\eta|} \> \widehat{\mathcal R f}(s,\xi-\eta)
\widehat{\mathcal Rf}(s,\eta-\sigma) \widehat{\mathcal R
f}(s,\sigma)  d\sigma d\eta ds\notag\\
=&\int e^{-it\phi}\frac{i}{\langle \xi \rangle+\langle
\xi-\eta\rangle-2} \frac{t\> \partial_{\xi}
\phi}{\phi_0(\xi,\eta)}\langle \xi
\rangle^{4+2\delta}\langle\eta\rangle m_2(\xi,\eta,\sigma,t)\>
\notag\\
&\quad\frac{\eta}{|\eta|} \> \widehat{\mathcal R f}(t,\xi-\eta)
 \widehat{\mathcal Rf}(t,\eta-\sigma) \widehat{\mathcal R
f}(t,\sigma)  d\sigma d\eta \label{J19_e6a}\\
&- \int_0^t \int e^{-is\phi}\frac{i}{\langle \xi
\rangle+\langle \xi-\eta\rangle-2} \frac{ \partial_{\xi}
\phi}{\phi_0(\xi,\eta)}\langle \xi
\rangle^{4+2\delta}\langle\eta\rangle
\partial_s(sm_2(\xi,\eta,\sigma,s))\>
\notag\\
&\quad\frac{\eta}{|\eta|} \> \widehat{\mathcal R f}(s,\xi-\eta)
 \widehat{\mathcal Rf}(s,\eta-\sigma) \widehat{\mathcal R
f}(s, \sigma)  d\sigma d\eta ds\label{J19_e6b} \\
&+ \int_0^t \int e^{-is\phi}\frac{\langle \eta-\sigma \rangle+\langle
\sigma\rangle-2}{\langle \xi \rangle+\langle \xi-\eta\rangle-2}
\frac{s\> \partial_{\xi} \phi}{\phi_0(\xi,\eta)}\langle \xi
\rangle^{4+2\delta}\langle\eta\rangle m_2(\xi,\eta,\sigma,s)\>
\notag\\
&\quad\frac{\eta}{|\eta|} \> \widehat{\mathcal R f}(s,\xi-\eta)
 \widehat{\mathcal Rf}(s,\eta-\sigma) \widehat{\mathcal R
f}(s,\sigma)  d\sigma d\eta ds\label{J19_e6c}\\
&- \int_0^t \int e^{-is\phi}\frac{i}{\langle \xi \rangle+\langle \xi-\eta\rangle-2}
\frac{s \partial_{\xi} \phi}{\phi_0(\xi,\eta)}\langle \xi
\rangle^{4+2\delta}\langle\eta\rangle m_2(\xi,\eta,\sigma,s)\>
\notag\\
&\quad\frac{\eta}{|\eta|} \> \partial_s\widehat{\mathcal R f}(s,\xi-\eta)\>
 \widehat{\mathcal Rf}(s,\eta-\sigma) \widehat{\mathcal R
f}(s,\sigma)  d\sigma d\eta ds \label{J19_e6d}\\
&- \int_0^t \int e^{-is\phi}\frac{i}{\langle \xi \rangle+\langle \xi-\eta\rangle-2}
\frac{s \partial_{\xi} \phi}{\phi_0(\xi,\eta)}\langle \xi
\rangle^{4+2\delta}\langle\eta\rangle m_2(\xi,\eta,\sigma,s)\>
\notag\\
&\quad\frac{\eta}{|\eta|} \>\widehat{\mathcal R f}(s,\xi-\eta)\>
 \partial_s\big[\widehat{\mathcal Rf}(s,\eta-\sigma) \widehat{\mathcal R
f}(s,\sigma)\big]  d\sigma d\eta ds \label{J19_e6e}.
\end{align}
For \eqref{J19_e6a}, by using \eqref{J19_e5} and Lemma \ref{lem_rbj_5}, it is not difficult to check that the functions
\begin{align*}
\widetilde{m_1}(\xi,\eta)&=\frac{i}{\langle \xi \rangle+\langle
\xi-\eta\rangle-2} \frac{\partial_{\xi}
\phi}{\phi_0(\xi,\eta)}\langle \xi
\rangle^{4+2\delta+\frac{\delta}{100}}\langle\eta\rangle  \chi_{|\xi|\ge 3}\chi_{|\eta|\lesssim 1}\cdot \langle \xi-\eta\rangle^{-(4+2\delta+\frac{\delta}{99})},\\
\widetilde{m_2}(\xi,\eta)&=\frac{i}{\langle \xi \rangle+\langle
\xi-\eta\rangle-2} \frac{\partial_{\xi}
\phi}{\phi_0(\xi,\eta)}\langle \xi
\rangle^{4+2\delta+\frac{\delta}{100}}\langle\eta\rangle  \chi_{|\xi|< 3}\chi_{|\eta|\lesssim 1}\cdot \langle t\rangle^{-\frac{6\delta_0}{N_1}}
\chi_{\langle\xi\rangle+\langle\xi-\eta\rangle-2\gtrsim \langle t \rangle^{-\frac{2\delta_0}{N_1}}}
\end{align*}
satisfy \eqref{rbj_1e1}. Therefore by Lemma \ref{lem_rbj_1}, we have
\begin{align}
\|\mathcal F^{-1}(\eqref{J19_e6a})
\|_{2-\frac{\delta}{100}}\lesssim& \langle t
\rangle^{1+\frac{\delta}{100}}\Big\|T_{\widetilde{m_1}(\xi,\eta)}\Big(\lnr^{4+
2\delta+\frac{\delta}{99}}P_{\le \langle t\rangle^{\delta_0}}h,\notag\\
&\quad \mathcal RP_{\le \langle t\rangle^{-\delta_0}}
\big(\mathcal RP_{\le \langle t\rangle^{\delta_0}}
h \cdot \mathcal
RP_{\le \langle t\rangle^{\delta_0}}
P_{\le 2\langle t\rangle^{-\delta_0}}h\big)\Big)\Big\|_{2-\frac{\delta}{100}}\notag\\
&\quad +\langle t\rangle^{1+\frac{\delta}{100}+\frac{6\delta_0}{N_1}}\Big\|T_{\widetilde{m_2}(\xi,\eta)}\Big(P_{\le \langle t\rangle^{\delta_0}}P_{\lesssim 1}h,\notag\\
&\quad \mathcal RP_{\le \langle t\rangle^{-\delta_0}}
\big(\mathcal RP_{\le \langle t\rangle^{\delta_0}}
h \cdot \mathcal RP_{\le \langle t\rangle^{\delta_0}}
P_{\le 2\langle t\rangle^{-\delta_0}}h\big)\Big)\Big\|_{2-\frac{\delta}{100}}\notag\\
\lesssim &\langle t\rangle^{1+\frac{\delta}{100}}\Big\|\lnr^{4+2\delta+\frac{\delta}{99}}P_{\lesssim\langle t\rangle^{\delta_0}}h\Big\|_{(\frac1{2-\frac{\delta}{100}}-2\delta)^{-1}}\>\|h\|_{\frac{1}{\delta}}^2\notag\\
&\quad +\langle t\rangle^{1+\frac{\delta}{100}+\frac{6\delta_0}{N_1}}\|P_{\lesssim1}h\|_{(\frac1{2-\frac{\delta}{100}}-2\delta)^{-1}}\>\|h\|_{\frac{1}{\delta}}^2\notag\\
\lesssim &\langle t\rangle^{1+\frac{\delta}{100}+\frac{6\delta_0}{N_1}-2(1-2\delta)} \|h\|_{X_t}^3\lesssim \|h\|_{X_t}^3.\label{J19_e8}
\end{align}

To estimate \eqref{J19_e6b}, we need a simple fact. Namely, if $\psi=\psi(x)$ is a smooth cut-off function localized to $\{x:|x|\le1\}$, then for any real number $\alpha$,
\begin{align*}
\frac{\partial}{\partial s}\big(\psi(\frac{x}{\langle s \rangle^\alpha})\big)&=\Big[\frac{x}{\langle s \rangle^\alpha}\cdot \nabla\psi(\frac{x}{\langle s \rangle^\alpha} )\Big]\cdot O\Big(\frac1{\langle s\rangle}\Big)\\
&=\chi_{\le\langle s \rangle^\alpha} \cdot O\Big(\frac1{\langle s\rangle}\Big),
\end{align*}
i.e. the function $\partial_s\big(\psi(\frac{x}{\langle s \rangle^\alpha})\big)$ has the same support as $\psi(\frac{x}{\langle s \rangle^\alpha})$ and picks up a decay factor $\frac1{\langle s\rangle}$. Using this fact, we can write
$$
\partial_s(sm_2(\xi,\eta,\sigma,s))=\widetilde{m_2}(\xi,\eta,\sigma,s),
$$
where $\widetilde{m_2}$ has essentially the same form as $m_2$. By essentially repeating the estimate as in \eqref{J19_e6a} (see \eqref{J19_e8}), we have
\begin{align*}
\|\mathcal F^{-1}(\eqref{J19_e6a})\|_{2-\frac{\delta}{100}}\lesssim& \int_0^t
\langle s\rangle^{1+\frac{\delta}{100}+\frac{6\delta_0}{N_1}-2(1-2\delta)}\,ds \|h\|_{X_t}^3\\
\lesssim & \int_0^t \langle s \rangle^{-1-}\,ds  \> \|h\|_{X_t}^3\lesssim  \|h\|_{X_t}^3.
\end{align*}

For \eqref{J19_e6c}, we need to use the third identity in \eqref{J19_e5}.
Note that $|\eta|\le \frac{25}{24}
\langle s \rangle^{-\delta_0}, |\sigma|\le 3\langle s \rangle^{-\delta_0}$, and we can insert a fattened cut-off $P_{\lesssim \langle s \rangle^{-\delta_0}}$ when it is needed. By an estimate similar to that in \eqref{J19_e8}, we have
\begin{align*}
\|\mathcal F^{-1}(\eqref{J19_e6c})
\|_{2-\frac{\delta}{100}}\lesssim& \int_0^t
\langle s\rangle^{1+\frac{\delta}{100}+\frac{6\delta_0}{N_1}}\, \|h\|_{X_t} \,
\left\|\frac{\Delta}{\lnr+1}P_{\lesssim \langle s \rangle^{-\delta_0}}h\right\|_{\frac1\delta}\>\|h\|_{\frac1\delta}\,ds\\
\lesssim &\int_0^t \langle s\rangle^{1+\frac{\delta}{100}+\frac{6\delta_0}{N_1}-2\delta_0}\langle s\rangle^{-2(1-2\delta)}\,ds\>\|h\|_{X_t}^3\\
\lesssim & \int_0^t \langle s \rangle^{-1-}\,ds  \> \|h\|_{X_t}^3\lesssim  \|h\|_{X_t}^3,
\end{align*}
where we need $(2-\frac{6}{N_1})\delta_0>(4+\frac1{100})\delta$.

We turn now to the estimate of \eqref{J19_e6d}. For this we need a lemma.
\begin{lem}\label{lem62}
For any $\beta\ge0$, $2\le p<\frac1\delta$, we have
$$
\Big\|\lnr^\beta e^{it\lnr}\partial_t(\mathcal Rf (t) )\Big\|_p\lesssim
\left\|\lnr^{\beta+1}h (t)\right\|_{(\frac1p-\delta)^{-1}}\|h (t)\|_{\frac1\delta}.
$$
\end{lem}
\begin{proof}[Proof of Lemma \ref{lem62}] By \eqref{p10_J11_e2}, we have
$$
e^{it\lnr}\partial_t(\mathcal Rf(t) )=\lnr \mathcal R\big(\mathcal Rh(t) \>\mathcal Rh(t)\big).
$$
Then the result follows from the product rule.
\end{proof}

Now we continue the estimate of \eqref{J19_e6d}.

 By Lemma \ref{lem62} and a similar computation as in \eqref{J19_e8}, we have
\begin{align*}
\|\mathcal F^{-1}(\eqref{J19_e6d})\|_{2-\frac{\delta}{100}}
\lesssim& \int_0^t \langle s\rangle^{1+\frac{\delta}{100}+\frac{6\delta_0}{N_1}}\|\lnr^{5+2\delta+\frac{\delta}{99}}h\|_{(\frac1{2-\frac{\delta}{100}}-2\delta)^{-1}} \|h\|_{\frac1\delta}^3\,ds\\
\lesssim &\int_0^t \langle s\rangle^{1+\frac{\delta}{100}+\frac{6\delta_0}{N_1}}
\langle s\rangle^{-3(1-2\delta)}\,ds\>\|h\|_{X_t}^4\\
\lesssim & \int_0^t \langle s \rangle^{-1-}\,ds  \> \|h\|_{X_t}^4\lesssim  \|h\|_{X_t}^4.
\end{align*}

In a similar way, we bounded \eqref{J19_e6e} as
\begin{align*}
\|\mathcal F^{-1}(\eqref{J19_e6e})\|_{2-\frac{\delta}{100}}\lesssim& \int_0^t
\langle s\rangle^{1+\frac{\delta}{100}+\frac{6\delta_0}{N_1}}\|h\|_{X_t}\>\|e^{is\lnr}\partial_s(\mathcal Rf)\|_{\frac1{2\delta}}\>\|h\|_{\frac1\delta}\,ds\\
\lesssim &\int_0^t \langle s\rangle^{1+\frac{\delta}{100}+\frac{6\delta_0}{N_1}}
\|\lnr h\|_{\frac1\delta}\>\| h\|_{\frac1\delta}\>\|\lnr h\|_{\frac1\delta}\,ds\>\|h\|_{X_t}\\
\lesssim &\int_0^t \langle s\rangle^{1+\frac{\delta}{100}+\frac{6\delta_0}{N_1}}
\langle s\rangle^{-3(1-2\delta)}\,ds\>\|h\|_{X_t}^4\\
\lesssim & \int_0^t \langle s \rangle^{-1-}\,ds  \> \|h\|_{X_t}^4\lesssim  \|h\|_{X_t}^4.
\end{align*}

\texttt{Subcase 3c}: estimate of $I_3$.

In this subcase, we have $|\eta|\le \frac{25}{24}\langle s \rangle^{-\delta_0}$,
$2\langle s \rangle^{-\delta_0} \le |\sigma| \le \frac {25}{24} \langle s \rangle^{\delta_0}$ on
the support of $m_3(\xi,\eta,\sigma,s)$. Then clearly,
$$
|2\sigma-\eta|\ge \frac12|\sigma|.
$$
By \eqref{J19_e4} and \eqref{rbj_4e2}, we then have
\begin{align}\label{J20_e1}
|\partial_\sigma \phi|=\left|\frac{\sigma-\eta}
{\langle \sigma-\eta\rangle}-\frac{\sigma}
{\langle \sigma\rangle}\right|
 & \gtrsim \frac{|\sigma|} {\langle \sigma \rangle^2} \notag \\
 &\gtrsim \langle s \rangle^{-2\delta_0}.
\end{align}
Using the identity
\begin{align*}
s\> e^{-is \phi} = i\frac{\partial_{\sigma} \phi}{|\partial_{\sigma} \phi|^2} \cdot \partial_{\sigma}
(e^{-is \phi} ),
\end{align*}
we integrate by parts in $\sigma$ in $I_3$.  This gives us
\begin{align}
I_3=&-i\int_0^t \int e^{-is\phi} \frac{\partial_{\xi}
\phi}{\phi_0(\xi,\eta)}\langle \xi
\rangle^{4+2\delta}\langle\eta\rangle \partial_\sigma\cdot\Big(\frac{\partial_{\sigma} \phi}{|\partial_{\sigma} \phi|^2}m_3(\xi,\eta,\sigma,s)\Big)\>
\notag\\
&\qquad\cdot\frac{\eta}{|\eta|} \> \widehat{\mathcal R f}(s,\xi-\eta)
 \widehat{\mathcal Rf}(s,\eta-\sigma) \widehat{\mathcal R
f}(s,\sigma)  d\sigma d\eta ds\label{J19_e20a}\\
&\quad-i\int_0^t \int e^{-is\phi} \frac{\partial_{\xi}
\phi}{\phi_0(\xi,\eta)}\langle \xi
\rangle^{4+2\delta}\langle\eta\rangle\>\frac{\partial_{\sigma} \phi}{|\partial_{\sigma} \phi|^2}m_3(\xi,\eta,\sigma,s)\>
\notag\\
&\qquad\cdot\frac{\eta}{|\eta|} \> \widehat{\mathcal R f}(s,\xi-\eta)
\partial_\sigma\Bigl( \widehat{\mathcal Rf}(s,\eta-\sigma) \widehat{\mathcal R
f}(s,\sigma) \Bigr) d\sigma d\eta ds\label{J19_e20b}.
\end{align}
For \eqref{J19_e20a}, note that
\begin{align*}
&\partial_\sigma\cdot\Big(\frac{\partial_{\sigma} \phi}{|\partial_{\sigma} \phi|^2}\chi_{|\sigma|\le \lsr}
\chi_{|\eta-\sigma|\le \lsr}\chi_{|\sigma|>2\langle s \rangle^{-\delta_0}}\Big)\\
=&\partial_\sigma\cdot\Big(\frac{\partial_{\sigma} \phi}{|\partial_{\sigma} \phi|^2}\Big)\>\chi_{|\sigma|\le \lsr}
\chi_{|\eta-\sigma|\le \lsr}\chi_{|\sigma|>2\langle s \rangle^{-\delta_0}}\\
&+\frac{\partial_{\sigma} \phi}{|\partial_{\sigma} \phi|^2}\>\langle s \rangle^{-\delta_0}\widetilde{\chi}_{|\sigma|\sim \lsr}
\chi_{|\eta-\sigma|\le \lsr}\chi_{|\sigma|>2\langle s \rangle^{-\delta_0}}\\
&+\frac{\partial_{\sigma} \phi}{|\partial_{\sigma} \phi|^2}\>\chi_{|\sigma|\le \lsr}
\langle s \rangle^{-\delta_0}\widetilde{\chi}_{|\eta-\sigma|\sim \lsr}\chi_{|\sigma|>2\langle s \rangle^{-\delta_0}}\\
&+\frac{\partial_{\sigma} \phi}{|\partial_{\sigma} \phi|^2}\>\chi_{|\sigma|\le \lsr}
{\chi}_{|\eta-\sigma|\le \lsr}\langle s \rangle^{\delta_0}\widetilde
{\chi}_{|\sigma|\sim2\langle s \rangle^{-\delta_0}},
\end{align*}
where $\tilde \chi$ are some modified cut-offs.

By \eqref{J20_e1}, it is easy to check that the functions
\begin{align*}
\widetilde{m_1}(\eta,\sigma)&=\chi_{|\partial_\sigma\phi|\gtrsim \langle s \rangle^{-2\delta_0}}\partial_\sigma\cdot\Big(\frac{\partial_{\sigma} \phi}{|\partial_{\sigma} \phi|^2}\Big)\>\langle s \rangle^{-10\delta_0}\>\langle\eta-\sigma\rangle^{-(1+\frac{\delta}{400})}\langle\sigma\rangle^{-(1+\frac{\delta}{400})},\\
\widetilde{m_2}(\eta, \sigma)&=\chi_{|\partial_\sigma\phi|\gtrsim \langle s \rangle^{-2\delta_0}}\frac{\partial_{\sigma} \phi}{|\partial_{\sigma} \phi|^2}\>\langle s \rangle^{\delta_0}\langle s \rangle^{-10\delta_0}\>\langle\eta-\sigma\rangle^{-(1+\frac{\delta}{400})}\langle\sigma\rangle^{-(1+\frac{\delta}{400})}
\end{align*}
satisfy \eqref{rbj_1e1}. Therefore by Lemma \ref{lem_rbj_1}, we have
\begin{align*}
\|\mathcal F^{-1}(\eqref{J19_e20a})\|_{2-\frac{\delta}{100}}\lesssim& \int_0^t \langle s \rangle^{\frac{\delta}{100}}\Big\|\lnr^{5+2\delta+\frac{\delta}{100}}\mathcal RP_{\le\langle s\rangle^{\delta_0}}h\Big\|_{(\frac1{2-\frac{\delta}{100}}-2\delta)^{-1}}\\
&\quad\langle s \rangle^{10\delta_0}
\sum\limits_{i=1}^2\Big\|T_{\widetilde{m_i}(\eta,\sigma)}\Big(\lnr^{1+\frac{\delta}{400}}\mathcal RP_{\lesssim\langle s\rangle^{\delta_0}}h,\\
&\qquad\lnr^{1+\frac{\delta}{400}}\mathcal RP_{\lesssim\langle s\rangle^{\delta_0}}P_{\gtrsim \langle s\rangle^{-\delta_0}}h\Big)\Big\|_{\frac{1}{2\delta}}\,ds\\
\lesssim&\int_0^t \langle s \rangle^{\frac{\delta}{100}}\|h\|_{X_t}
\> \langle s \rangle^{10\delta_0+\delta_0\frac{\delta}{200}}
\|\lnr h\|_{\frac{1}{\delta}}^2\,ds\\
\lesssim&\int_0^t \langle s \rangle^{\frac{\delta}{100}+(10+\frac{\delta}{200})\delta_0-2(1-2\delta)}\,ds\>\|h\|_{X_t}^3\\
\lesssim&\int_0^t \langle s \rangle^{-1-}\,ds\>\|h\|_{X_t}^3\lesssim \|h\|_{X_t}^3.
\end{align*}
Similarly for \eqref{J19_e20b}, we use Lemma \ref{lem-J19-1} to obtain
\begin{align*}
\|\mathcal F^{-1}(\eqref{J19_e20b})\|_{2-\frac{\delta}{100}}\lesssim& \int_0^t \langle s \rangle^{\frac{\delta}{100}}\Big\|\lnr^{5+2\delta+\frac{\delta}{100}}\mathcal RP_{\lesssim\langle s\rangle^{\delta_0}}h\Big\|_{\frac{1}{\delta}}\\
&\quad\langle s \rangle^{8\delta_0}
\|\lnr^{1+\frac{\delta}{400}}P_{\lesssim\langle s\rangle^{\delta_0}}e^{is\lnr}\mathcal F^{-1}
\big(\partial_\sigma( \widehat{\mathcal Rf })\big)\|_{(\frac1{2-\frac{\delta}{100}}-2\delta)^{-1}}\\
&\qquad\|\lnr^{1+\frac{\delta}{400}}\mathcal RP_{\lesssim\langle s\rangle^{\delta_0}}h\|_{\frac{1}{\delta}}\,ds\\
\lesssim&\int_0^t \langle s \rangle^{\frac{\delta}{100}+(4+2\delta+\frac{\delta}{100})\delta_0}\|\lnr h\|_{\frac{1}{\delta}}\>\langle s \rangle^{8\delta_0} \langle s \rangle^{\delta_0(1+\frac{\delta}{400}+\frac{2}{2+\delta}-4(\frac{1}{2-\frac{\delta}{100}}-2\delta))}
\\
&\quad\langle s \rangle^{1-2(\frac{1}{2-\frac{\delta}{100}}-2\delta)}\|\langle x\rangle f\|_{2+\delta}\>\langle s \rangle^{\frac{\delta}{400}\delta_0}
\|\lnr h\|_{\frac{1}{\delta}}\,ds\\
\lesssim&\int_0^t \langle s \rangle^{-1-}\,ds\>\|h\|_{X_t}^3\lesssim \|h\|_{X_t}^3.
\end{align*}
This ends the estimate of $I_3$.

\texttt{Subcase 3d}: estimate of $I_4$.

Note that in this subcase, $|\eta|\gtrsim \langle s \rangle^{-\delta_0}$. By Lemma \ref{lem_phase},
we have
$$
\partial_\xi \phi=Q_1(\xi,\eta,\sigma)\partial_\eta \phi+Q_2(\xi,\eta,\sigma)\partial_\sigma \phi,
$$
where
\begin{equation*}
    \big|\partial_{\xi}^\alpha\partial_\eta^\beta\partial_\sigma^\gamma
    Q_i(\xi,\eta,\sigma)\big|\lesssim_{\alpha,\beta,\gamma} \langle|\xi|+|\eta|+\sigma\rangle^3, \quad i=1,2.
\end{equation*}
Obviously,
\begin{align*}
s \>\partial_{\xi} \phi \>e^{-is \phi} = i\Bigl(Q_1\partial_\eta ( e^{-is \phi})+Q_2\partial_\sigma ( e^{-is \phi})
\Bigr).
\end{align*}
Using the above identity,
we shall integrate by parts in $\eta$ and $\sigma$. It is not difficult to check that the functions
\begin{align*}
\widetilde{m_i}(\xi,\eta,\sigma)&=\frac{\partial_\xi \phi\>\langle\xi\rangle^{4+2\delta}}{\phi_0(\xi,\eta)}\>\langle\eta\rangle\>\frac{\eta}{|\eta|}m_4(\xi,\eta,\sigma,s)\>Q_i(\xi,\eta,\sigma)\>\langle s \rangle^{-(13+2\delta)\delta_0},\quad i=1,2;\\
\widetilde{m_3}(\xi,\eta,\sigma)&=
\partial_\eta \widetilde{m_i}(\xi,\eta,\sigma,s)\>\langle s \rangle^{-(14+2\delta)\delta_0} ,\quad i=1,2;\\
\widetilde{m_4}(\xi,\eta,\sigma)&=\partial_\sigma \widetilde{m_i}(\xi,\eta,\sigma,s)\>\langle s \rangle^{-(13+2\delta)\delta_0} ,\quad i=1,2
\end{align*}
satisfy \eqref{cond_three}. Therefore by Corollary \ref{cor_rbj_1}, we have
\begin{align*}
\|\mathcal F^{-1}(I_4)\|_{2-\frac{\delta}{100}}\lesssim&
\int_0^t \langle s \rangle^{\frac{\delta}{100}+(14+2\delta)\delta_0}
\Big\|T_{\widetilde{m_3}+\widetilde{m_4}}(\mathcal Rh, \mathcal Rh,\mathcal Rh)\Big\|_{2-\frac{\delta}{100}}\,ds\\
&+\int_0^t \langle s \rangle^{\frac{\delta}{100}+(13+2\delta)\delta_0}
\Big\|T_{\widetilde{m_1}}\left(P_{\lesssim\langle s\rangle^{\delta_0}}
e^{is\lnr}\mathcal F^{-1}\big(\partial_\eta( \widehat{\mathcal Rf} )\big),\mathcal Rh, \mathcal Rh\right)\Big\|_{2-\frac{\delta}{100}}\,ds\\
&+\int_0^t \langle s \rangle^{\frac{\delta}{100}+(13+2\delta)\delta_0}\Big\|
T_{\widetilde{m_2}}\left(\mathcal Rh, \mathcal Rh,P_{\lesssim\langle
s\rangle^{\delta_0}}e^{is\lnr}\mathcal F^{-1}\big(\partial_\sigma( \widehat{\mathcal Rf })\big)\right)\Big\|_{2-\frac{\delta}{100}}\,ds\\
\lesssim&\int_0^t \langle s \rangle^{\frac{\delta}{100}+(14+2\delta)\delta_0-2(1-2\delta)}\,ds\>\|h\|_{X_t}^3\\
\lesssim&\int_0^t \langle s \rangle^{-1-}\,ds\>\|h\|_{X_t}^3\lesssim \|h\|_{X_t}^3.
\end{align*}
Hence Case 3 is finished.

\texttt{Case 4}:
\begin{align}
\phi(\xi,\eta,\sigma) =
 \langle \xi \rangle+
\langle \xi -\eta \rangle + \langle \eta -\sigma \rangle -\langle \sigma \rangle. \label{phase_plus1}
\end{align}

In this case we decompose (see \eqref{J19_e1aa}),
\begin{align*}
m_{\text{low}}(\xi,\eta,\sigma)&=m_{\text{low}}(\xi,\eta,\sigma)\chi_{|\eta|\le \langle s \rangle^{-\delta_0}}+m_{\text{low}}(\xi,\eta,\sigma)\chi_{|\eta|> \langle s \rangle^{-\delta_0}}\\
&=m_{\text{low}}^{(1)}(\xi,\eta,\sigma)+m_{\text{low}}^{(2)}(\xi,\eta,\sigma),
\end{align*}
and denote the corresponding integral in \eqref{J19_e1} as $I_1$ and $I_2$ respectively.

\texttt{Subcase 4a}: estimate of $I_1$.

We again use the partial normal form trick. Note that
$$
\langle \sigma -\eta\rangle-\langle\sigma \rangle=\frac{(2\sigma-\eta)\cdot(-\eta)}{\langle \sigma -\eta\rangle+\langle\sigma \rangle}.
$$
Using the identity
$$
e^{-is(\langle \xi \rangle+\langle
\xi-\eta\rangle)}=\frac{i}{\langle \xi \rangle+\langle
\xi-\eta\rangle}\partial_s(e^{-is(\langle \xi \rangle+\langle
\xi-\eta\rangle)})
$$
and integrating by parts in the time variable $s$, we get
\begin{align}
I_1=&\int e^{-it\phi} \frac{i t\> \partial_{\xi}
\phi}{\phi_0(\xi,\eta)}\frac{\langle \xi
\rangle^{4+2\delta}}{\langle \xi \rangle+\langle
\xi-\eta\rangle}\langle\eta\rangle m_{\text{low}}^{(1)}\>
\notag\notag\\
&\quad\frac{\eta}{|\eta|} \> \widehat{\mathcal R f}(t,\xi-\eta)
 \widehat{\mathcal Rf}(t,\eta-\sigma) \widehat{\mathcal R
f}(t,\sigma)  d\sigma d\eta \label{J22_e10a}\\
&- \int_0^t \int e^{-is\phi}\frac{(2\sigma-\eta)\cdot(-\eta)}
{\langle \sigma -\eta\rangle+\langle\sigma \rangle}\frac{\langle \xi
\rangle^{4+2\delta}}{\langle \xi \rangle+\langle
\xi-\eta\rangle}\langle\eta\rangle
\frac{s\partial_{\xi} \phi}{\phi_0(\xi,\eta)}\frac{\eta}{|\eta|} m_{\text{low}}^{(1)}\>
\notag\\
&\quad\> \widehat{\mathcal R f}(s,\xi-\eta)
 \widehat{\mathcal Rf}(s,\eta-\sigma) \widehat{\mathcal R
f}(s,\sigma)  d\sigma d\eta ds\label{J22_e10b}\\
&- \int_0^t \int e^{-is\phi}\frac{\langle \xi
\rangle^{4+2\delta}}{\langle \xi \rangle+\langle
\xi-\eta\rangle}\langle\eta\rangle
\frac{i\partial_{\xi} \phi}{\phi_0(\xi,\eta)}\frac{\eta}{|\eta|}
\partial_s(sm_{\text{low}}^{(1)})\>
\notag\\
&\quad \widehat{\mathcal R f}(s,\xi-\eta)
 \widehat{\mathcal Rf}(s,\eta-\sigma) \widehat{\mathcal R
f}(s, \sigma)  d\sigma d\eta ds\label{J22_e10c}\\
&- \int_0^t \int e^{-is\phi}\frac{\langle \xi
\rangle^{4+2\delta}}{\langle \xi \rangle+\langle
\xi-\eta\rangle}\langle\eta\rangle
\frac{is\partial_{\xi} \phi}{\phi_0(\xi,\eta)}\frac{\eta}{|\eta|} m_{\text{low}}^{(1)}\>
\notag\\
&\quad\partial_s\Big(\widehat{\mathcal R f}(s,\xi-\eta)\>
 \widehat{\mathcal Rf}(s,\eta-\sigma) \widehat{\mathcal R
f}(s, \sigma)\Big)  d\sigma d\eta ds. \label{J22_e10d}
\end{align}

The estimate of \eqref{J22_e10a} is similar to \eqref{J19_e6a}, and we have
$$
\left\|\mathcal F^{-1}(\eqref{J22_e10a})\right\|_{2-\frac{\delta}{100}}\lesssim \|h\|_{X_t}^3.
$$

For \eqref{J22_e10b}, note that $\frac{(2\sigma-\eta)}{\langle \sigma -
\eta\rangle+\langle\sigma \rangle}$ is a Coifman-Meyer multiplier. We compute
\begin{align*}
\left\|\mathcal F^{-1}(\eqref{J22_e10b})\right\|_{2-\frac{\delta}{100}}&\lesssim
\int_0^t \langle s \rangle^{1+\frac{\delta}{100}}\big
\|\lnr^{4+3\delta} P_{\le \langle s \rangle^{\delta_0}} \mathcal Rh\big\|_{(\frac1{2-\frac{\delta}{100}}
-2\delta)^{-1}}\\
&\quad
\left\|\nabla P_{\lesssim \langle s \rangle^{-\delta_0}}
T_{\frac{(2\sigma-\eta)}{\langle \sigma -\eta\rangle+
\langle\sigma \rangle}}(P_{\le \langle s \rangle^{\delta_0}}
\mathcal Rh, P_{\le \langle s \rangle^{\delta_0}}\mathcal Rh)\right\|_{\frac{1}{2\delta}}^2 ds\\
&\lesssim \int_0^t \langle s \rangle^{1+\frac{\delta}{100}}\|h\|_{H^{N^{\prime}}}\>\langle s
 \rangle^{-\delta_0}\|h\|_{\frac{1}{\delta}}^2\,ds\\
&\lesssim \int_0^t \langle s \rangle^{1+\frac{\delta}{100}-\delta_0-2(1-2\delta)}\,ds\>
\|h\|_{X_t}^3\\
&\lesssim \int_0^t \langle s \rangle^{-1-}\,ds\>\|h\|_{X_t}^3\lesssim\|h\|_{X_t}^3.
\end{align*}

The estimate of \eqref{J22_e10c} is similar to \eqref{J19_e6b}, and we have
\begin{align*}
\left\|\mathcal F^{-1}(\eqref{J22_e10c})\right\|_{2-\frac{\delta}{100}}&\lesssim
\|h\|_{X_t}^3.
\end{align*}

The estimate of \eqref{J22_e10d} is also similar to that of \eqref{J19_e6d} and \eqref{J19_e6e}.  We have
\begin{align*}
\left\|\mathcal F^{-1}(\eqref{J22_e10d})\right\|_{2-\frac{\delta}{100}}&\lesssim
\|h\|_{X_t}^3.
\end{align*}

\texttt{Subcase 4b}: estimate of $I_2$.

It is not difficult to check that
\begin{equation}\label{J22_e10ee}
\langle \xi \rangle+
\langle \xi -\eta \rangle + \langle \eta -\sigma \rangle -\langle \sigma \rangle
\gtrsim \frac{1}{\langle\xi\rangle}, \quad \forall\; \xi,\eta,\sigma\in \R^2.
\end{equation}
Using the identity
$$
e^{-is\phi}=\frac{i}{\phi}\partial_s(e^{-is\phi}),
$$
we integrate by parts in the variable $s$. This gives
\begin{align}
I_2=&\int e^{-it\phi}\frac{i}{\phi} \frac{ t\> \partial_{\xi}
\phi}{\phi_0(\xi,\eta)}{\langle \xi
\rangle^{4+2\delta}}\langle\eta\rangle\frac{\eta}{|\eta|} \> m_{\text{low}}^{(2)}\>
\notag\\
&\quad \widehat{\mathcal R f}(t,\xi-\eta)
 \widehat{\mathcal Rf}(t,\eta-\sigma) \widehat{\mathcal R
f}(t,\sigma)  d\sigma d\eta \label{J22_e11a}\\
&- \int_0^t \int e^{-is\phi}\frac{i}{\phi} \frac{ \partial_{\xi}
\phi}{\phi_0(\xi,\eta)}{\langle \xi
\rangle^{4+2\delta}}\langle\eta\rangle\frac{\eta}{|\eta|}
\partial_s(sm_{\text{low}}^{(2)})\>
\notag\\
&\quad \widehat{\mathcal R f}(s,\xi-\eta)
 \widehat{\mathcal Rf}(s,\eta-\sigma) \widehat{\mathcal R
f}(s,\sigma)  d\sigma d\eta ds\label{J22_e11b}\\
&- \int_0^t \int e^{-is\phi}\frac{i}{\phi} \frac{ s\> \partial_{\xi}
\phi}{\phi_0(\xi,\eta)}{\langle \xi
\rangle^{4+2\delta}}\langle\eta\rangle\frac{\eta}{|\eta|}
 m_{\text{low}}^{(2)}\>
\notag\\
&\quad\partial_s\Big(\widehat{\mathcal R f}(s,\xi-\eta)\>
 \widehat{\mathcal Rf}(s,\eta-\sigma) \widehat{\mathcal R
f}(s, \sigma)\Big)  d\sigma d\eta ds. \label{J22_e11c}
\end{align}

For \eqref{J22_e11a}, by using \eqref{J22_e10ee} and Lemma \ref{lem_rbj_5}, it is not difficult to check that the function
\begin{align*}
\widetilde{m}(\xi,\eta,\sigma,s)=&\frac{i}{\phi} \frac{ \partial_{\xi}
\phi}{\phi_0(\xi,\eta)}{\langle \xi
\rangle^{4+2\delta}}\langle\eta\rangle\frac{\eta}{|\eta|}\\
&\quad\chi_{|\xi-\eta|\le \lsr}
\chi_{|\eta-\sigma|\le \lsr}\chi_{|\sigma|\le \lsr}\chi_{|\eta|>\langle s \rangle^{-\delta_0}}\langle s\rangle^{-(14+3\delta)\delta_0}
\end{align*}
satisfies \eqref{cond_three}. Therefore by Corollary \ref{cor_rbj_1}, we have
\begin{align*}
\|\mathcal F^{-1}(\eqref{J22_e11a})\|_{2-\frac{\delta}{100}}
\lesssim &\langle t\rangle^{1+\frac{\delta}{100}+(14+3\delta)\delta_0}\|h(t)\|_{(\frac1{2-\frac{\delta}{100}}-2\delta)^{-1}}\>\|h(t)\|_{\frac{1}{\delta}}^2\notag\\
\lesssim & \langle t\rangle^{1+\frac{\delta}{100}+(14+3\delta)\delta_0-2(1-2\delta)} \|h\|_{X_t}^3\lesssim \|h\|_{X_t}^3.
\end{align*}
Similarly,
\begin{align*}
\left\|\mathcal F^{-1}(\eqref{J22_e11b})\right\|_{2-\frac{\delta}{100}}&\lesssim
\int_0^t \langle s\rangle^{1+\frac{\delta}{100}+(14+3\delta)\delta_0-2(1-2\delta)} ds \|h\|_{X_t}^3\\
&\lesssim
\int_0^t \langle s\rangle^{-1-}\,ds \|h\|_{X_t}^3
\lesssim\|h\|_{X_t}^3.
\end{align*}
In a similar way, using Lemma \ref{lem62}, we have
\begin{align*}
\left\|\mathcal F^{-1}(\eqref{J22_e11c})\right\|_{2-\frac{\delta}{100}}&\lesssim
\int_0^t \langle s\rangle^{1+\frac{\delta}{100}+(14+3\delta)\delta_0} \|e^{is\lnr}\partial_s\big(\mathcal Rf\big)\|_{(\frac1{2-\frac{\delta}{100}}-2\delta)^{-1}}\|h\|_{\frac{1}{\delta}}^2\,ds\\
&\lesssim
\int_0^t \langle s\rangle^{1+\frac{\delta}{100}+(14+3\delta)\delta_0} \|h\|_{\frac{1}{\delta}}^3\|h\|_{H^3}\,ds\\
&\lesssim
\int_0^t \langle s\rangle^{1+\frac{\delta}{100}+(14+3\delta)\delta_0-3(1-2\delta)} \,ds\|h\|_{X_t}^4\\
&\lesssim
\int_0^t \langle s\rangle^{-1-}\,ds \|h\|_{X_t}^4
\lesssim\|h\|_{X_t}^4.
\end{align*}

\noindent
\texttt{Case 5:}
\begin{align*}
\phi(\xi,\eta,\sigma) = \langle \xi \rangle + \langle \xi-\eta \rangle
-\langle \eta-\sigma \rangle + \langle \sigma \rangle.
\end{align*}
This is exactly the same as Case 4 after the change of variable $\sigma\to \eta-\sigma$.

\noindent
\texttt{Case 6:}
\begin{align}
\phi(\xi,\eta,\sigma) = \langle \xi \rangle -\langle \xi-\eta\rangle
+ \langle \eta-\sigma \rangle + \langle \sigma \rangle.
\end{align}
In this case we decompose (see \eqref{J19_e1aa}),
\begin{align*}
m_{\text{low}}(\xi,\eta,\sigma)&=m_{\text{low}}(\xi,\eta,\sigma)\chi_{|\eta|\le \langle s \rangle^{-\delta_0}}+m_{\text{low}}(\xi,\eta,\sigma)\chi_{|\eta|> \langle s \rangle^{-\delta_0}}\\
&=m_{\text{low}}^{(1)}(\xi,\eta,\sigma)+m_{\text{low}}^{(2)}(\xi,\eta,\sigma),
\end{align*}
and denote the corresponding integral in \eqref{J19_e1} as $I_1$ and $I_2$ respectively. The estimate of $I_2$ is exactly the same as Subcase 4b. Hence we only need to estimate $I_1$. In this situation, note that
$$
\partial_\xi \phi=\frac{\xi}{\langle \xi\rangle}-\frac{\xi-\eta}{\langle \xi-\eta\rangle}=Q(\xi,\eta)\eta,
$$
where
$$
\big|\partial_\xi^\alpha\partial_\eta^{\beta} Q(\xi,\eta)\big|\lesssim_{\alpha,\beta}1.
$$
Therefore,
\begin{align*}
\left\|\mathcal F^{-1}(I_1)\right\|_{2-\frac{\delta}{100}}&\lesssim
\int_0^t \langle s\rangle^{1+\frac{\delta}{100}} \|\lnr^{5+3\delta}h\|_{(\frac1{2-\frac{\delta}{100}}-2\delta)^{-1}}\\
&\quad \left\|P_{\le \langle s \rangle^{-\delta_0}}\nabla\big(P_{\le \lsr}
\mathcal R(P_{\le \lsr}\mathcal R h \cdot P_{\le \lsr}\mathcal Rh)\big)\right\|_{\frac{1}{2\delta}} ds\\
&\lesssim
\int_0^t \langle s\rangle^{1+\frac{\delta}{100}}\|h\|_{H^{N^{\prime}}}\>\langle s \rangle^{-\delta_0}\|h\|_{\frac{1}{\delta}}^2\,ds\\
&\lesssim
\int_0^t \langle s\rangle^{1+\frac{\delta}{100}+\delta-\delta_0-2(1-2\delta)} \,ds\|h\|_{X_t}^3\\
&\lesssim
\int_0^t \langle s\rangle^{-1-}\,ds \|h\|_{X_t}^3
\lesssim\|h\|_{X_t}^3.
\end{align*}
This settles Case 6.

\texttt{Case 7}:
\begin{align}
\phi(\xi,\eta,\sigma) = \langle \xi \rangle -\langle \xi-\eta\rangle
-\langle \eta-\sigma \rangle -\langle \sigma \rangle.
\end{align}
In this case we again decompose
\begin{align*}
m_{\text{low}}(\xi,\eta,\sigma)&=m_{\text{low}}(\xi,\eta,\sigma)\chi_{|\eta|\le \langle s \rangle^{-\delta_0}}+m_{\text{low}}(\xi,\eta,\sigma)\chi_{|\eta|> \langle s \rangle^{-\delta_0}}\\
&=m_{\text{low}}^{(1)}(\xi,\eta,\sigma)+m_{\text{low}}^{(2)}(\xi,\eta,\sigma),
\end{align*}
and denote the corresponding integral in \eqref{J19_e1} as $I_1$ and $I_2$ respectively.  Note that
$$
|\phi(\xi,\eta,\sigma)|\gtrsim \frac{1}{\langle \xi\rangle}
$$
and
$$
\partial_\xi \phi=\frac{\xi}{\langle \xi\rangle}-\frac{\xi-\eta}{\langle \xi-\eta\rangle}.
$$
The estimates of $I_1$ and $I_2$ are exactly the same as in Case 6. Hence Case 7 is settled.

\texttt{Case 8}:
\begin{align}
\phi(\xi,\eta,\sigma) = \langle \xi \rangle +\langle \xi-\eta\rangle
+\langle \eta-\sigma \rangle + \langle \sigma \rangle.
\end{align}
In this case we again decompose (see \eqref{J19_e1aa})
\begin{align*}
m_{\text{low}}(\xi,\eta,\sigma)=&m_{\text{low}}(\xi,\eta,\sigma)\chi_{|\eta|\le \langle s \rangle^{-\delta_0}}\chi_{|\sigma|>2 \langle s \rangle^{-\delta_0}}\\
&+m_{\text{low}}(\xi,\eta,\sigma)\chi_{|\eta|\le \langle s \rangle^{-\delta_0}}\chi_{|\sigma|\le2 \langle s \rangle^{-\delta_0}}\\
&+m_{\text{low}}(\xi,\eta,\sigma)\chi_{|\eta|> \langle s \rangle^{-\delta_0}}.
\end{align*}
and denote the corresponding integral in \eqref{J19_e1} as $I_1$, $I_2$ and $I_3$ respectively.   We discuss three subcases.

\texttt{Subcase 8a}: estimate of $I_1$.
This subcase is exactly the same as Case 3c which was estimated before. Therefore,
$$
\left\|\mathcal F^{-1}(I_1)\right\|_{2-\frac{\delta}{100}}\lesssim \|h\|_{X_t}^3+\|h\|_{X_t}^4.
$$

\texttt{Subcase 8b}: estimate of $I_2$.
In this subcase, we shall again use the partial normal form trick. Write
$$
e^{-is\phi}=
\frac{i}{\langle \xi \rangle+\langle
\xi-\eta\rangle+2}\partial_s(e^{-is(\langle \xi \rangle+\langle
\xi-\eta\rangle+2)})e^{-is(\langle \eta-\sigma \rangle +\langle
\sigma\rangle-2)}.
$$

Note that by \eqref{J19_e5},
$$
\langle \eta-\sigma \rangle +\langle
\sigma\rangle-2=\frac{(\eta-\sigma)\cdot(\eta-\sigma)}{\langle \eta-\sigma \rangle+1}
+\frac{\sigma\cdot\sigma}{\langle \sigma \rangle+1}.
$$
Integrating by parts in $s$, we arrive at essentially the same situation as in Case 3b which was estimated before. Hence we have
$$
\left\|\mathcal F^{-1}(I_2)\right\|_{2-\frac{\delta}{100}}\lesssim \|h\|_{X_t}^3+\|h\|_{X_t}^4.
$$

\texttt{Subcase 8c}: estimate of $I_3$.

In this subcase we note that $|
\eta|\gtrsim \langle s \rangle^{-\delta_0}$ and
$$
\phi(\xi,\eta,\sigma)\gtrsim 1.
$$
We can integrate by parts in the time variable $s$ and use the same estimates as in Case 4b. Hence
$$
\left\|\mathcal F^{-1}(I_3)\right\|_{2-\frac{\delta}{100}}\lesssim \|h\|_{X_t}^3+\|h\|_{X_t}^4.
$$
We have completed the estimates of all phases. The proposition is now proved.

\section{Proof of Theorem \ref{thm_main}}
In this section we complete the proof of Theorem \ref{thm_main}. Define
\begin{align*}
a(t)=&\|\langle \tau\rangle^{-\delta}h(\tau)\|_{C_{\tau}^0 H^N([0,t])}+\|h(\tau)\|_{C_{\tau}^0 H^{N'}([0,t])}\\
&\quad+\|\langle \tau\rangle |\nabla|^\delta\langle\nabla\rangle h(\tau)\|_{L^\infty_\tau L^\infty_x([0,t])}+\|\langle \tau\rangle^{1-2\delta} \lnr h(\tau)\|_{L^\infty_\tau L^\frac{1}{\delta}_x([0,t])}\\
&\quad+\|x(1-\Delta)e^{-i\tau\lnr}h(\tau)\|_{L^\infty_\tau L^{2+\delta}_x([0,t])}.
\end{align*}
By the local theory in Section \ref{sec_local}, we have $a(t)$ is a continuous function
 of $t$. Also from the energy estimates therein, we have
\begin{align*}
\frac{d}{d\tau}\big(\|h(\tau)\|_{H^N}\big)&\lesssim
\big(\|u(\tau)\|_\infty+\|\nabla u(\tau)\|_\infty+\|\nabla \v v (\tau) \|_\infty\big)\>
\|h(\tau) \|_{H^N}\\
&\lesssim \||\nabla|^\delta\lnr h(\tau)\|_\infty \|h (\tau) \|_{H^N} \notag \\
& \lesssim a(\tau)^2 \langle \tau \rangle^{-1+\delta}.
\end{align*}
Integrating in time and using the monotonicity of $a(\tau)$ gives us
$$
\|h(s)\|_{H^N}\lesssim \|h_0\|_{H^N}+a(s)^2\langle s\rangle^\delta,
$$
or
$$
\|\langle \tau\rangle^{-\delta}h(\tau)\|_{C_{\tau}^0 H^N([0,t])}\lesssim  \| e^{i\tau \jnab} h_0 \|_{X_{\infty}}+a(t)^2.
$$
By the analysis in Section 4-7, we also have
\begin{align*}
\| \langle \tau \rangle |\nabla|^{\frac 12}  \langle \nabla \rangle h (\tau) \|_{L_{\tau}^\infty L_x^\infty([0,t])}
&+ \|  h(\tau) \|_{C_\tau^0 H_x^{N'}([0,t])}\\
+ \| \langle \tau \rangle^{1-2 \delta} h(\tau) \|_{L_\tau^\infty L_x^{\frac 1 {\delta}}([0,t])}
& + \| \langle x  \rangle  (1-\Delta) e^{-i\tau\lnr} h(\tau) \|_{L_\tau^\infty L_x^{2+\delta}([0,t])}\\
&\lesssim  \| e^{i\tau \jnab} h_0 \|_{X_{\infty}} +a(t)^2+a(t)^3+a(t)^4.
\end{align*}
Therefore we have proved for some constant $C>0$,
$$
a(t)\le C \cdot \bigl( \| e^{i\tau \jnab} h_0 \|_{X_{\infty}} +a(t)^2+a(t)^3+a(t)^4 \bigr).
$$
Since $a(t)$ is a continuous function of $t$ and $a(0) \le
\| e^{i\tau \jnab} h_0 \|_{X_{\infty}}$, by a standard argument,
we conclude that if $\| e^{i\tau \jnab} h_0 \|_{X_{\infty}} $ is sufficiently small, then $
a(t)$ is bounded for all $t\ge 0$. Note that the scattering of $H^{N'}$ norm is a simple consequence
of the analysis in Section 4. This concludes the proof of Theorem \ref{thm_main}.


\begin{thebibliography}{m}

\bibitem{C86}
Christodoulou, D.
Global solutions for nonlinear hyperbolic equations for small data.
Comm. Pure Appl. Math. 39 (1986), 267--282.

\bibitem{CT08}
Chae, D.; Tadmor, E.
On the finite time blow-up of the Euler-Poisson equations in $\mathbb R^n$.
Commun. Math. Sci. 6 (2008), no. 3, 785--789.


\bibitem{GMS1}
Germain, P., Masmoudi, N. and Shatah, J. Global solutions for 3D quadratic Scrh\"odinger
equations. Int. Math. Res. Not. 2009, no. 3, 414--432.

\bibitem{GMS2}
Germain, P., Masmoudi, N. and Shatah, Global solutions for the gravity water waves equation
in dimension 3. Ann. Math. 175, No.2, 691--754 (2012).

\bibitem{GMS3}
Germain, P., Masmoudi, N. Global existence for the Euler-Maxwell system. Preprint.


\bibitem{GNT06}
Gustafson, S., Nakanishi, K. and Tsai, T-P.
Global dispersive solutions for the Gross-Pitaevskii equation in two and three dimensions.
Ann. Henri Poincar\'e 8 (2007), no. 7, 1303--1331.


\bibitem{Godin05}
Godin, P.
The lifespan of a class of smooth spherically symmetric solutions
of the compressible Euler equations with variable entropy in three space dimensions.
Arch. Ration. Mech. Anal. 177 (2005), no. 3, 479--511.


\bibitem{Guo98}
Guo, Y.
Smooth irrotational flows in the large to the Euler-Poisson system in $\bold R^{3+1}$.
Comm. Math. Phys. 195 (1998), no. 2, 249--265.

\bibitem{GP11}
Guo Y. and Pausader B.: Global Smooth Ion Dynamics in the Euler-Poisson System, Comm.
Math. Phys. 303 (2011), 89-125.

\bibitem{GuoTZ99}
Guo, Y. and Tahvildar-Zadeh A.S.
Formation of singularities in relativistic fluid dynamics and in spherically
symmetric plasma dynamics.
Contemp. Math. 238, p151--161 (1999).


\bibitem{IP11}
Ionescu, A. and Pausader, B.
The Euler--Poisson system in 2D: global stability of the constant equilibrium solution.
http://arxiv.org/abs/1110.0798.



\bibitem{jlz}
Jang J., Li D. and Zhang X.
Smooth global solutions for the two dimensional Euler-Poisson system,
to appear in Forum Mathematicum.

\bibitem{Jang_radial}
Jang J. The 2D Euler-Poisson System with Spherical Symmetry. arXiv:1109.2643.


\bibitem{K85}
Klainerman, S.
Global existence of small amplitude solutions to nonlinear Klein-Gordon equations
in four space-time dimensions. Comm. Pure. Appl. Math. 38, 631--641 (1985).


\bibitem{K86}
Klainerman, S.
The null condition and global existence to nonlinear wave equations.
Lect. Appl. Math. 23 (1986), 293--326.


\bibitem{LZprep}
Li, D. and Zhang, X. Wave operators for nonlinear wave equations
with null structure. To appear in Comm. Contemp. Math.


\bibitem{Liprep}
Li, D. Sharp decay of solutions to the Gross-Pitaevskii equation. Preprint.

\bibitem{Oz96}
Ozawa, T.; Tsutaya, K. and Tsutsumi, Y.
Global existence and asymptotic behavior of solutions for the Klein-Gordon
equations with quadratic nonlinearity in two space dimensions.
Math. Z. 222 (1996), no. 3, 341--362.


\bibitem{Ra89}
Rammaha, M. A.
Formation of singularities in compressible fluids in two-space dimensions.
Proc. Amer. Math. Soc. 107 (1989), no. 3, 705--714.

\bibitem{Sh85}
Shatah, J.
Normal forms and quadratic nonlinear Klein-Gordon equations.
Comm. Pure Appl. Math. 38 (1985), no. 5, 685--696.

\bibitem{Si85}
Sideris, T. C.
Formation of singularities in three-dimensional compressible fluids.
Comm. Math. Phys. 101 (1985), no. 4, 475--485.

\bibitem{Si91}
Sideris, T. C.
The lifespan of smooth solutions to the three-dimensional compressible Euler equations and the incompressible limit.
Indiana Univ. Math. J. 40 (1991), no. 2, 535--550.

\bibitem{ST93}
Simon J.C.H. and Taflin E. The Cauchy problem for nonlinear Klein--Gordon equations,
Comm. Math. Phys. 152 (1993) 433--478.


\end{thebibliography}
\end{document}